\documentclass[11pt]{article}
\usepackage{color}
\usepackage{fullpage} 
\usepackage{mathptmx}
\usepackage{times}
\usepackage{subfig}
\usepackage{latexsym}
\usepackage{multirow}
\usepackage{amsthm}

\usepackage{amsmath, amssymb, amsfonts} 
\usepackage{graphicx}
\usepackage[table]{xcolor}
\usepackage{algorithm, algorithmic, ifthen}

\usepackage{caption}

\usepackage{epsfig}
\usepackage{epstopdf}
\usepackage{amssymb}
\usepackage{amsmath}
\usepackage{amsfonts}

\newtheorem{observation}{Observation}
\newtheorem{definition}{Definition}
\newtheorem{theorem}{Theorem}
\newtheorem{lemma}{Lemma}

\newtheorem{example}{Example}

\newtheorem{corollary}{Corollary}
\newtheorem{proposition}{Proposition}

\newcommand{\Flip}{{\tt\textsc{Flip}}}
\newcommand{\EFlip}{{\tt\textsc{EFlip}}}

\newcommand{\Out}{{\tt\textsc{Out}}}

\newcommand{\attr}{{\tt{Att}}}

\newcommand{\Worlds}{{\tt{Worlds}}}
\newcommand{\widebar}[1]{{\overline{#1}}}
\newcommand{\ul}[1]{{\underline{#1}}}

\newcommand{\secureview}{{\tt Secure-View}~}  

\newcommand{\scream}[1]{{\color{blue}{\texttt{\textbf{ * #1 *}}}{\typeout{#1}}}}
\newcommand{\red}[1]{{\color{red}{\texttt{\textbf{ * #1 *}}}{\typeout{#1}}}}

\newcommand{\eat}[1]{}
\newcommand{\neweat}[1]{}
\newcommand{\eg}{\emph{e.g.}}
\newcommand{\ie}{\emph{i.e.}}

\newcommand{\angb}[1]{\langle #1 \rangle}
\newcommand{\comple}[1]{\widebar{#1}}

\newcommand{\cost}{{\tt cost}}

\newcommand{\be}{\begin{enumerate}}
\newcommand{\ee}{\end{enumerate}}

\newcommand{\Dom}{\texttt{Dom}}
\newcommand{\Codomain}{\texttt{CoDom}}
\newcommand{\tup}[1]{\mathbf{#1}}
\newcommand{\proj}[2]{\Pi_{{#1}}({#2})}

\newcommand{\true}{\textsc{True}}
\newcommand{\false}{\textsc{False}}

\newcommand{\safe}{{\mathbf{S}}}
\newcommand{\uds}{{\mathbf{U}}}

\newcommand{\UDS}{{\texttt{UD-safe}}}
\newcommand{\DS}{{\texttt{D-safe}}}
\newcommand{\US}{{\texttt{U-safe}}}
\newcommand{\wrt}{{with respect to~}}
\begin{document}
\title{A Propagation Model for Provenance Views of\\ Public/Private Workflows}
\eat{
\author{
Susan B. Davidson\\
      University of Pennsylvania\\
       \texttt{susan@cis.upenn.edu}
       
       Tova Milo\\
       Tel Aviv University\\
       \texttt{milo@cs.tau.ac.il}
\alignauthor
Sudeepa Roy\\
       University of Pennsylvania\\
       \texttt{sudeepa@cis.upenn.edu}
}
}

\author{Susan B. Davidson\thanks{University of Pennsylvania; Philadelphia, PA, USA. {\tt susan@cis.upenn.edu}.}
~~~~~~~~~~~~~~~~~~Tova Milo\thanks{Tel Aviv University; Tel Aviv, Israel. {\tt milo@cs.tau.ac.il}.}
~~~~~~~~~~~~~~~~~~Sudeepa Roy\thanks{University of Washington; Seattle, WA, USA. {\tt sudeepa@cs.washington.edu}. This work was done while the author was in the University of Pennsylvania.}
}

\maketitle

\begin{abstract}
	We study the problem of concealing functionality of a proprietary or private module
when provenance information is shown over repeated executions of a 
workflow which contains both {\em public} and {\em private} modules.    
Our approach is to use {\em provenance views} to hide carefully chosen subsets of data over 
all executions of the workflow to ensure $\Gamma$-privacy:  for each private module and 
each input $x$, the module's output $f(x)$ is indistinguishable from $\Gamma -1$ other possible 
values given the visible data in the workflow executions. We show that $\Gamma$-privacy cannot be achieved 
simply by combining solutions for individual private modules; data hiding 
must also be {\em propagated} through public modules. We then examine how much additional data 
must be hidden and when it is safe to stop propagating data hiding.  The answer depends strongly 
on the workflow topology as well as the behavior of public modules on the visible data. 
In particular, for a class of workflows (which include the common tree and chain workflows),
taking private solutions for each private module, augmented with a {\em public closure} that
is {\em upstream-downstream safe},  ensures $\Gamma$-privacy. We define these notions formally and show that
the restrictions are necessary. We also study the related optimization problems
of minimizing the amount of hidden data.

\end{abstract}

 \section{Introduction}\label{sec:intro}

Workflow provenance has been extensively studied, and is
increasingly captured in workflow  systems to ensure reproducibility, enable debugging, and verify the validity and
reliability of results.  However,  as pointed out in
~\cite{DBLP:conf/icdt/DavidsonKRSTC11}, there is a tension between
provenance and privacy:  Confidential intermediate data may be shown
({\em data privacy}); the functionality of proprietary modules may
become exposed by showing the input and output values to that module
over all  executions of the workflow ({\em module privacy});  and
the exact execution path taken in a specification, hence details of
the connections between data, may be revealed ({\em structural
privacy}).  An increasing amount of attention is therefore being
paid to specifying privacy concerns, and developing techniques to
guarantee that these concerns are
addressed~\cite{NiXBSH09,TanGMJMTM06,CadenheadKT11,CadenheadKKT11}.

This paper focuses on privacy of module functionality, in particular
in  the general -- and common -- setting in which proprietary ({\em
private}) modules are used in workflows which also contain
non-proprietary ({\em public}) modules, whose functionality  is
assumed to be known by users.  
There are proprietary modules for tasks like gene sequencing, protein folding, medical diagnoses, 
that are commercially available and are combined with
other modules in a \emph{workflow} for different biological or medical experiments \cite{smartgene, dna20}.
The functionality of these proprietary modules 
(\ie\ what result will be output for a given input) is not known, and
owners of these proprietary modules would like to ensure that
their functionality is not revealed when the provenance information is published.
In contrast for a public module (\eg\ a reformatting or sorting module), given an input to the module a user can construct the
 output even if the exact algorithm
used by the module is not known by users (\eg\ Merge sort vs
Quick sort).
 
 
 \eat{
 In contrast, the functionality of private modules
 (i.e. what result will be output for a given input)
 is not known and
 should not be revealed by the visible provenance information.
 Examples of private modules include proprietary gene sequencing and medical diagnosis modules.
 }

Following \cite{DKM+11}, the approach we use is to extend the
notion of $\ell$-diversity~\cite{MKG+07} to the workflow setting by
carefully choosing a subset of intermediate input/output data to
hide over {\em all} executions of the workflow so that each private
module is ``$\Gamma$-private'':  for every input $x$, the actual
value of the output of the module, $f(x)$, is indistinguishable from
$\Gamma-1$ other possible values w.r.t. the visible data values in
the provenance information (in Section~\ref{sec:related} we discuss ideas
related to differential privacy).
The complexity of the problem arises from the fact that modules
interact with each other through data flow defined by the 
workflow structure, and therefore merely hiding subsets of inputs/outputs for
private modules may not guarantee their privacy when embedded in a
workflow.
We consider workflows with directed acyclic graph (DAG) structure,
that are commonly used in practice \cite{myexpt}, contain common chain and tree workflows,
and comprise a fundamental yet non-trivial class of workflows 
for analyzing module privacy.

As an example, consider a private module $m_2$, which we assume is non-constant.
Clearly, when executed in isolation as a {\em standalone} module,
then either hiding all its inputs or hiding all
its outputs over all executions guarantees privacy for 
any privacy parameter $\Gamma$.   However, suppose $m_2$ is embedded in a simple
chain workflow $m_1 \longrightarrow m_2 \longrightarrow m_3$,  where both $m_1$ and
$m_3$ are public, equality modules. Then even if we hide {\em both}
the input and output of $m_2$, their values can be retrieved from
the input to $m_1$ and the output from $m_3$.   Note that the same
problem would arise if $m_1$ and $m_3$ were invertible functions, e.g.
reformatting modules, a common case in practice.

In \cite{DKM+11}, we showed that in a workflow
 with only private modules
 (an {\em all-private workflow}) the problem has a simple, elegant
solution:  If a set of hidden input/output data guarantees
$\Gamma$-standalone-privacy for a private module, then if the module
is placed in an all-private workflow where a superset of that data
is hidden, then $\Gamma$-workflow-privacy is guaranteed for that
module in the workflow. In other words, in an all-private workflow,
hiding the union of the corresponding hidden data of the individual
modules guarantees $\Gamma$-workflow-privacy for all of  them.
Clearly, as illustrated above, this does not hold when the private
module is placed in a workflow which contains public and private
modules (a {\em public/private workflow}).   In \cite{DKM+11} we
therefore explored {\em privatizing} public modules, i.e. hiding the
names of carefully selected public modules so that their function
is no longer known, and then hiding subsets of input/output data
to ensure their $\Gamma$-privacy.  Returning to the example above,
if it were no longer known that $m_1$ was an equality module then
hiding the input to $m_2$ (output of $m_1$) would be sufficient.
Similarly, if $m_3$ was privatized then hiding the output of $m_2$
(input to $m_3$) would be sufficient.
It may appear that merging some public modules with preceding or succeeding
private modules may give a workflow with all private modules and then the
methods from \cite{DKM+11} can be applied. 
However, merging may be difficulty for workflows with complex network structure,
large amount of data may be needed to be hidden, and more importantly, 
it may not be possible to merge at all when the structure of the workflow is known.

%

 Although privatization is a reasonable approach in some cases, there
 are many practical scenarios where it cannot be employed. For instance,
 when the workflow specification (the module names and
 connections) is already known to the users, or when the identity of
 the privatized public module can be discovered through the structure
 of the workflow and the names or types of its inputs/outputs.

 {\em To overcome this problem, we propose an alternative
 novel solution, based on the propagation of data hiding through
 public modules. }
 Returning to our example, if the input to $m_2$
were hidden then the input to $m_1$ would also be hidden, although
the user would still know that $m_1$ was the equality function.
Similarly, if the output of $m_2$ were hidden then the output of
$m_3$ would also be hidden; again, the user would still know that
$m_3$ was the equality function.
 While in this example things appear to be simple, several
 technically challenging issues must be addressed when employing
 such a propagation model in the general case: 1) whether to propagate hiding upward (e.g. to $m_1$) or
 downward (e.g. to $m_3$); 2) how far to propagate data hiding; and
 3) which data of public modules must be hidden.
 Overall the goal is to guarantee that the functionality of private
 modules is not revealed while minimizing the amount of hidden data.



In this paper we focus on  {\em downward} propagation, for reasons that
will be discussed in Section~\ref{sec:prelims-propagation}. {\em
Using a downward propagation model, we show the following strong
results}: For a special class of common workflows, {\em
single (private)-predecessor workflows}, or simply {\em single-predecessor workflows} (which include the common tree and
chain workflows), taking solutions for $\Gamma$-standalone-privacy
of each private module ({\em safe subsets}) augmented with specially
chosen input/output data of public modules in their {\em public closure}
(up to a successor private module) that is rendered {\em
upstream-downstream safe (\UDS)} by the data hiding, and hiding the
union of data in the augmented solutions for each private module
will ensure $\Gamma$-workflow privacy for all private modules.
 We define these notions formally in Section~\ref{sec:prelims-propagation}
and go on to show that single-predecessor workflows is
 the largest class of workflows for which propagation of  data hiding only
within the public closure suffices.

\eat{

In this paper we focus on  downward propagation, for reasons that
will be discussed in Section~\ref{sec:prelims-propagation}. {\em
Using a downward propagation model, we show the following strong
results}: For a special class of common workflows, {\em
single-predecessor workflows} (which include the common tree and
chain workflows), taking solutions for $\Gamma$-standalone-privacy
of each private module ({\em safe subsets}) augmented with
specially chosen input/output data of public modules up to a successor private module
(technically, a {\em public closure} that is rendered {\em upstream-downstream safe (\UDS)} by
the data hiding),  and hiding the union of data in the augmented solutions for each private
module will ensure $\Gamma$-workflow privacy for all private
modules.
 We define these notions formally in Section~\ref{sec:prelims-propagation}
and go on to show that single-predecessor workflows is
 the largest class of workflows for which propagation of  data hiding only
within the public closure suffices.

}

Since data may have different {\em costs} in terms of hiding, and
there may be many different safe subsets for private modules and \UDS\ 
subsets for public modules, {\em  the next problem we address is finding a
minimum cost solution -- the  {\em optimum view problem}}.
Using the result from above, we show that for single-predecessor workflows
the optimum view problem may be solved by first identifying safe and
\UDS\ subsets for the private and public modules, respectively,
then assembling them together optimally. The complexity of
identifying safe subsets for a private module was studied in
\cite{DKM+11} and the problem was shown to be NP-hard (EXP-time) in
the number of module attributes. 
We show here that
identifying \UDS\ subsets for public modules is of similar complexity:
Even deciding whether a given subset is \UDS\ for a module is coNP-hard in the number of input/output data.
We note however that this is not as negative as it might appear, since the number of
inputs/outputs of individual modules is not high; furthermore, the computation may
be performed as a pre-processing step with the cost being amortized
over possibly many uses of the module in different workflows. In
particular we show that, given the computed subsets, for chain and
tree-shaped workflows, the optimum view problem has a polynomial
time solution in the size of the workflow and the maximum number of
safe/\UDS\ subsets for a private/public modules. Furthermore, the
algorithm can be applied to general single-predecessor workflows
where the public closures have chain or tree shapes. In contrast,
when the public closure has an arbitrary DAG shape, the problem becomes
NP-hard (EXP-time) in the size of the public closure.



{\em We then consider general acyclic workflows}, and give a sufficient condition
to  ensure $\Gamma$-privacy that is not the trivial solution of
hiding all data in the workflow.  In contrast to single-predecessor
workflows, hiding data within a public closure no longer suffices;
data hiding must continue through other private modules to the
entire downstream workflow. \eat{Significantly, while propagation
for a private module in a single-predecessor workflows stops when
another private module is encountered, in a general workflow
propagation must continue through other private modules.} In return,
the requirement from data hiding for public modules is somewhat
weaker here: hiding must  only ensure that the module is {\em
downstream-safe} (\DS), which typically involves fewer input/output
data than upstream-downstream-safety (\UDS). \eat{Since
single-predecessor workflows are a special case of general
workflows, this also yields an alternate, incomparable solution for
single-predecessor workflows. In the first, more input/output data of
public modules may be hidden to satisfy 
\UDS ty, 
but
propagation stops when another private module is encountered.  In
the second, fewer input/output data may be hidden to satisfy 
\DS ty, 
but propagation extends to more modules.}

%

The remainder of the paper is organized as follows:   Our workflow
model  and notions of standalone- and workflow-module privacy are
given in Section~\ref{sec:prelims-PODS11}.
Section~\ref{sec:prelims-propagation} describes our propagation
model,  defines upstream-downstream-safety and single-predecessor
workflows, and states the privacy theorem.
Section~\ref{sec:privacy-theorems} discusses the proof of the
privacy theorem, and the necessity of the upstream-downstream-safety
condition as well as the single-predecessor restriction.  The
optimization problem is studied in Section~\ref{sec:optimization}.
We then discuss general public/private workflows in
Section~\ref{sec:privacy-theorems}, before giving related work in Section~\ref{sec:related} and concluding in
Section~\ref{sec:conclusion}.



\section{Preliminaries}\label{sec:prelims-PODS11}
We start  by reviewing the formal definitions and
notions of module privacy from \cite{DKM+11}, and then extend them
to the context studied in this paper.\footnote{The example in this section is also taken from \cite{DKM+11}.} Readers familiar with the
definitions and results in \cite{DKM+11} can 
move directly to Section~\ref{sec:prelims-propagation}.

\subsection{Modules, Workflows  and Relations} \label{sec:relations}
\vspace{-3mm}
\paragraph*{Modules} A module $m$ with a set 
$I$ of input data and a set
$O$ of (computed) output data is modeled as a relation 
$R$.
$R$ has the 
set of attributes 
$A = I \cup O$, and 
satisfies the
functional dependency 
$I \rightarrow O$.  
We assume that
$I \cap O = \emptyset$ 
and will refer to 
$I$ and $O$ as the
\emph{input attributes} and \emph{output attributes} of $R$ respectively.

We assume that the values of each attribute 
$a \in A$
 come from a
finite but arbitrarily large domain $\Delta_{a}$, and let $\Dom =
\prod_{a \in I}\Delta_{a}$ and $\Codomain = \prod_{a \in
O}\Delta_{a}$ denote the {\em domain} and {\em co-domain} of the
module $m$ respectively.\footnote{We distinguish between the possible
range $O$ of the function $m$ that we call \emph{co-domain} and
the {\em actual range} $\{\tup{y}: \exists \tup{x} \in I ~s.t.~ \tup{y}
= m(\tup{x})\}$}
The relation $R$ thus
represents the (possibly partial) function $m: \Dom \rightarrow
\Codomain$ and tuples in $R$ describe executions of $m$, namely
for every $t \in R$, $\proj{O}{\tup{t}}=m(\proj{I}{\tup{t}})$.
We overload the standard notation for projection, $\proj{A}{R}$,
and use it for a tuple $\tup{t} \in R$. Thus $\proj{A}{\tup{t}}$,
for a set $A$ of attributes, denotes the projection of $\tup{t}$ to
the attributes in $A$.

\paragraph*{Workflows} A workflow $W$ consists of a set of modules $m_1, \cdots, m_n$,
connected as a DAG (see, for instance, the workflow in
Figure~\ref{fig:wf-view}).
\eat{
				\scream{Rev-1 said why not cyclic workflows,\\
				 should we mention that acyclic wfs are \\
				 fundamental, and commonly found case, and \\
				 dynamic workflows cannot be modeled with \\
				 fixedrelations?}
				\scream{You should mention this in the introduction:  it is not only a common case but is already difficult.}
}
\eat{Each module $m_i$ has a set $I_i$ of
input attributes and a set $O_i$ of output attributes, and $A_i = I_i \cup O_i$
is the set of all attributes.}
We assume  that
\eat{(1) for each module, its input and output attributes are
disjoint, i.e. $I_i \cap O_i = \emptyset$, }
(1) the output attributes of distinct modules are disjoint, namely $O_i \cap O_j =
\emptyset$, for $i\neq j$ (i.e. each data item is produced by a
unique module); and (2) whenever an output of a module $m_i$ is fed
as input to a module $m_j$ the corresponding output and input
attributes of $m_i$ and $m_j$ are the same. The DAG shape of the
workflow guarantees that these requirements are not contradictory.

We model executions of $W$ as a relation $R$ over the set of
attributes $A=\cup_{i=1}^{n} A_i$, 
satisfying the set of
functional dependencies $F = \{I_i \rightarrow O_i : i \in [1,
n]\}$. Each tuple in $R$ describes an execution of the workflow $W$.
In particular, for every $t \in R$, and every $i \in [1, n]$,
$\proj{O_i}{\tup{t}}=m_i(\proj{I_i}{\tup{t}})$. One can think of
$R$ as containing (possibly a subset of) the join of the individual
module relations.

\begin{example}\label{ex:module}
Figure~\ref{fig:wf-view} shows a  workflow involving three
modules $m_1,m_2,m_3$ with boolean input and output attributes
implementing the following functions:
\eat{
Module $m_1$ takes as input two data items, $a_1$ and $a_2$, and
computes $a_3\! =\! a_1\!\vee\! a_2$, $a_4\! =\! \neg({a_1\!\wedge\!
a_2})$ and $a_5\! =\! \neg({a_1\!\oplus \!a_2})$. (The symbol
$\oplus$ denotes XOR).
}
(i) $m_1$ computes $a_3\! =\! a_1\!\vee\! a_2$, $a_4\! =\! \neg({a_1\!\wedge\!
a_2})$ and $a_5\! =\! \neg({a_1\!\oplus \!a_2})$, where $\oplus$ denotes XOR;
(ii) $m_2$ computes $a_6\! =\! \neg({a_3\! + \!a_4})$; and (iii) $m_3$ computes
 $a_7\! =\! a_4\!\wedge\! a_6$.
The relational representation (functionality) $R_1$
of module $m_1$  
with the functional dependency $a_1 a_2 \longrightarrow a_3 a_4
a_5$ is shown in Figure~\ref{fig:m1}. For clarity, we have added $I$ (input) and $O$ (output) above
the attribute names to indicate their role. The relation $R$
describing the workflow executions is shown in Figure~\ref{fig:wf-prov}
which has the functional dependencies $a_1 a_2 \longrightarrow a_3 a_4
a_5$, $a_3 a_4 \longrightarrow a_6$, $a_4 a_5 \longrightarrow a_7$
from modules $m_1, m_2, m_3$ respectively.
\eat{
 The
input and output attributes of modules $m_1, m_2, m_3$ respectively
are (i) $I_1 = \{a_1, a_2\}$, $O_1 = \{a_3, a_4, a_5\}$, (ii) $I_2 =
\{a_3, a_4\}$, $O_2 = \{a_6\}$ and (iii) $I_3 = \{a_4, a_5\}$, $O_3
= \{a_7\}$. The underlying functional dependencies in the relation
$R$ in Figure~\ref{fig:wf-prov} reflect the keys of the constituent
modules, e.g. from $m_1$ we have $a_1 a_2 \longrightarrow a_3 a_4
a_5$, from $m_2$ we have $a_3 a_4 \longrightarrow a_6$, and from
$m_3$ we have $a_4 a_5 \longrightarrow a_7$.
}
\end{example}

\begin{figure}[t]
\centering
{
~\subfloat[{\scriptsize $R_1$: Functionality of $m_1$} \vspace{-5.3mm}]{\label{fig:m1}
\begin{tabular}{|cc|ccc|}
\hline \multicolumn{2}{|c|}{$I$} & \multicolumn{3}{c|}{$O$}\\ \hline
$a_1$ & $a_2$ & $a_3$ & $a_4$ & $a_5$ \\ \hline  \hline 0 & 0 & 0 & 1
& 1 \\ \hline 0 & 1 & 1 & 1 & 0 \\ \hline 1 & 0 & 1 & 1 & 0 \\
\hline 1 & 1 & 1 & 0 & 1 \\ \hline
\end{tabular}
}~
 \subfloat[{\scriptsize $R$:  Workflow executions}]{
\begin{tabular}{|lllllll|} \hline
$a_1$ & $a_2$ & $a_3$  & $a_4$ & $a_5$ & $a_6$& $a_7$\\  \hline
\hline 0 & 0 & 0 & 1 & 1 & 1& 0\\ \hline 0 & 1 & 1 & 1 & 0 & 0 & 1\\
\hline 1 & 0 & 1 & 1 & 0 & 0 &1\\ \hline 1 & 1 & 1 & 0 &1 & 1 &1
\\\hline
\end{tabular}\label{fig:wf-prov} }
 ~\subfloat[{\scriptsize \!\! 
 $R' \!=\!\proj{A \setminus H}{R_1}$}]{
 \label{fig:m1'}
\begin{tabular}{|c|cc|}
\hline \multicolumn{1}{|c|}{
$I \setminus H$} & \multicolumn{2}{c|}{ 
$O \setminus H$} \\ \hline $a_1$ & $a_3$ & $a_5$ \\ \hline  \hline 0 & 0 & 1
\\ \hline 0 & 1 & 0 \\ \hline 1 & 1 & 0 \\ \hline 1 & 1 & 1 \\
\hline
\end{tabular}
} ~\subfloat{
\begin{tabular}{c}
    \includegraphics[scale=.15]{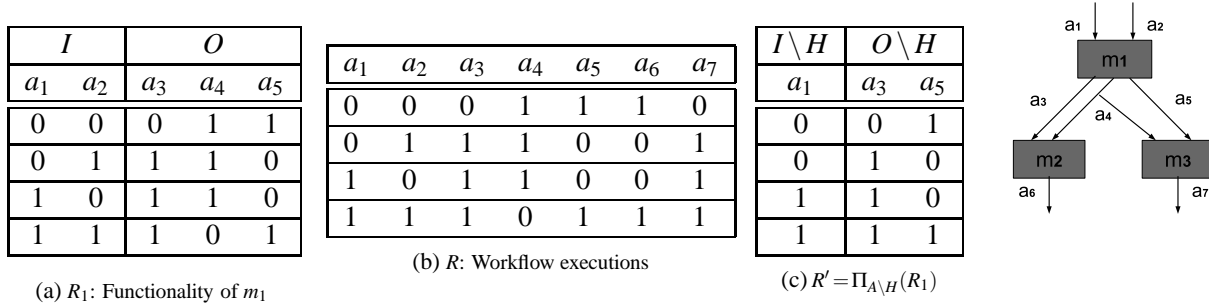}
\end{tabular}
  }
 \caption{Module and workflow executions as
relations, and view} 
\label{fig:wf-view}
}
\end{figure}

\textbf{Data sharing} refers to an output attribute of a module acting as input to more than one module
(hence $I_i\cap I_j \neq \emptyset$ for $i \neq j$).  In the example above, attribute $a_4$ is shared by both $m_2$ and $m_3$.


%

%
%

\subsection{Module Privacy}
\label{sec:priv} We consider the privacy of a single module, which is
called {\em standalone module privacy}, then privacy of modules when
they are connected in a workflow, which is called \emph{workflow
module privacy}. We study this given two types of modules, {\em
private} modules (the focus of \cite{DKM+11}) and {\em public}
modules (the focus here).

\paragraph*{Standalone module privacy}
Our approach to ensuring standalone module privacy, for a module
represented by the relation $R$, is to hide a carefully chosen
subset $H$ of $R$'s attributes (called \emph{hidden attributes}). 
In other words, we  project $R$ on a
restricted subset $A \setminus H$, where $A$
is the set of all attributes of $m$. 
The set $A \setminus H$ is called \emph{visible attributes}.
The users are allowed access only to the view
$R' =\proj{A \setminus H}{R}$.

One may distinguish two types of modules. (1) {\em Public modules}
whose behavior is fully known to users. Here users have a prior
knowledge about the full content of $R$ and, even if given only
the view 
$R'$, they are able to fully (and exactly) reconstruct
$R$. Examples include reformatting or sorting modules.
  (2) \emph{Private modules} where such a priori knowledge does not
exist. Here, the only information available to users, on the
module's behavior, is the one given by 
$R'$. Examples include
proprietary software,  e.g. a genetic disorder susceptibility
module.

Given a view (projected relation) 
$R'$
of a private module $m$, the
\emph{possible worlds} of $m$ are all the possible full relations
(over the same schema as $R$) that are consistent with 
the view $R'$.
Formally,
\begin{definition}\label{def:pos-worlds-standalone}
Let $m$ be a private module with a corresponding relation $R$,
having input and output attributes $I$ and $O$ respectively.
Let $A = I \cup O$ be the set of all attributes. 
Given a set of hidden attributes $H$, the set of \emph{\bf possible worlds} for $R$ \wrt 
$H$, denoted
$\Worlds(R, H)$, consists of all relations $R'$ over the same
schema as $R$ that satisfy the functional dependency $I
\rightarrow O$, and where 
$\proj{A \setminus H}{R'} = \proj{A \setminus H}{R}$.
\end{definition}

To guarantee privacy of a module $m$, the view 
$R'$ 
should ensure some level of uncertainly \wrt the value of the output
$m(\proj{I}{\tup{t}})$, for tuples $t\in R$. To define this, we
introduce the notion of $\Gamma$-standalone-privacy, for a given
parameter $\Gamma \geq 1$. Informally, 
a view $R'$ is
$\Gamma$-standalone-private if for every $t \in R$,  
$\Worlds(R, H)$
contains at least $\Gamma$ distinct output
values that could be the result of $m(\proj{I}{\tup{t}})$.
\begin{definition}\label{def:standalone-privacy}
Let $m$ be a private module with a corresponding relation $R$
having input and output attributes $I$ and $O$ resp. Then $m$ is
\emph{\bf $\Gamma$-standalone-private} \wrt a set of 
hidden attributes $H$, if for every tuple $\tup{x} \in \proj{I}{R}$, \mbox{
$|\Out_{\tup{x}, m, H}| \geq \Gamma$}, where
  $\Out_{\tup{x}, m, H} = \{\tup{y} \mid  \exists R' \in \Worlds(R, H),~~\exists \tup{t'} \in R'~~ s.t~~
  \tup{x} = \proj{I}{\tup{t'}} \wedge \tup{y}=\proj{O}{\tup{t'}} \}$.\footnote{In \cite{DKM+11}, we (equivalently) defined 
privacy \wrt\ visible attributes $V$ instead of hidden attributes $H$, and 
we used the notation ``$\Out_{\tup{x}, m}$ with respect to $V$'' instead of $\Out_{\tup{x}, m, H}$.}
\par
If $m$ is \emph{$\Gamma$-standalone-private} \wrt 
hidden attributes $H$, then we call 
$H$ a \emph{safe subset} for
$m$ and $\Gamma$.
\end{definition}

A module cannot be differentiated from its possible worlds with respect to the visible attributes,
and therefore, whether the original module, or one from its possible worlds 
is being used cannot be recognized. Hence,
$\Gamma$-standalone-privacy
implies that for \emph{any} input
the adversary cannot guess $m$'s output with probability
$> \frac{1}{\Gamma}$, even if the module is executed an arbitrary
number of times. 

\begin{example}\label{eg:sa-worlds}
Returning to module $m_1$, suppose the 
hidden attributes are $H = \{a_2, a_4\}$
 resulting in the view 
 $R'$ in
Figure~\ref{fig:m1'}. For clarity, we have added 
$I \setminus H$ (visible
input) and 
$O \setminus H$ (visible output) above the attribute names to
indicate their role. Naturally, $R_1 \in \Worlds(R_1, H)$,
and we can check that overall there are 64 relations in $\Worlds(R_1, H)$.

Furthermore, it can be verified that, if 
$H = \{a_2, a_4\}$,  then for all $\tup{x} \in \proj{I}{R_1}$,
$|\Out_{\tup{x}, m_1, H}| \geq 4$, so
$\{a_1, a_3, a_5\}$ is safe for $m_1$ and $\Gamma = 4$.
As an example, 
when $\tup{x} = (0, 0)$,
 $\Out_{\tup{x},m, H} \supseteq \{ (0, \underline{0}, 1)$, $(0, \underline{1}, 1)$, $(1, \underline{0}, 0)$, $(1, \underline{1}, 0) \}$
 (hidden attributes are underlined) -- we can define four possible worlds that map $(0, 0)$ to these outputs (see \cite{DKM+11} for details).
Also, hiding any two output attributes from $O = \{a_3, a_4, a_5\}$
ensures standalone privacy for $\Gamma = 4$,
\eg\ if $H = \{a_2, a_4\}$, 
then the input $(0, 0)$ can be mapped to one of $(0,
\underline{0}, \underline{0})$, $(0, \underline{0}, \underline{1})$,
$(0, \underline{1}, \underline{0})$ and $(0, \underline{1},
\underline{1})$; this holds for other assignments of input
attributes as well. However,
$H = \{a_1, a_2\}$ (input attributes) is not safe for $\Gamma = 4$: for
any input $\tup{x}$, $\Out_{\tup{x}, m, H} = \{(0,1,1)$, $(1,1,0)$, $(1,0,1)\}$,
containing only three possible output tuples.
\end{example}

\vspace{-4mm}
\paragraph*{Workflow Module Privacy} To define privacy in the context of a
workflow, we first extend the notion of {\em possible worlds} to a
workflow view. Consider the view 
$R' =\proj{A \setminus }{R}$
of the relation $R$ of a workflow
$W$, where $A$ is the set of all attributes across all modules in $W$. Since 
$W$ may contain private as well as
public modules, a possible world for 
$R'$ is a full relation that
not only agrees with 
$R'$ on the content of the visible attributes and satisfies the functional dependency,
but is also consistent \wrt the expected behavior of the public
modules. In the following definitions, $m_1, \cdots, m_n$ 
are the modules in 
$W$ and $F = \{I_i\! \rightarrow \!O_i : 1 \leq i \leq n\}$ 
is the set of functional
dependencies in $R$.

\begin{definition}\label{def:pos-worlds-workflow}
The set of \emph{\bf possible worlds for the workflow relation $R$}
\wrt 
hidden attributes $H$ (denoted by \\
$\Worlds(R, H)$)
 consists of all 
relations $R'$ over the same attributes as $R$ that satisfy 
(1) the functional dependencies in $F$, (2) $\proj{A \setminus H}{R'}$ =
$\proj{A \setminus H}{R}$, and (3) $\proj{O_i}{\tup{t'}}$ =
$m_i(\proj{I_i}{\tup{t'}})$ for every public module $m_i$ in $W$ and every
tuple $\tup{t'} \in R'$.
\end{definition}

We can now define the notion of $\Gamma$-workflow-privacy, for a
given parameter $\Gamma \geq 1$. Informally, a view 
$R'$ is
$\Gamma$-workflow-private if for every tuple $t \in R$, and every
private module $m_i$ in the workflow, the possible worlds
$\Worlds(R, H)$ contain at least $\Gamma$ distinct output values
that could be the result of $m_i(\proj{I_i}{\tup{t}})$.

\begin{definition}\label{def:workflow-privacy}
A private module $m_i$ in $W$ is \emph{\bf $\Gamma$-workflow-private}
\wrt a set of 
hidden attributes $H$, if for every tuple
$\tup{x} \in \proj{I_i}{R}$,
  $|\Out_{\tup{x}, W, H}| \geq \Gamma$, where
  $\Out_{\tup{x}, W, H} =$ $\{\tup{y} \mid  \exists R'$ $\in \Worlds(R, H),$ $s.t., $ 
  $\forall$ $\tup{t'}$ $\in R',~ $  $\tup{x} = \proj{I_i}{\tup{t'}}$ $\Rightarrow \tup{y}$ $=\proj{O_i}{\tup{t'}} \}$.

$W$ is called \emph{$\Gamma$-private} if every private module $m_i$
in $W$ is $\Gamma$-workflow-private. If $W$ (resp. $m_i$) is
\emph{$\Gamma$-private} ($\Gamma$-workflow-private) \wrt $H$, then
we call $H$ a \emph{safe subset} for $\Gamma$-privacy of $W$
($\Gamma$-workflow-privacy of $m_i$).
\end{definition}

Similar to standalone module privacy, $\Gamma$-workflow-privacy ensures
that for any input to a module $m_i$, the output cannot be 
guessed with probability $\geq \frac{1}{\Gamma}$ even if $m_i$
belongs to a workflow with arbitrary DAG structure and interacts with other modules
with known or unknown functionality, and even the workflow is
executed an arbitrary number of times.
For simplicity, the above definition assume that the privacy
requirement of every module $m_i$ is the same $\Gamma$. The results
and proofs in this paper remain unchanged when different modules
$m_i$ have different privacy requirements $\Gamma_i$.
Note that there is a subtle difference in workflow privacy of a module defined as above and
 standalone-privacy (Definition~\ref{def:standalone-privacy}); the former 
 uses the logical implication operator ($\Rightarrow$) for defining $\Out_{\tup{x}, W, H}$ 
 while the latter uses
 conjunction ($\wedge$) for defining $\Out_{\tup{x}, m, H}$.
 This is due to the fact that some modules are not \emph{onto}\footnote{For a function $f: D \rightarrow C$, $D$ is the \emph{domain}, 
$C$ is the \emph{co-domain}, and $R = \{y \in C : \exists x \in D, f(x) = y\}$ is the \emph{range}. The function $f$ is \emph{onto}
if $C = R$.};  and as a result the input $x$ itself may not appear in any execution of the possible world $R'$.
Nevertheless, there is an alternative definition of module $m_i$ that maps $x$ to $y$
and can be used in the workflow for $R'$ consistently with the visible data.


\subsection{Composability Theorem and Optimization}\label{sec:privacy-opt}
 Given a workflow $W$ and 
parameter $\Gamma$, there may be several incomparable (in terms of
set inclusion) safe subsets 
$H$ for the (standalone) modules in $W$
and for the workflow as a whole. Some of the corresponding 
$R'$ views may be preferable to others, e.g. they provide users with more
useful information, allow more common/critical user
queries  to be answered, etc. 
If $\cost(H)$
denotes the penalty of hiding the attributes in
$H$, a natural goal is to choose a safe subset 
$H$ that
minimizes 
$\cost(H)$. A particular instance of the problem
is when the cost function is additive: each attribute $a$ has some
penalty value $\cost(a)$ and the penalty of hiding 
$H$ is
$\cost(H) = \Sigma_{a\in H} \cost(a)$. 
\par
On the negative side, it was shown in \cite{DKM+11} that the
corresponding decision problem is hard in the number of
attributes, even for a single module and even in the presence of an
oracle that tests whether a given attribute subset is safe. On the
positive side, however, it was 
shown that {\em when the
workflow consists only of private modules} (we call these ``\emph{all-private}'' workflows), once privacy has been
analyzed for the individual modules, the results can be lifted to
the whole workflow. 
In particular, the following theorem says that,
hiding the union of 
hidden attributes of standalone-private solutions of the individual modules
in an all-private workflow
guarantees $\Gamma$-workflow-privacy for all of them.



\begin{theorem}\label{thm:privacy-private} {\bf(Composability Theorem for 
All-private Workflows \cite{DKM+11})} Let $W$ be a workflow
consisting only of private modules $m_1, \cdots, m_n$. 
For each 
$i \in [1, n]$,
let 
$H_i \subseteq A_i$ 
be a set of safe hidden attributes 
for $\Gamma$-standalone-privacy of $m_i$. 
Then the workflow $W$ is
$\Gamma$-private \wrt 
hidden attributes $H = \bigcup_{i = 1}^n H_i$.
\end{theorem}

%

\eat{
This theorem says that,
if a set of visible attributes guarantees
$\Gamma$-standalone-privacy for a private module, then if the module
is placed in an all-private workflow where only a subset of those
attributes is made visible, then $\Gamma$-workflow-privacy is
guaranteed for the module in this workflow. In other words,
 in an
all-private workflow, hiding the union of the corresponding hidden
attributes of standalone-private solutions of the individual modules
guarantees $\Gamma$-workflow-privacy for all of them.
}
\par
It was also observed in \cite{DKM+11}
that the number of attributes of individual modules can be 
much smaller than the total number of attributes in a workflow,
and 
that a  proprietary module 
may be used in many different workflows.
Therefore, the obvious brute-force algorithm, which is essentially the best possible, 
can be used (possibly as a pre-processing step)
to find all standalone-private solutions of individual modules.
Then  any set of ``local solutions'' for each module can be composed 
to give a global feasible solution.
Moreover, the composability theorem 
ensure that the private solutions are valid even with respect to future workflow executions which have not yet been 
recorded in the workflow relation.

\par
Given Theorem~\ref{thm:privacy-private},
 \cite{DKM+11} focused on a modified optimization problem:
combine standalone-private solutions optimally to get a workflow-private solution.
This optimization problem, which we refer to as \textbf{optimal composition problem},
 remains NP-hard even in the simplest scenario,
and therefore, \cite{DKM+11} proposed efficient approximation algorithms.

\eat{
				The importance of the privacy theorem is that it allows us to efficiently
				find a feasible private solution for the entire workflow with $n$ modules
				by solving smaller problems for each module locally.
				In turn, the composability theorem 
				enables a variety
				of optimization problems (with different representation of the
				standalone solutions\footnote{There may be multiple hidden subsets of
				attributes for individual modules that guarantee their
				$\Gamma$-standalone-privacy (see Example~\ref{eg:sa-worlds}) and
				different representations may be used to describe them.}, with or
				without data sharing, etc) that aim to optimally combine standalone solutions for the
				private modules to get 
				a private solution for the full workflow.
				 While even the most restricted case turned out to be
				NP-hard, PTIME (in number of modules, number of attributes, and the
				size to represent the standalone solutions) approximation algorithms
				for these problems are possible and were studied in \cite{DKM+11}.
				Note that, an optimal combination of the standalone solutions
				may or may not give an optimal private view for the workflow. 
				However, directly finding the optimal solution is hard even for the simplest possible workflows
				(\ie, workflows containing only one module), the optimal combination of private solutions
				is an acceptable solution for finding a private solution for the entire workflow.
}

\section{Privacy via propagation}\label{sec:prelims-propagation}
Workflows with both public and private modules
are harder to handle than workflows with all private modules.
In particular, the composability theorem (Theorem~\ref{thm:privacy-private})
does not hold any more.  
To see why, we revisit the example mentioned in the introduction.

\begin{example}\label{eg:public-challenge}
Consider a workflow with three modules $m_1, m_2$ and $m_3$ as shown in Figure~\ref{fig:wf-propagation-need}.
For simplicity, assume that all modules have a boolean input and a boolean output, and 
 implement the equality function (\ie, $a_1 = a_2 = a_3 = a_4$).
Module $m_2$ is private, and the modules $m_1, m_3$ are public. 
When the private module $m_2$ is standalone, it can be verified that either hiding its input $a_2$ or hiding its output $a_3$
guarantees $\Gamma$-standalone-privacy for $\Gamma = 2$. 
However, in the workflow, if $a_1$ and $a_4$ are visible
then the actual values of $a_2$ and $a_3$ can be found exactly since it is known that the public modules $m_1, m_3$ are equality modules.
\end{example}

One intuitive way to overcome this problem is to  propagate  the
hiding of data through the problematic public modules, i.e., to hide
the attributes of public models that may disclose information about
hidden attributes of private modules. To continue with the above
example, if we choose to hide input $a_2$ (respectively, output $a_3$) to protect the privacy of module
$m_2$, then we propagate the hiding {\em upstream} (resp. {\em downstream}) to the
public modules and hide the
input attribute $a_1$ of $m_1$ (respectively, the output attribute $a_4$ of $m_3$).

\eat{
		choose to protect the privacy of $m_2$ by hiding $\log
		\Gamma$ of its input (resp. output) attributes, then we also need to
		propagate the hiding {\em upward} (resp. {\em downward}) to the
		public module supplying the input (receiving the output) and hide
		the 
		corresponding input (output) attributes of the equality module $m_1$ ($m_3$).
}

The workflow in the above example has a simple structure, and the functionality of its component modules
is also simple. In general, three main issues arise when employing such a propagation model:
(1) upward vs. downward propagation; (2) repeated propagation; and (3) choosing
which attributes to hide.  We discuss these issues next.

\subsection{Upstream vs. Downstream propagation} 
\eat{
			\scream{Sue: I tried to simplify this, see if it is correct. I found the whole "vector  of attributes" confusing, 
			and the argument had nothing to do with the
			specific structure of the workflow in Fig. 2a.}
			
}
Which form of propagation
can be used depends on the safe subsets chosen for the private
modules as well as properties of the public modules. To see
this, consider again Example \ref{eg:public-challenge}, and assume
now that public module $m_1$ computes some constant function 
(\eg, $m_1(0) = m_1(1) = 0$).
If input attribute $a_2$ for module $m_2$ is hidden, then using upward propagation to hide 
the input attribute $a_1$ of $m_1$ does not preserve the
$\Gamma$-workflow-privacy of $m_2$ for 
$\Gamma > 1$. This is
because it suffices to look at the (visible) output attribute $a_3 = 0$ of
$m_2$ to know 
that $m_2(0) = 0$.
\eat{
In general, consider arbitrary attribute vectors $\tup{a_1}, \cdots, \tup{a_4}$
in the workflow in Figure~\ref{fig:wf-propagation-need}. Suppose hiding a subset of input attributes of 
the private module $m_2$ gives $\Gamma_1$-standalone-privacy. However, 
hiding the same attribute subset in the workflow gives $\Gamma_2$-workflow-privacy where $\Gamma_1 \geq \Gamma_2$
even if the entire input vector $\tup{a_1}$ to the predecessor public module 
$m_1$ is hidden. 
It is possible that $\Gamma_1 >> \Gamma_1$ 
unless $m_1$ is an onto function
\footnote{For a function $f: D \rightarrow C$, $D$ is the \emph{domain}, 
$C$ is the \emph{co-domain}, and $R = \{y \in C : \exists x \in D, f(x) = y\}$ is the \emph{range}. The function $f$ is \emph{onto}
if $C = R$.};  
in the worst case, if $m_1$ is a constant function, then $\Gamma_2 = 1$ whereas $\Gamma_1$ can be arbitrarily large. 
}
In general, upward propagation from a subset of input attributes 
which gives $\Gamma_1$-standalone-privacy for a private module $m$ will only yield 
$\Gamma_2$-workflow-privacy for $m$, where $\Gamma_1 \geq \Gamma_2$.  
It is possible that $\Gamma_1 >> \Gamma_1$ 
unless upstream public modules are onto functions;
in the worst case, if upstream modules are constant functions, then $\Gamma_2 = 1$ whereas $\Gamma_1$ can be arbitrarily large. 
Unfortunately, it is not common for modules to be onto functions (e.g. some
output values 
may be well-known to be non-existent).  

\par
In contrast,  when the privacy of a private module 
is achieved by {\em hiding output
attributes only}, using downstream propagation it is possible to  achieve the same
privacy guarantee in the workflow as with the standalone case without imposing any restrictions on
the public modules. 
Observe that safe subsets of output attributes always exist for all private modules --
one can always hide {\em all} the output attributes. They may incur higher cost
than that of an optimal subset of both input and
output attributes, but, in terms of privacy, by hiding only output
attributes one does not harm its maximum achievable privacy. 
In particular, it is not hard to see
that hiding all input attributes can give a maximum of
$\Gamma_1$-workflow-privacy, where $\Gamma_1$ is the size of the
range of the module. On the other hand hiding all output attributes
can give a maximum of $\Gamma_2$-workflow-privacy, where $\Gamma_2$ is
the size of the co-domain of the module, which can be much larger
than the actual range.
{\em We therefore focus in the rest of this paper on
safe subsets that contain only output attributes.}

\begin{figure}[t]
\centering
\subfloat 
{\label{fig:wf-propagation-need}
{\includegraphics[scale=0.25]{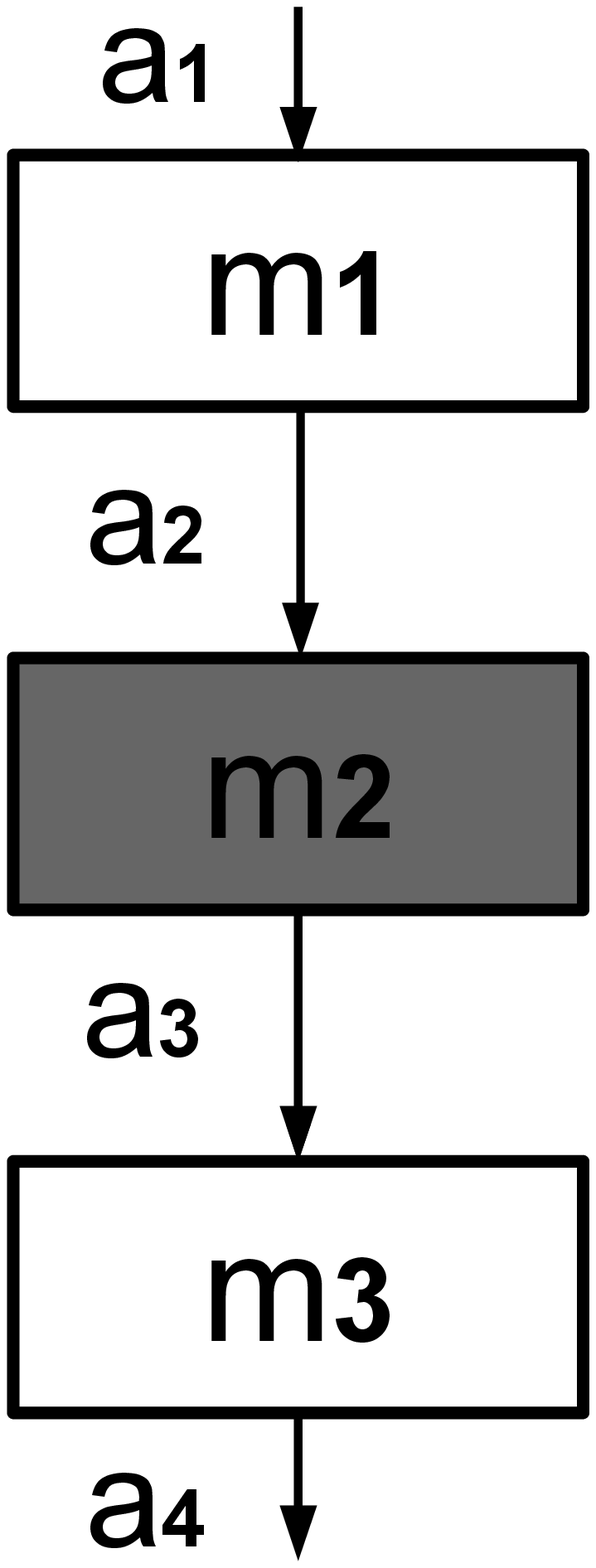} 
}
}~~~~~~~~
\subfloat [\vspace{-2mm}] 
{\label{fig:wf-example}
\includegraphics[scale=0.3]{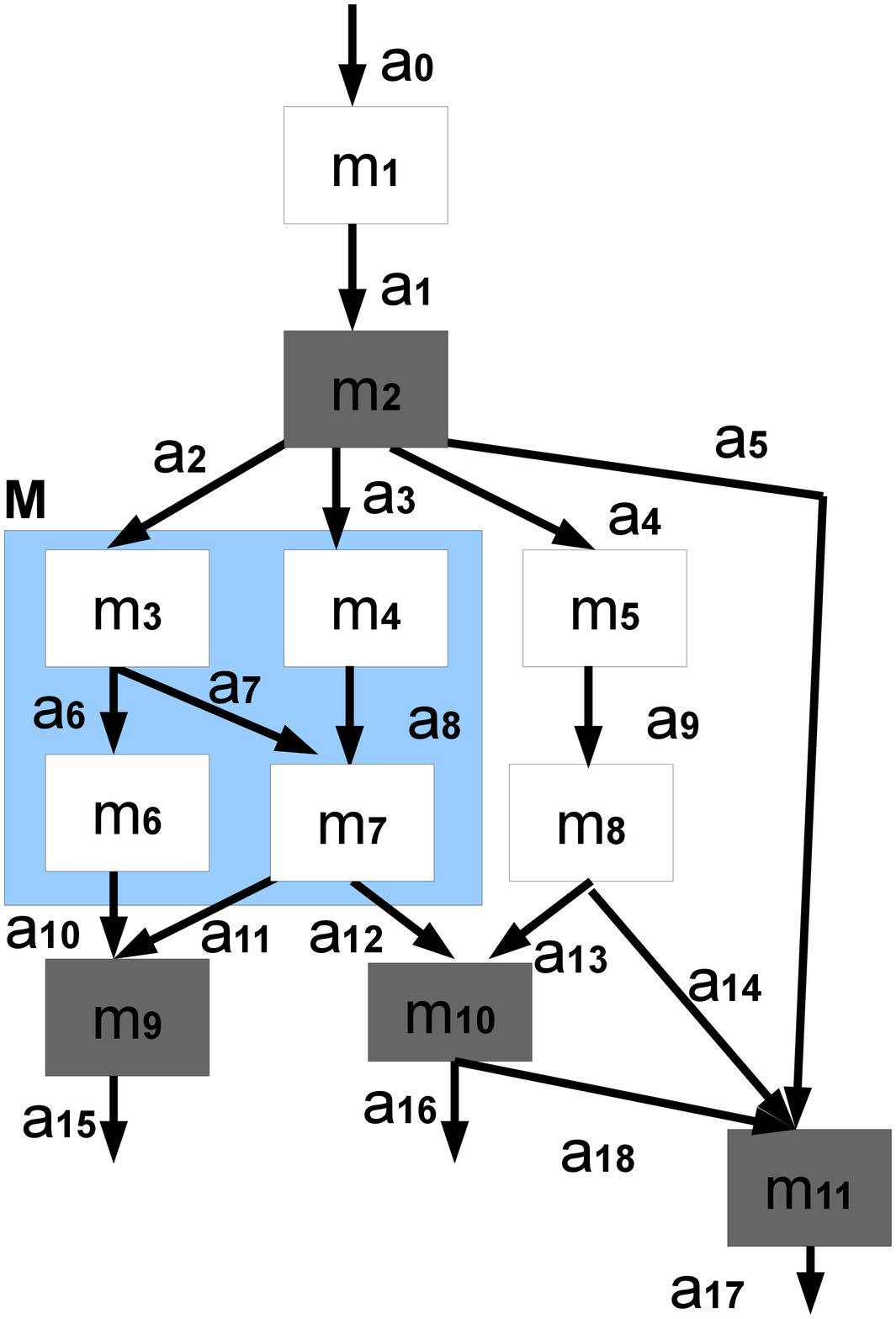}
}
\caption{(a) Propagation model, (b) A single-predecessor workflow. White modules are public, grey are private; the 
 box denotes
the composite module $M$ for 
$H_2 = \{a_3\}$.}
\label{fig:wf-example-propagation}
\vspace{-5mm}
\end{figure}

\subsection{Repeated 
Propagation} Consider again Example
\ref{eg:public-challenge}, and assume now that public module $m_3$
sends its output to another public module $m_4$ that  implements
an equality function (or a  one-one invertible function). Even if
the output of $m_3$ is hidden as described above, if the output of
$m_4$ remains visible, the privacy of $m_2$ is again jeopardized
since the output of $m_3$ can be inferred using the inverse
function of $m_4$. We thus need to propagate the attribute hiding
to $m_4$ as well. More generally, 
we need to propagate the attribute hiding 
repeatedly, through all adjacent public
modules, until we reach another private module.
\par
To formally define the {\em closure} of public modules to which
attributes hiding must be propagated, we use the
notion of 
a {\em public path}. Intuitively, there is a  
public path from a public module $m_i$ to a public
module $m_j$ if we can reach $m_j$ from $m_i$ by 
a path comprising only public modules.
In what follows, we define both directed and undirected public paths; recall that $A_i = I_i \cup O_i$ denotes the set of input
 and output attributes of module $m_i$.


\begin{definition}\label{def:publicpath}
A public module $m_1$ has \emph{\textbf{a directed (resp. an undirected) public
path}} to a public module $m_2$ if there is a sequence of public
modules $m_{i_1},m_{i_2}, \cdots, m_{i_j}$ such that $m_{i_1}=m_1$,
$m_{i_j}=m_2$, and for all $1\leq k < j$, $O_{i_k} \cap I_{i_{k+1}}
\neq \emptyset$ (resp. $A_{i_k} \cap A_{i_{k+1}} \neq \emptyset$).

%
%
%
\end{definition}

This notion naturally extends to module attributes. We say that an
input attribute $a \in I_1$ of a public module $m_1$ has an
(un)directed public path to a public module $m_2$ (and also to 
any output attribute $b \in O_2$), if there is an (un)directed public path from
$m_1$ to $m_2$.  The set of public modules to which attribute
hiding will be propagated can now be defined as follows.

\noindent
\begin{definition}\label{def:closure}
Given a private module $m_i$ and a set of hidden output 
attributes $h_i \subseteq O_i$ of $m_i$, the \emph{\textbf{public-closure}} $C(h_i)$ of
$m_i$ \wrt $h_i$ 
is the 
set of
public modules 
reachable from some attribute in $h_i$ by an undirected public path.

\eat{
			Given a private module $m_i$ and a set of hidden output attributes
			$h_i \subseteq O_i$ of $m_i$, the \emph{\textbf{public-closure}} $C(h_i)$ of
			$m_i$ \wrt $h_i$, 
			is the 
			set of
			public modules consisting of (1) all public modules $m_j$ s.t. $h_i
			\cap I_j \neq \emptyset$, and (2) all public modules to which there
			exists a undirected public path from the modules in $C(h_i)$.
}
\end{definition}

\eat{
\begin{figure}[ht!]
\centering
\includegraphics[scale=0.16]{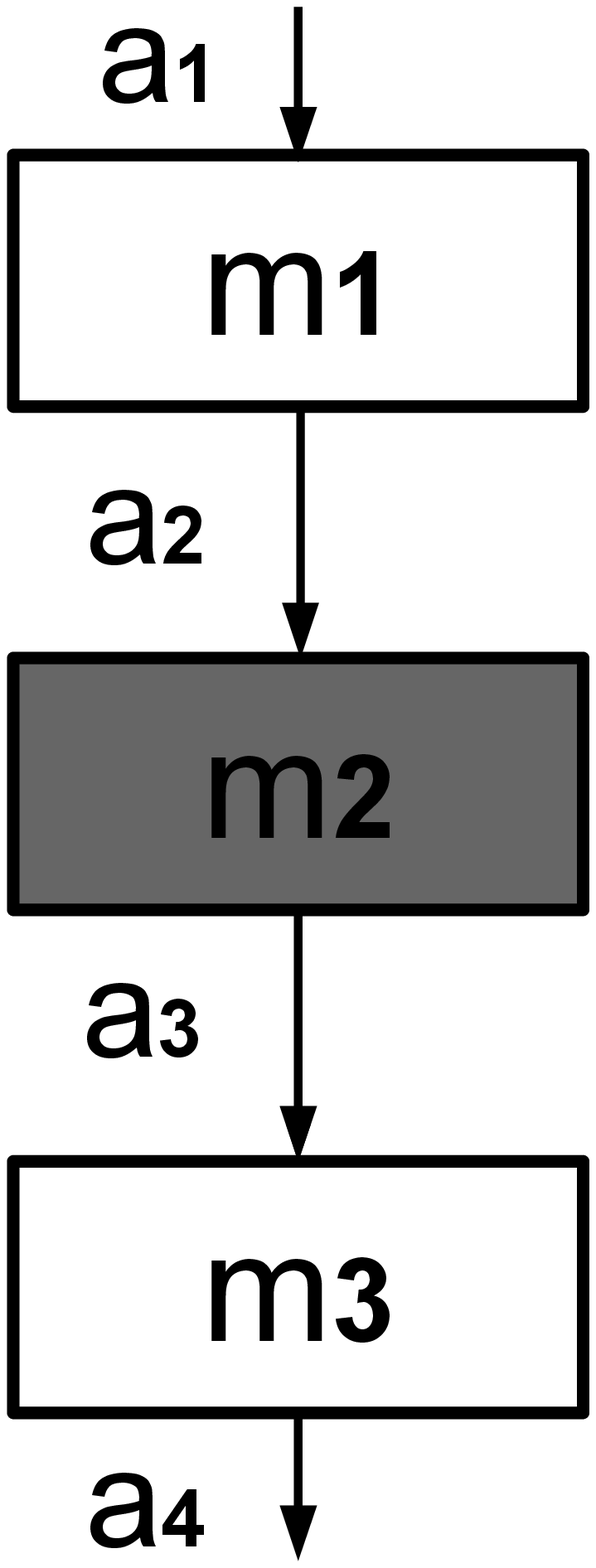}  
\includegraphics[scale=0.25]{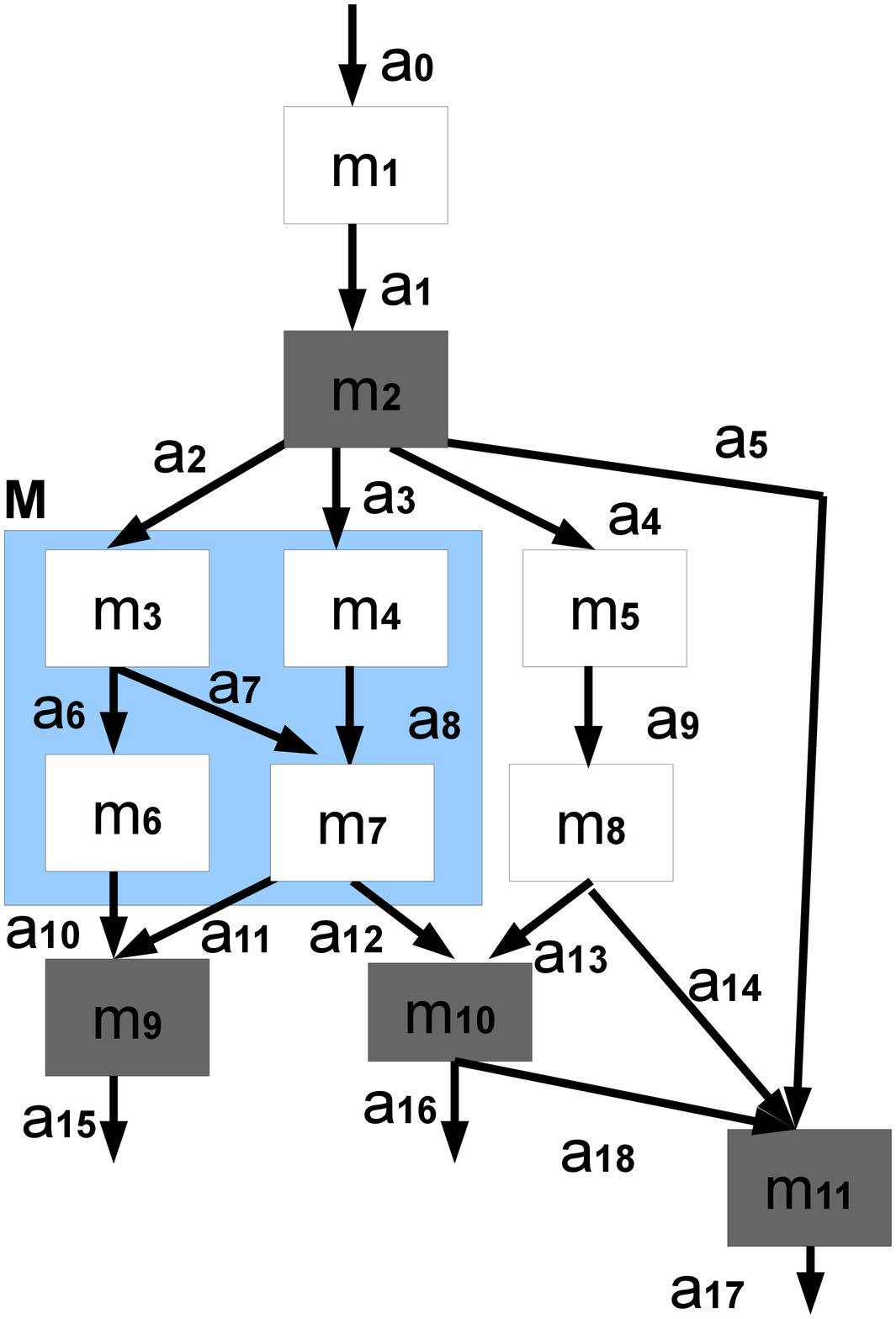}
\vspace{-3mm}
\caption{A single-predecessor workflow. White modules are public, grey are private; the 
box denotes
the composite module $M$ for 
$H_2 = \{a_2\}$.}
\vspace{-5mm}
\label{fig:wf-example}
\end{figure}
}

\begin{example} We illustrate these notions using Figure~\ref{fig:wf-example}.
 The public module $m_4$ has an undirected public path to the public
 module $m_6$ through the modules $m_7$ and  $m_3$. For private module $m_2$,
 if hidden output attributes $h_2$  
$= \{a_2\}$, $\{a_3\}$, or $\{a_2, a_3\}$, 
the public closure 
$C(h_2) = \{m_3, m_4, m_6, m_7\}$.
For 
$h_2 = \{a_4\}$, 
$C(h_2) = \{m_5, m_8\}$. In our subsequent analysis,                                                                                                     it will be convenient to view the public-closure as a virtual
\textbf{composite module} that encapsulates the sub-workflow and
behaves like it. For instance, the 
box in
Figure~\ref{fig:wf-example} denotes the composite module $M$
representing $C(\{a_2\})$, that has input attributes $a_2, a_3$, and
output attributes $a_{10}, a_{11}$ and $a_{12}$.
\end{example}

\subsection{Selection of hidden attributes}
 In Example
\ref{eg:public-challenge}, it is fairly easy to see which attributes
of $m_1$ or $m_3$ need to be hidden to preserve the privacy of
$m_2$. For the general case, where the public modules are not as simple
as equality functions,
 to determine which attributes of a given public
module need to be hidden we use the notions of {\em upstream} and {\em
downstream} safety. To define them we use the following notion of
tuple equivalence \wrt a given 
set of hidden attributes.
Recall that $A$ denotes the set of all attributes in the workflow; we also use
bold-faced letters $\tup{x}, \tup{y}, \tup{z},$ etc. to denote
tuples in the workflow or module relations with one or more attributes.

\begin{definition}\label{def:equiv}
Given two tuples $\tup{x}$ and $\tup{y}$ on a subset of attributes 
$B \subseteq A$, and a subset of 
hidden attributes $H \subseteq A$, we say that
$\tup{x} \equiv_H \tup{y}$ iff $\proj{B \setminus H}{\tup{x}} = \proj{B \setminus H}{\tup{y}}$.
\end{definition}

\begin{definition}\label{def:uds}
Given a subset of 
hidden attributes $H \subseteq A_i$ of a public module
$m_i$, $m_i$ is called
\begin{itemize}
\item \texttt{\textbf{downstream-safe}} (or, \textbf{\DS} in short) 
\wrt\ $H$ 
if for any two equivalent input tuples $\tup{x}, \tup{x'}$ to $m_i$ \wrt $H$, their outputs are
also equivalent:
        $$\left[\tup{x} \equiv_{H} \tup{x'} \right] \Rightarrow \left[m_i(\tup{x}) \equiv_{H} m_i(\tup{x'})\right],$$
        
\item \texttt{\textbf{upstream-safe}} (or, \textbf{\US} in short) 
\wrt\ $H$
if for any two equivalent outputs 
$\tup{y}, \tup{y'}$ of $m_i$
        \wrt $H$,  all of their preimages 
        are also equivalent:
        $$\left[(\tup{y} \equiv_{H} \tup{y'}) \wedge (m_i(\tup{x}) = \tup{y}, m_i(\tup{x'}) = \tup{y'})\right] \Rightarrow \left[\tup{x} \equiv_{H} \tup{x'}\right],$$
\item \texttt{\textbf{upstream-downstream-safe}} (or, \textbf{\UDS} in short) 
\wrt\ $H$ if it is both \US\ and \DS. 
\end{itemize}
\end{definition}

Note that if 
$H = A$ (i.e. all attributes are hidden) then
$m_i$ is clearly \UDS\ \wrt to $H$. We call this the \emph{trivial}
\UDS\ subset for $m_i$.


\begin{example}
Figure~\ref{fig:UDS} shows some example module relations. For an
(identity) module having relation $R_1$ in Figure~\ref{tab:m1-UDS},
the hidden subsets $\{a_1, a_3\}$ and $\{a_2, a_4\}$ are \UDS.
Note that $H = \{a_1, a_4\}$ is not a \UDS\ subset:  for tuples having the
same values of visible attribute $a_2$, say 0, the values of $a_3$ are not the same.
For a module having relation $R_2$ in Figure~\ref{tab:m2-UDS}, a \UDS\
hidden subset is $\{a_2\}$, but there is no 
\UDS\ subset that does not include 
$a_2$. It can also be checked that the module
$m_1$ 
in Figure~\ref{fig:m1} does not have any non-trivial \UDS\
subset.
\end{example}

\begin{figure}[t] 
{
\begin{center}
\centering \subfloat[{\scriptsize $R_1$}]{\label{tab:m1-UDS}
\begin{tabular}{|cc|cc|}
\hline $a_1$ & $a_2$ & $a_3$ & $a_4$ \\ \hline
\hline0 & 0 & 0 & 0 \\
\hline 0 & 1 & 0 & 1\\
\hline 1 & 0 & 1 & 0\\
\hline 1 & 1 & 1 & 1\\ \hline
\end{tabular}
}
~~~~~~~~~~~~
\subfloat[{\scriptsize $R_2$}]{\label{tab:m2-UDS}
\begin{tabular}{|cc|cc|}
\hline $a_1$ & $a_2$ & $a_3$ & $a_4$ \\ \hline
\hline 0 & 0 & 1 & 0\\
\hline 0 & 1 & 1 & 0\\
\hline 1 & 0 & 0 & 1 \\
\hline 1 & 1 & 0 & 1\\ \hline
\end{tabular}
}
\end{center}
 \caption{\UDS\ solutions for modules} \vspace{-5mm} \label{fig:UDS}
}
\end{figure}


The first question we attempt to answer is whether there is a composability
theorem analogous to Theorem~\ref{thm:privacy-private} that works in
the presence of public modules. 
In particular, we will show that for
a  class of workflows called \emph{single-predecessor
workflows} one can construct a private solution for the whole
workflow  by taking safe standalone solutions for the private modules, 
and then ensuring the \UDS\ properties of the public modules in the corresponding public-closure.
Next we define this class of workflows:


\begin{definition}\label{def:single-pred-wf}
A workflow $W$ is called a \emph{\textbf{single-predecessor workflow}}, if
\begin{enumerate}
    \item $W$ has no data-sharing, i.e. for  $m_i \neq m_j$,
    $I_i \cap I_j = \emptyset$, and,
    \item for every public module $m_j$  that belongs to a public-closure \wrt some output attribute(s) of a
    private module $m_i$, 
    $m_i$ is the only private module that has a directed public path to $m_j$
    (i.e. $m_i$ is the single private predecessor of $m_j). $
\end{enumerate}
\end{definition}



\begin{example}
Again consider Figure~\ref{fig:wf-example} which shows a
single-predecessor workflow. Modules $m_3, m_4$, $ m_6, m_7$ have
undirected public paths from $a_2 \in O_2$ (output attribute of
$m_2$), whereas $m_5, m_8$ have an 
undirected (also directed) public path from $a_4
\in O_2$; also $m_1$ is the 
single private-predecessor of
$m_3,...,m_8$ that has a directed path to each of module. 
The public module $m_1$ does not have any private
predecessor, but $m_1$ does not belong to the public-closure \wrt
the output attributes of any private module.
\end{example}

Although single-predecessor workflows are more restrictive than general
workflows, the above example illustrates that they can still capture
fairly intricate workflow structures, and more importantly, they can 
capture commonly found chain and tree workflows \cite{myexpt}.
Next in Section~\ref{sec:single-pred-wf}, we 
focus on single-predecessor workflows;
 then we explain in Section \ref{sec:general-wf} how general workflows can be handled.

\section{Single-Predecessor Workflows}\label{sec:single-pred-wf}
The main motivation behind the study of single-predecessor workflows is to
obtain a composability theorem similar to Theorem~\ref{thm:privacy-private}
combining solutions of standalone private and public modules. 
In Section~\ref{sec:privacy-theorems}, we show that
such a composability theorem indeed exists for this class of workflows.
Then we study how to optimally compose
the standalone solutions in Section~\ref{sec:optimization}.

\subsection{Composability Theorem for Privacy}\label{sec:privacy-theorems}

\eat{
			We start in Section~\ref{sec:privacy-single-pred} by proving
			Theorem~\ref{thm:privacy-downward} and explaining the
			role of the \UDS\ requirement. Then, in
			Section~\ref{sec:privacy-thm-eg}, we prove
			Proposition~\ref{prop:single-pred-needed} by illustrating the
			necessity of the restriction to single-predecessor workflows in
			Theorem~\ref{thm:privacy-downward}.
}

The following composability theorem says that, for each private module $m_i$,
it suffices to (i) find a safe hidden subset of output attributes 
(downstream propagation), (ii)
find a superset of these hidden attributes 
such that each public module in their public closure
is \UDS, and (iii) no attributes outside the public closure and $m_i$
are hidden (\ie\ no unnecessary hiding). 
Then union of these subsets of hidden attributes is workflow-private for each private module
in the workflow.
Theorem~\ref{thm:privacy-downward} stated below formalizes these three conditions.
\eat{
					Let $M^+$ be the set of indices for public modules
					($M^+ = \{i : m_i$ is public $\}$) and $M^-$ be the set of indices for private
					modules.
					 Recall that $I_i, O_i$ denote the subset of input and
					output attributes of module $m_i$
					and $A_i = I_i \cup O_i$; given a set of visible attributes $V_i \subseteq A_i$
					of $m_i$, the set of hidden attributes is $\widebar{V_i} = A_i \setminus V_i$
					(and \emph{not} $A \setminus V_i$ where $A$ is the set of all attributes in a workflow).
}
\begin{theorem}\label{thm:privacy-downward} ({\bf 
Composability Theorem
for Single-predecessor Workflows})
Let $W$ be a single-predecessor
workflow. For each private module $m_i$ in $W$, let  
$H_i$ be a subset of hidden attributes such that 
(i) $h_i = H_i \cap O_i$ is 
safe for $\Gamma$-standalone-privacy of $m_i$, 
(ii) each public module $m_j$ in the public-closure 
$C(h_i)$
is \UDS\ \wrt\ $A_j \cap H_i$, 
and (iii) $H_i \subseteq O_i \cup \bigcup_{j: m_j \in
C(h_i)} A_j$. Then the workflow $W$ is
$\Gamma$-private \wrt\ 
$H = \bigcup_{i: m_i \textrm{ is private}}H_i$.
\end{theorem}

First, in Section~\ref{sec:privacy-thm-eg}, we argue why the conditions and assumptions in the above theorem are necessary;
then we prove the theorem in Section~\ref{sec:privacy-single-pred}.

\eat{
				In other words, in a single-predecessor workflow, taking solutions
				for the standalone private modules, expanding them to the
				public-closure of the modules, following the \UDS\ safety
				requirements, then hiding the union of the attributes in these sets,
				guarantees $\Gamma$-workflow-privacy for all of the private modules.
}

%
\eat{
			The next section is dedicated to proving these two results. 
			In Section~\ref{sec:optimization} 
			we then show how this refined privacy theorem can be
			used to choose hidden attributes optimally in single-predecessor workflows.
			Solutions for
			workflows that are not single-predecessor are considered in Section
			\ref{sec:general-wf}.
}

\subsubsection{Necessity of the Assumptions in Theorem~\ref{thm:privacy-downward}}\label{sec:privacy-thm-eg}

Theorem~\ref{thm:privacy-downward} has two non-trivial conditions:
(1) the workflows are single-predecessor workflows, and (2)
the public modules in the public closure must be \UDS\ \wrt\ the
hidden subset; the third condition that there is no unnecessary data hiding
is required since the property \UDS ty of public modules is not valid with respect to set inclusion.
 The necessity of the first two conditions are discussed
in Propositions~\ref{prop:single-pred-needed} and \ref{prop:uds-need}
respectively.
\par
In the proof of these propositions 
we will consider the different possible worlds of the workflow view and
focus on the behavior (input-to-output mapping) $\widehat{m}_i$ of
the module $m_i$ as seen in these worlds. This may be different than
its true behavior recorded in the actual workflow relation $R$, and
we will say that $m_i$ is {\em redefined} as $\widehat{m}_i$ in the
given world. Note that $m_i$ and $\widehat{m}_i$, viewed as
relations, agree on the visible attributes of the the view but may
differ in the non visible ones.

\paragraph*{Necessity of Single-Predecessor Workflows}
The next proposition shows that single-predecessor workflows constitute
i
the largest class of workflows for which 
a composability theorem involving both public and private modules can succeed.

\begin{proposition}\label{prop:single-pred-needed}
There is a workflow $W$, which is not a single-predecessor workflow,
and a private module $m_i$ in $W$, where even hiding all output attributes
of $m_i$ and all attributes of all the public modules in $W$
does not give $\Gamma$-privacy for \emph{any} $\Gamma > 1$.
\end{proposition}

 \begin{proof}
By Definition~\ref{def:single-pred-wf}, a  workflow $W$ is
\emph{not} a single-predecessor workflow if one of  the following
holds: (i) there is a public
module $m_j$ in $W$ that belongs to a public-closure of a private module
$m_i$ but has no directed path from $m_i$, or, (ii) such a public
module $m_j$ has a directed path from more than one private module, or (iii)
$W$ has data sharing.
We now show an example for condition (i).  Examples for the remaining conditions can be found
in Appendix~\ref{sec:proof-prop:single-pred-needed}.

\eat{
We provide three example workflows where exactly one of the
violating conditions (i), (ii), (iii) holds, and
Theorem~\ref{thm:privacy-downward} does not hold in those workflows.
For space constraints, case (ii) is shown below whereas cases (1)
and (iii) are deferred to Appendix
\ref{sec:proof-prop:single-pred-needed}.
}


Consider the workflow $W_a$ in Figure~\ref{fig:wf-no-pred}.  
Here the public module $m_2$ belongs to the public-closure $C(\{a_3\})$ of $m_1$, but there is no directed public path from $m_1$ to $m_2$,
thereby violating the condition of single-predecessor workflows (though there is no data sharing).
Module functionality is as follows:
(i) $m_1$ takes $a_1$ as input and produces $a_3 = m_1(a_1) = a_1$.
(ii) $m_2$ takes $a_2$ as input and produces $a_4 = m_2(a_2) = a_2$.
(iii) $m_3$ takes $a_3, a_4$ as input and produces $a_5 = m_3(a_3, a_4) = a_3 \vee a_4$ (OR).
(iv) $m_4$ takes $a_5$ as input and produces $a_6 = m_4(a_5) = a_5$.
All attributes take values in $\{0, 1\}$.
%

Clearly, hiding output 
$\{a_3\}$ of $m_1$ gives $2$-standalone privacy.
We claim that hiding all output attributes of $m_1$ and all attributes of all public modules
(\ie\ $\{a_2, a_3, a_4, a_5\}$) gives only trivial 1-workflow-privacy for $m_1$,
although it satisfies the \UDS\ condition of $m_2, m_3$.
 To see this, consider the relation $R_a$ of all executions of  $W_a$ given in Table~\ref{tab:reln-no-pred},
 where the hidden values are in Grey. 
 The rows (tuples) here are numbered $r_1,\ldots,r_4$ for later reference.

\eat{
\begin{table}[ht]
\centering
\begin{tabular} {r|c|cccc|c|}
\hline
& $a_1$ & \cellcolor[gray]{0.6}$a_2$ & \cellcolor[gray]{0.6}$a_3$ & \cellcolor[gray]{0.6}$a_4$ & \cellcolor[gray]{0.6}$a_5$ & $a_6$ \\ \hline \hline
$r_1$ &  0 & \cellcolor[gray]{0.6}0 & \cellcolor[gray]{0.6}0 & \cellcolor[gray]{0.6}0 & \cellcolor[gray]{0.6}0 & 0 \\\hline
$r_2$ &  0 & \cellcolor[gray]{0.6}1 & \cellcolor[gray]{0.6}0 & \cellcolor[gray]{0.6}1 & \cellcolor[gray]{0.6}1 & 1 \\\hline
 $r_3$ & 1 & \cellcolor[gray]{0.6}0 & \cellcolor[gray]{0.6}1& \cellcolor[gray]{0.6}0 & \cellcolor[gray]{0.6}1 &  1 \\\hline
 $r_4$ & 1 & \cellcolor[gray]{0.6}1 & \cellcolor[gray]{0.6}1 & \cellcolor[gray]{0.6}1 &\cellcolor[gray]{0.6}1 & 1\\\hline
\end{tabular}
 \caption{Relation $R_a$ for workflow $W_a$ given in Figure~\ref{fig:wf-datashare-uds} (b)}
\label{tab:reln-no-pred}
\end{table}
}

\begin{table}[t]
{
\centering
\begin{tabular} {r|c|cccc|c|}
\cline{2-7}
& $a_1$ & \cellcolor[gray]{0.6}$a_2$ & \cellcolor[gray]{0.6}$a_3$ & \cellcolor[gray]{0.6}$a_4$ & \cellcolor[gray]{0.6}$a_5$ & $a_6$ \\ \hline \hline
$r_1$ &  0 & \cellcolor[gray]{0.6}0 & \cellcolor[gray]{0.6}0 & \cellcolor[gray]{0.6}0 & \cellcolor[gray]{0.6}0 & 0 \\ \hline
$r_2$ &  0 & \cellcolor[gray]{0.6}1 & \cellcolor[gray]{0.6}0 & \cellcolor[gray]{0.6}1 & \cellcolor[gray]{0.6}1 & 1 \\ \hline
 $r_3$ & 1 & \cellcolor[gray]{0.6}0 & \cellcolor[gray]{0.6}1& \cellcolor[gray]{0.6}0 & \cellcolor[gray]{0.6}1 &  1 \\ \hline
 $r_4$ & 1 & \cellcolor[gray]{0.6}1 & \cellcolor[gray]{0.6}1 & \cellcolor[gray]{0.6}1 &\cellcolor[gray]{0.6}1 & 1\\ \hline
\end{tabular}
 \caption{Relation $R_a$ for workflow $W_a$ given in Figure~\ref{fig:wf-no-pred}}  
\label{tab:reln-no-pred}
}
\end{table}

When $a_3$ is hidden, a possible candidate output
of input $a_1 = 0$ to $m_1$ is $1$. So we need to have a possible world
where $m_1$ is redefined as $\widehat{m}_1(0) = 1$. This would restrict $a_3$ to 1 whenever $a_1 = 0$.
But note that whenever $a_3 = 1$, $a_5 = 1$, irrespective of the value of $a_4$ ($m_3$ is an OR function).

This affects the 
rows $r_1$ and $r_2$ in $R$.
Both these rows must have $a_5 = 1$, however $r_1$ has $a_6 = 0$, and $r_2$ has $a_6 = 1$.
But this is impossible since, whatever the new definition $\widehat{m}_4$ of private module $m_4$ is,
it cannot map $a_5$ to both 0 and 1; $\widehat{m}_4$ must be a function and maintain the functional dependency $a_5 \rightarrow a_6$.
Hence all possible worlds of $R_a$ must map $\widehat{m}_1(0)$ to 0, and therefore $\Gamma = 1$.
\end{proof}

\eat{
\red{Mention that the restrictions arise mainly if there are private successors of the considered private module}
\begin{figure}[ht!]
\centering
\includegraphics[scale=0.2]{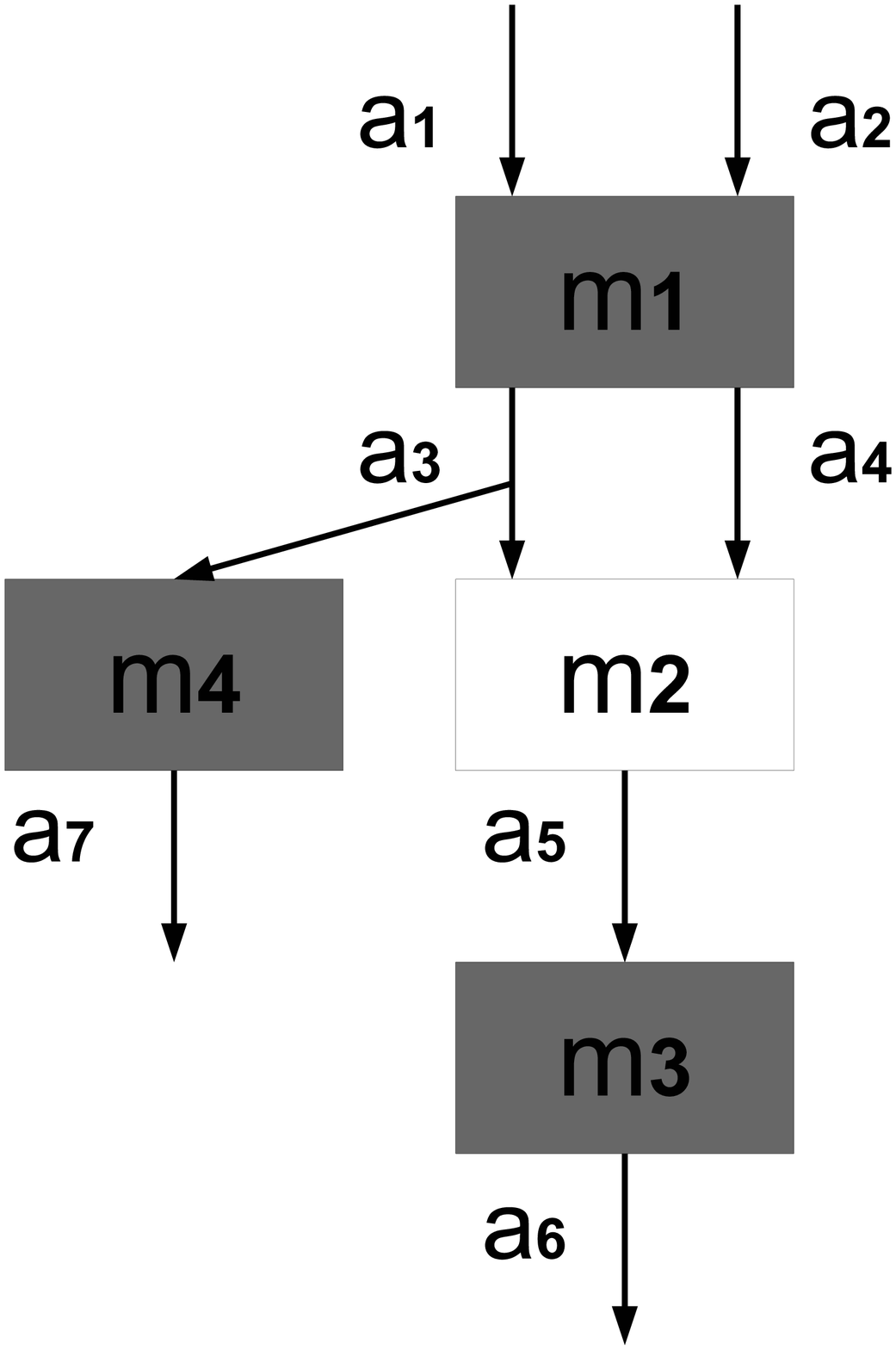} 
\includegraphics[scale=0.2]{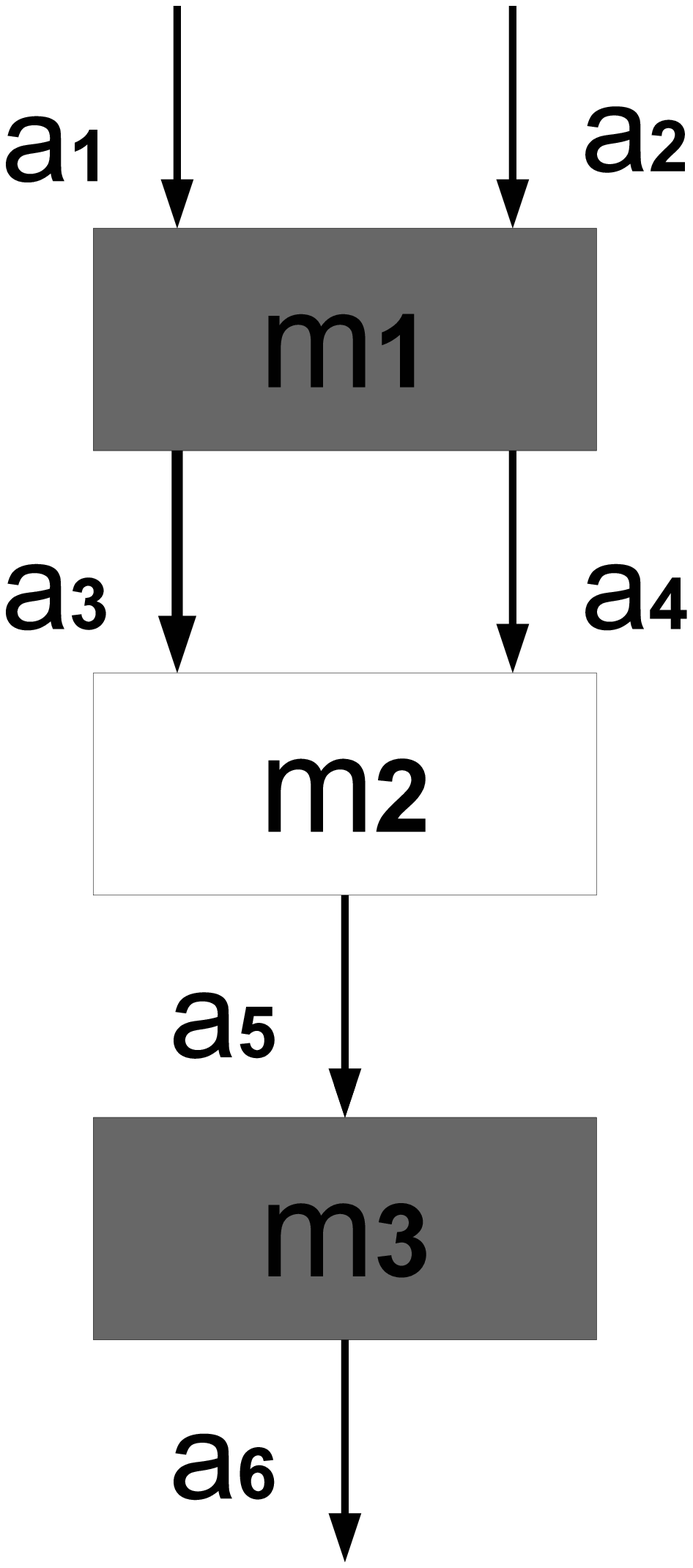} 
\caption{White modules are public, Grey are private (a) $a_3$ is shared as input to $m_2, m_3$,
(b) necessity of having \UDS\ condition.
}
\label{fig:wf-datashare-uds}
\end{figure}
}




\paragraph*{Necessity of \UDS ty for public modules} 
Example~\ref{eg:public-challenge} in the previous section motivated why the downward-safety condition is necessary and natural.
The following proposition illustrates the need for the  additional upward-safety condition in Theorem~\ref{thm:privacy-downward},
even when we consider downstream-propagation. 

\begin{figure}[t]
\centering
\subfloat[{
Workflow $W_a$}]{\label{fig:wf-no-pred}
\includegraphics[scale=0.25]{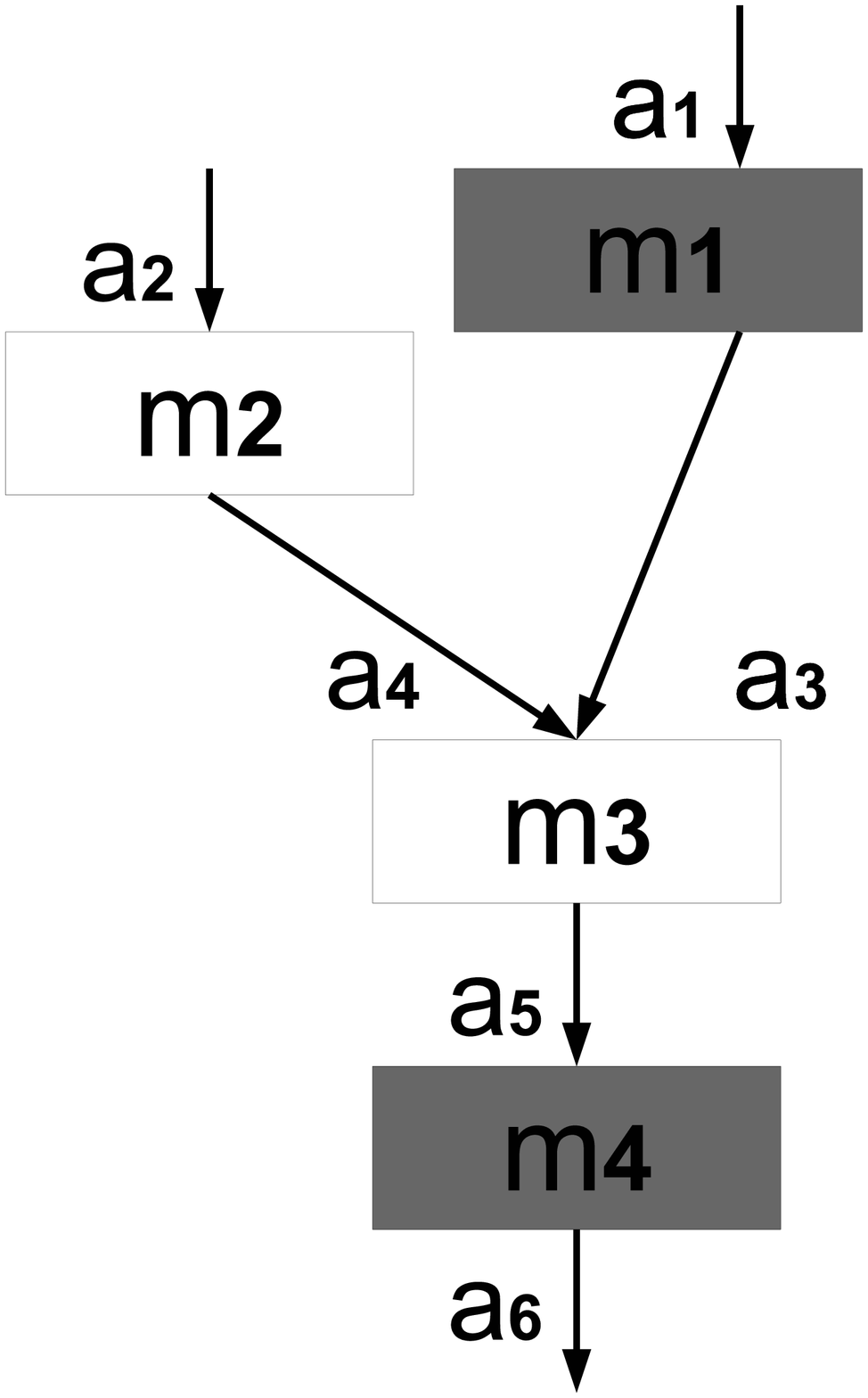}
}~~~~~
\subfloat[{
Workflow $W_b$}]{\label{fig:wf-uds}
\includegraphics[scale=0.25]{wf-uds-need.eps} 
}
\caption{Necessity of the conditions in Theorem~\ref{thm:privacy-downward}: (a) Single-predecessor workflows,
(b) \UDS ty for public modules; White modules are public, grey are private.}
\end{figure}

\begin{proposition}\label{prop:uds-need}
There is a workflow $W$ with a private module $m_i$, and a safe subset of hidden 
 attributes $h_i$ 
guaranteeing $\Gamma$-standalone-privacy for $m_i$ 
($\Gamma > 1$),
such that satisfying only the downstream-safety condition for the public modules in 
$C(h_i)$ does not give $\Gamma$-workflow-privacy for $m_i$ for \emph{any} $\Gamma > 1$.

\end{proposition}


\eat{ In the proof of this proposition and in other places in the
paper we will frequently use \emph{``$\widehat{m}_i$ as redefined
$m_i$''} while constructing a possible world for the workflow
relation $R$.
 By the definition of possible world (Definition~\ref{def:pos-worlds-workflow}),
every possible world $R'$ of $R$ must satisfy all the functional dependencies $I_i \rightarrow O_i$, $i = 1$ to of $n$
(and the projection of $R, R'$ on the attributes of public module and on the visible attributes should be the same).
Hence it suffices to redefine the private module $m_i$-s to
$\widehat{m}_i$-s, keep public modules $m_j$-s unchanged
($\widehat{m}_j = m_j$), and then showing that the relation $R'$
formed by the join of the standalone relations of $\widehat{m}_i$-s
is a possible world provided $R, R'$ have the same projections on
visible attributes. The redefined modules $\widehat{m_i}$-s, also
being functions, will satisfy the dependencies $I_i \rightarrow
O_i$, so their join $R$ will satisfy the functional dependencies. }

\begin{proof} 
Consider the chain workflow $W_b$ 
given in Figure~\ref{fig:wf-uds}
with three modules $m_1, m_2, m_3$ defined as follows.
(i) $(a_3, a_4) = m_1(a_1, a_2)$ where $a_3 = a_1$ and $a_4 = a_2$, 
(ii) $a_5 = m_2(a_3, a_4) = a_3 \vee a_4$ (OR),
(iii) $a_6 = m_3(a_5) = a_5$.
$m_1, m_3$ are private whereas $m_2$ is public.
All attributes take values in $\{0, 1\}$.
Clearly hiding output $a_3$ of $m_1$ gives $\Gamma$-standalone privacy for $\Gamma = 2$.
Now suppose $a_3$ is hidden in the workflow.
Since 
$m_2$ is public (known to be OR function), $a_5$ must be hidden (downstream-safety condition). Otherwise
from visible output $a_5$ and input $a_4$, some values of hidden input $a_3$ can be uniquely
determined (eg. if $a_5 = 0, a_4 = 0$, then $a_3 = 0$ and if $a_5 = 1, a_4 = 0$, then $a_3 = 1$).
On attributes $(a_1, a_2, \underline{a_3}, a_4, \underline{a_5}, a_6)$,
the original relation $R$ 
is shown in Table~\ref{tab:reln-uds}
(the hidden attributes and their values are underlined in the text and in grey in the table).

\begin{table}[ht]
{
\centering
\begin{tabular} {|cc|c|c|c|c|}
\hline
$a_1$ & $a_2$ & \cellcolor[gray]{0.6}$a_3$ & $a_4$ & \cellcolor[gray]{0.6}$a_5$ & $a_6$ \\ \hline \hline
0 &  0 & \cellcolor[gray]{0.6}0  & 0 & \cellcolor[gray]{0.6}0 & 0\\\hline
0 & 1 & \cellcolor[gray]{0.6} 0 & 1 & \cellcolor[gray]{0.6}1 & 1\\\hline
1 &  0 & \cellcolor[gray]{0.6}1 &  0 & \cellcolor[gray]{0.6}1 & 1\\\hline
1 &1 & \cellcolor[gray]{0.6} 1 &1 &\cellcolor[gray]{0.6}1 & 1\\\hline
\end{tabular}
\caption{Relation $R$ for workflow given in Figure~\ref{fig:wf-uds}}
\label{tab:reln-uds}
}
\end{table}


Let us first consider an input $(0,0)$ to $m_1$.  When $a_3$ is
hidden, a possible candidate output $y$ of input tuple $x = (0, 0)$
to $m_1$ is $(\underline{1}, 0)$. So we need to have a possible
world where $m_1$ is redefined as $\widehat{m}_1(0, 0) = (1, 0)$. To be
consistent on the visible attributes, this forces us to redefine
$m_3$ to $\widehat{m}_3$ where $\widehat{m}_3(1) = 0$; otherwise the row $(0, 0,
\underline{0}, 0, \underline{0}, 0)$ in $R$ changes to $(0, 0,
\underline{1}, 0, \underline{1}, 1)$. This in turn forces us to
define $\widehat{m}_1(1, 0) = (0, 0)$ and $\widehat{m}_3(0) = 1$. (This is because if we
map $\widehat{m}_1(1, 0)$ to any of $\{(1, 0), (0, 1), (1, 1)\}$, either we
have inconsistency on the visible attribute $a_4$, or $a_5 = 1$, and
$\widehat{m}_3(1) = 0$, which gives a contradiction on the visible attribute
$a_6 = 1$.)

Now consider the input $(1, 1)$ to $m_1$. For the sake of consistency on the visible attribute $a_3$, $\widehat{m}_1(1, 1)$
can take value $(1, 1)$ or $(0, 1)$. 	But if $\widehat{m}_1(1, 1) = (1, 1)$ or $(0, 1)$, we have an inconsistency on the visible attribute $a_6$.
For this input in the original relation $R$, $a_5 = a_6 = 1$.
Due to the redefinition of $\widehat{m}_3(1) = 0$, we have inconsistency on $a_6$.
But note that the downstream-safety condition has been satisfied so far by hiding $a_3$ and $a_5$.
To have consistency on the visible attribute $a_6$ in the row $(1, 1, \underline{1}, 1, \underline{1}, 1)$,
we must have $a_5 = 0$ (since $\widehat{m}_3(0) = 1$).  The pre-image of $a_5 = 0$ is $a_3 = 0, a_4 = 0$,
hence we have to redefine $\widehat{m}_1(1, 1) = (\underline{0}, 0)$. But $(\underline{0}, 0)$ is not equivalent to
original $m_1(1, 1) = (\underline{1}, 1)$ \wrt the visible attribute $a_4$. So the only solution
in this case for $\Gamma > 1$, assuming that we do not hide output $a_6$ of private module $m_3$,  is to hide $a_4$, which makes the public module $m_2$ both upstream and downstream-safe.
\end{proof}

\noindent
This example also suggests that upstream-safety is needed
only when a private module gets input from 
a module in the public-closure. 
We will see later the proof of Lemma~\ref{lem:main-private-single-pred} (Section~\ref{sec:privacy-single-pred}) that this is indeed the case.

\subsubsection{Proof of 
Composability Theorem}\label{sec:privacy-single-pred}
To prove $\Gamma$-privacy, we need to show the existence of 
at least $\Gamma$ possible outputs for each input to each private module,
originating from the possible worlds of the workflow relation \wrt the visible attributes.
First we present a crucial lemma, which shows the existence of many possible outputs for 
any fixed input to any fixed private module in the workflow, 
when the conditions in Theorem~\ref{thm:privacy-downward} are satisfied.
In particular, this lemma 
shows that any candidate output for a given input for standalone privacy remains a candidate output 
for  workflow-privacy, even when the private module interacts with
other private and  public module in a (single-predecessor) workflow.
Therefore, if there are $\geq \Gamma$ candidate outputs for standalone-privacy, there will be
$\geq \Gamma$ candidate outputs for workflow-privacy.
Later in this section we will formally prove Theorem~\ref{thm:privacy-downward}
using this lemma.

\begin{lemma}\label{lem:main-private-single-pred}
Consider a standalone private module $m_i$, a set of 
hidden attributes $h_i$, any input $\tup{x}$ to $m_i$, 
and any candidate output $\tup{y} \in \Out_{\tup{x}, m_i, h_i}$ of $\tup{x}$.
Then $\tup{y} \in \Out_{\tup{x}, W, H_i}$ when $m_i$ belongs to a single-predecessor workflow $W$,  
and a set attributes $H_i \subseteq A$ is hidden 
such that (i) 
$h_i \subseteq H_i$, (ii) only output attributes from $O_i$ are included in 
$h_i$ (i.e. $h_i \subseteq O_i$),
and (iii) every module $m_j$ in the public-closure $C(h_i)$ is \UDS\ \wrt $A_j \cap H_i$.
\end{lemma}

\eat{
			\begin{lemma}\label{lem:main-private-single-pred}
			Consider a single-predecessor workflow $W$,  any private module $m_i$ in $W$, a set of visible attributes $V_i$,
			any input $\tup{x} \in \proj{I_i}{R}$; and any candidate output $\tup{y} \in \Out_{\tup{x}, m_i, V_i}$. 
			Given a set of hidden attributes $H_i \subseteq A$,
			such that (i) the hidden attributes $\widebar{V_i} \subseteq H_i$, 
			(ii) only output attributes from $O_i$ are included in $\widebar{V_i}$ (i.e. $\widebar{V_i} \subseteq O_i$),
			and (iii) every module $m_j$ in the public-closure $C(\widebar{V_i})$ is \UDS\ \wrt $A_j \setminus H_i$.
			Then $\tup{y} \in \Out_{\tup{x}, W, V}$ 
			where $V = A \setminus H_i$.
			\end{lemma}
}

To prove the lemma, we will (arbitrarily) fix a private module $m_i$, 
an input $\tup{x}$ to $m_i$, a  hidden subset $h_i$, and a candidate output $\tup{y} \in \Out_{\tup{x}, m_i, h_i}$
for $\tup{x}$.
The proof 
comprises two steps:

\begin{enumerate}
    \item[(Step-1)] Consider the connected subgraph $C(h_i)$
    as a single \emph{composite} public module $M$, or equivalently assume that
    $C(h_i)$ contains a single public module.
By the properties of single-predecessor workflows, $M$ gets all its
inputs from $m_i$, but can send its outputs to 
one, multiple, or  zero (for final output) private
modules. 
Let $I$  (respectively $O$) be the input (respectively output) attribute sets of $M$.
In
Figure~\ref{fig:wf-example}, the 
 box is $M$, $I = \{a_2, a_3\}$
and $O = \{a_{10}, a_{11}, a_{12}, a_{13}\}$. We argue that when $M$
is \UDS\ \wrt visible attributes $(I \cup O) \cap H_i$, and the
other conditions of Lemma~\ref{lem:main-private-single-pred} are
satisfied,
then 
$\tup{y} \in \Out_{\tup{x}, W, H_i}$.

\item[(Step-2)] We show that if every public module in the composite module $M = C(h_i)$
is \UDS, then $M$ is \UDS. To continue with our example, in
Figure~\ref{fig:wf-example}, assuming that $m_3$, $m_4$, $m_6$, $m_7$ are
\UDS\ \wrt the hidden attributes, we have to show that $M$ is \UDS.
\end{enumerate}

\begin{figure}[t]
\centering
\subfloat[{\scriptsize 
Original modules $m_1, m_3$}]{\label{fig:wf-proof-lemma}
\includegraphics[scale=0.3]{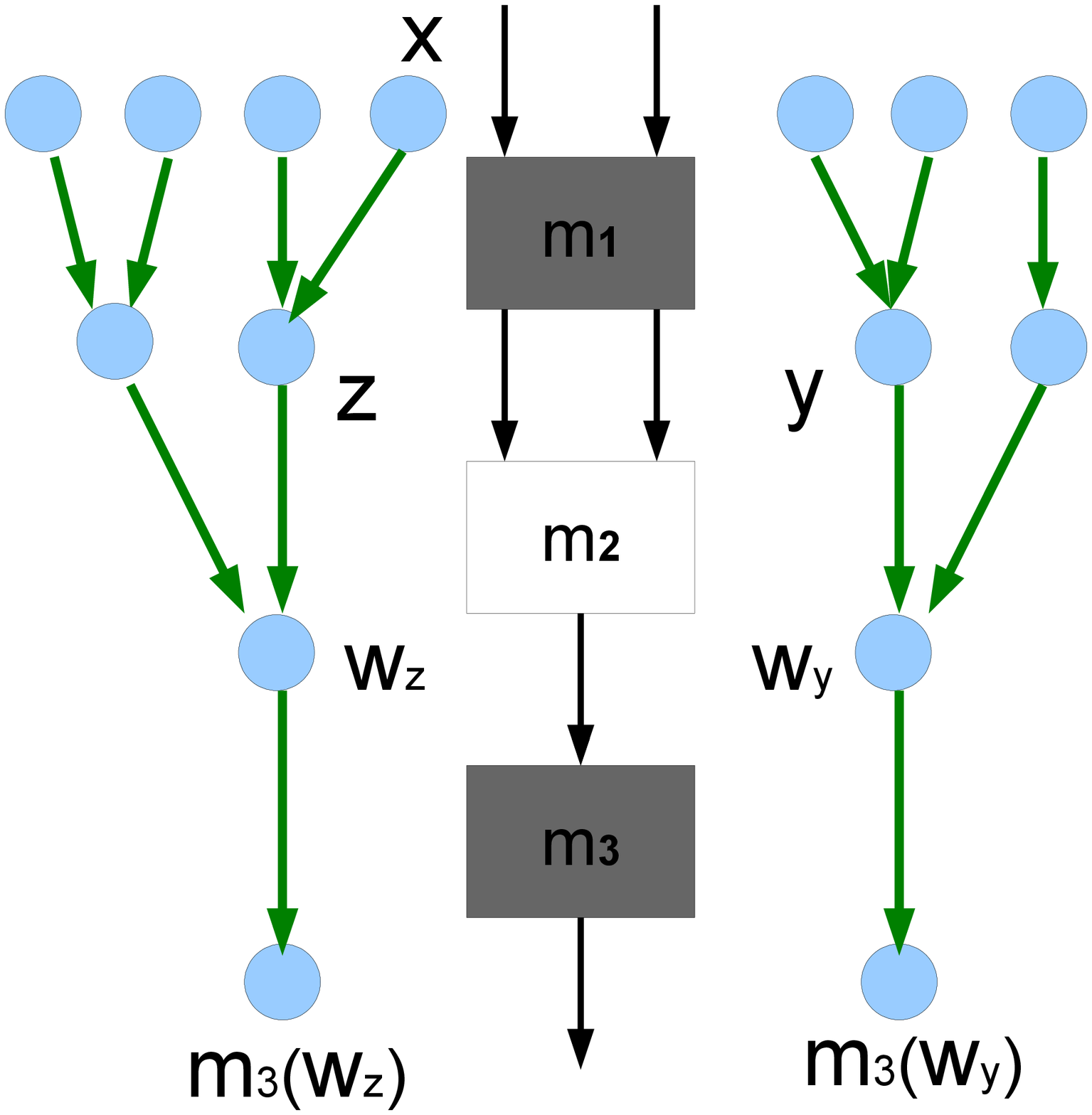}
}~~~~~~~~~~~~~~~
\subfloat[{\scriptsize 
Redefined $\widehat{m_1}, \widehat{m_3}$}]{\label{fig:wf-proof-lemma-modified}
\includegraphics[scale=0.3]{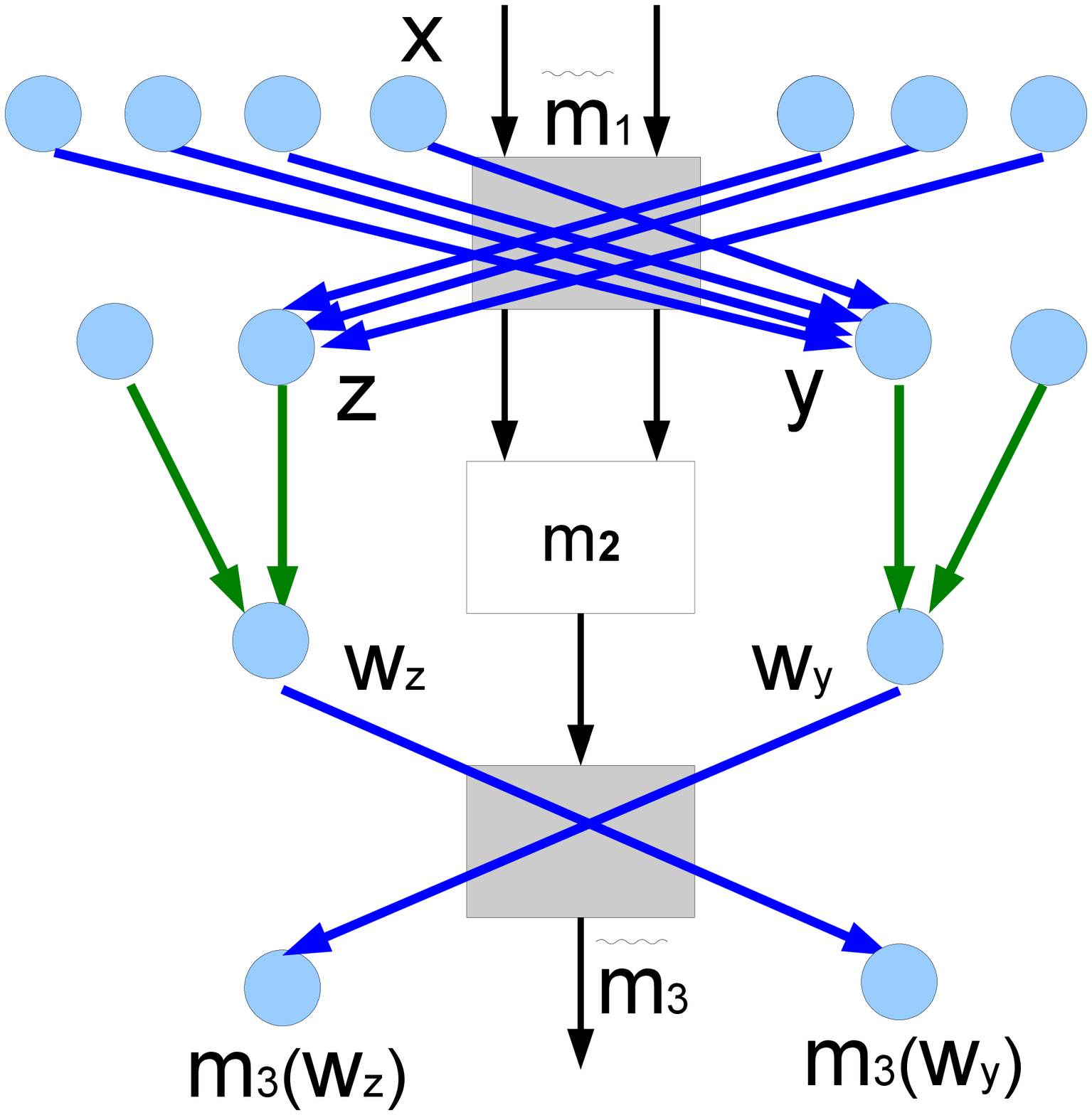} 
}
\caption{Illustration of Example~\ref{eg:main-lemma}: Input-output relationship in (a) original workflow,
(b) possible world mapping $\tup{x}$ to $\tup{y}$.}
\end{figure}

\noindent
\textbf{Proof of Step-1.~~}
The proof of Lemma~\ref{lem:main-private-single-pred} is involved
even for the restricted scenario in Step-1, in which $C(h_i)$
contains a single public module.  Due to space constraints, the proof is
given in  Appendix~\ref{sec:proof-lem:main-private-single-pred}, and we illustrate
here the key ideas using a simple example of a
chain workflow.

\begin{example}\label{eg:main-lemma}
Consider a chain workflow, for instance, the one given in Figure~\ref{fig:wf-uds} with the relation
in Table~\ref{tab:reln-uds}. Fix module $m_i = m_1$. 
Hiding its output $h_1 = \{a_3\}$ gives $\Gamma$-standalone-privacy for $\Gamma = 2$.
Fix input $\tup{x} = (0, 0)$, with original
output $\tup{z} = m_1(\tup{x}) = (\ul{0}, 0)$ (hidden attribute $a_3$ is underlined).
Also fix a candidate output 
$\tup{y} = (\ul{1}, 0) \in \Out_{\tup{x}, m_1, h_1}$. 
Note that $\tup{y}$ and $\tup{z}$ are equivalent on the visible attribute $\{a_4\}$.

First, consider the simpler case when $m_3$ does not exist, \ie\ $W$ contains only two modules $m_1, m_2$,
and the column for $a_6$ does not exist in Table~\ref{tab:reln-uds}.
As we mentioned before, when the composite public module does not have any private successor, we only need the downstream-safety property for
modules in $C(h_i)$; in this case,  $C(h_i)$ comprises a single public module, $m_2$.
We construct a possible world $R'$ of
$R$ by redefining module $m_1$ to $\widehat{m}_1$ as follows:
$\widehat{m}_1$ simply maps all pre-images of $\tup{y}$ to $\tup{z}$,
and all pre-images of $\tup{z}$ to $\tup{y}$. In this case, both $\tup{y}, \tup{z}$ have single pre-image.
So $\tup{x} = (0, 0)$ gets mapped to $(\ul{1}, 0)$ and input $(1, 0)$ gets mapped to $(\ul{0}, 0)$.
To make $m_2$ 
downstream-private, we 
hide output $a_5$ of $m_2$.
Therefore, the set of hidden attributes $H_1 = \{a_3, a_5\}$.
Finally $R'$ is formed by the join of relations for $\widehat{m}_1$ and $m_2$.
Note that the projection of $R, R'$, will be the same on visible attributes $a_1, a_2, a_4$
(in $R'$, the first row will be $(0, 0, \ul{1}, 0, \ul{0})$ and the third row will be $(1, 0, \ul{0}, 0, \ul{0})$).
\par
Next consider the more complicated case, when the modules in
$C(h_i)$ have private successors (in this example, 
when the private module $m_3$ is present). We already argued in the proof of
Proposition~\ref{prop:uds-need} that we also need to hide the input
$a_4$ to ensure workflow privacy for $\Gamma > 1$ (\UDS ty). 
Let us now describe the proof strategy when $a_4$ is hidden, \ie\ $H_1 = \{a_3, a_4, a_5\}$. 
\par
Let $\tup{w_y} =
m_2(\tup{y})$ and $\tup{w_z} = m_2(\tup{z})$ (see Figure~\ref{fig:wf-proof-lemma}). 
We redefine $m_1$ to
$\widehat{m}_1$ as follows (see Figure~\ref{fig:wf-proof-lemma-modified}). For all input $\tup{u}$ to $m_1$ such that
$\tup{u} \in m_1^{-1}m_2^{-1}(\tup{w_{z}})$ (respectively $\tup{u} \in
m_1^{-1}m_2^{-1}(\tup{w_{y}})$), we define $\widehat{m}_1(\tup{u}) = \tup{y}$
(respectively $\widehat{m}_1(\tup{u}) = \tup{z}$).
Note that the mapping of tuples $\tup{u}$
that are not necessarily 
$m_1^{-1}(\tup{y})$ or $m_1^{-1}(\tup{z})$ are being redefined  under $m_1$ (see Figure~\ref{fig:wf-proof-lemma-modified}). 
For $\widehat{m}_3$, we define,
$\widehat{m}_3(\tup{w_y}) =m_3(\tup{w_z})$ and $\widehat{m}_3(\tup{w_z}) =
m_3(\tup{w_y})$. 
Recall that $\tup{y} \equiv_{H_1} \tup{z}$ ($\tup{y}, \tup{z}$ have the same values of visible attributes).
 Since $m_2$ is downstream-safe $\tup{w_y} \equiv_{H_1} \tup{w_z}$.
  Since $m_2$ is also upstream-safe, for all input
$\tup{u}$  to $m_1$ that are being redefined by $\widehat{m}_1$, their 
images
under $m_1$ are equivalent \wrt $H_1$ (and therefore with $\tup{y}$ and $\tup{z}$). 
In our example, $\tup{w_y} =
m_2(\ul{1}, \ul{0}) = (\ul{1})$, and $\tup{w_z} = m_3(\ul{0}, \ul{0}) = (\ul{0})$.
$m_1^{-1}m_2^{-1}(\tup{w_{z}}) = \{(0, 0)\}$ and
$m_1^{-1}m_2^{-1}(\tup{w_{y}}) = \{(0, 1)$, $(1, 0)$, $(1, 1)\}$. So
$\widehat{m}_1$ maps $(0, 0)$ to $(\ul{1}, \ul{0})$ and all of $\{(0, 1)$, $(1, 0)$, $(1,
1)\}$ to $(\ul{0}, \ul{0})$; $\widehat{m}_3$ maps $(\ul{0})$ to $(1)$ and $\ul{1}$ to $(0)$. 
\par
Consider the relation $R'$ formed by joining the relations of $\widehat{m}_1$, $m_2$, $\widehat{m}_3$
(see Table \ref{tab:relation}).
The relation $R'$ has the same projection on visible
attributes $\{a_1, a_2, a_6\}$ as $R$ in Table~\ref{tab:reln-uds},
and the public module $m_2$ is unchanged. So $R'$ is a possible
world of $R$ that maps $\tup{x} = (0, 0)$ to $\tup{y} = (1, 0)$ as
desired, i.e. $\tup{y} \in \Out_{\tup{x}, W, H_1}$. \hfill $\Box$

\begin{table}[t]
{
\centering
\begin{tabular} {|cc|ccc|c|}
\hline
$a_1$ & $a_2$ & \cellcolor[gray]{0.6}$a_3$ & \cellcolor[gray]{0.6}$a_4$ & \cellcolor[gray]{0.6}$a_5$ & $a_6$ \\ \hline \hline
0 &  0 & \cellcolor[gray]{0.6} 1  & \cellcolor[gray]{0.6}0 & \cellcolor[gray]{0.6}1 & 0\\ \hline
0 & 1 & \cellcolor[gray]{0.6} 0 & \cellcolor[gray]{0.6}0 & \cellcolor[gray]{0.6}0 & 1\\ \hline
1 &  0 & \cellcolor[gray]{0.6}0 &  \cellcolor[gray]{0.6}0 & \cellcolor[gray]{0.6}0 & 1\\ \hline
1 &1 & \cellcolor[gray]{0.6} 0 & \cellcolor[gray]{0.6}0 &\cellcolor[gray]{0.6}0 & 1\\ \hline
\end{tabular}
\caption{Relation $R'$, a possible world of the relation $R$ for the workflow in Figure~\ref{fig:wf-uds} \wrt $H_1 = \{a_3, a_4, a_5\}$.} \label{tab:relation}
}
\end{table}

\end{example}

The argument for more general single-predecessor workflows, like the one given in Figure~\ref{fig:wf-example},
is more complex. Here a private module (like $m_{11}$) can get inputs from $m_i$ (in Figure~\ref{fig:wf-example}, $m_2$), from its public-closure
 $C(h_i)$ (in the figure, $m_8$), and also from the private successors of the modules in $C(h_i)$
 (in the figure, $m_{10}$).
 In this case, the tuples $\tup{w_y}, \tup{w_z}$ are not well-defined, and redefining the private modules is
 more complex. In the 
 proof  of the lemma we give the formal argument using an \emph{extended flipping function}, that selectively
 changes part of inputs and outputs of the private module based on their connection with the private module $m_i$
considered in the lemma.\\


\noindent
\textbf{Proof of Step-2.~~}
The following lemma formalizes the claim in Step-2:
\begin{lemma}\label{lem:composite-UDS}
Let $M$ be a composite module consisting only of public modules.
Let $H$ be a subset of hidden
attributes such that every public module $m_j$ in $M$ is \UDS\ \wrt
$A_j \cap H$. Then $M$ is \UDS\ \wrt $(I \cup O) \cap H$.
\end{lemma}

\begin{proof}[Sketch]
The formal proof of this lemma is given in
Appendix~\ref{sec:proof-lem:composite-UDS}. We sketch here the main
ideas. To prove the lemma, we show that if every module in the
public-closure is downstream-safe (respectively upstream-safe), then $M$ is
downstream-safe (respectively upstream-safe). For downstream-safety, we
consider the modules in $M$ in topological order, say $m_{i_1},
\cdots, m_{i_k}$ (in Figure~\ref{fig:wf-example}, $k = 4$ and the
modules in order may be $m_3, m_6, m_4, m_7$). Let $M^j$ be the
(partial) composite public module formed by the union of modules
$m_{i_1}, \cdots, m_{i_j}$, and let $I^j, O^j$ be its input and
output (the attributes that are either from a module not in $M^j$ to
a module in $M^j$, or to a module not in $M^j$ from a module in
$M^j$. Clearly, $M^{1} = \{m_{i_1}\}$ and $M^k = M$. Then by
induction from $j = 1$ to $k$, we show that $M^j$ is downstream-safe
\wrt $(I^j \cup O^j) \cap H$ if all of $m_{i_\ell}$, $1 \leq
\ell \leq j$ are downstream-safe \wrt $(I_{i_\ell} \cup
O_{i_{\ell}}) \cap H = A_{i_{\ell}} \cap H$. For
upstream-safety, we consider the modules in \emph{reverse
topological order}, $m_{i_k}, \cdots, m_{i_1}$, and give a similar
argument by an induction on $j = k$ down to 1.
\end{proof}


\noindent
\paragraph*{Proof of Theorem~\ref{thm:privacy-downward}}
Now we complete the proof of Theorem~\ref{thm:privacy-downward} using Lemma~\ref{lem:main-private-single-pred}.
\begin{proof}[ of Theorem~\ref{thm:privacy-downward}]
We first argue that if $H_i$ satisfies the conditions in Theorem~\ref{thm:privacy-downward}
then $H_i' = \bigcup_{\ell: m_{\ell} \textrm{ is private}} H_{\ell}$ satisfies the
conditions in Lemma~\ref{lem:main-private-single-pred}.
Since $h_i = H_i \cap O_i$, (i) $h_i \subseteq H_i \subseteq \bigcup_{\ell: m_{\ell} \textrm{ is private}}H_{\ell} = H_i'$ ; 
and (ii)  $h_i \subseteq O_i$. Next we argue that the third requirement in the lemma,
(iii) every module $m_j$ in the public-closure
$C(h_i)$ is \UDS\ \wrt $H_i' \cap A_j$, also holds.
\par
To see (iii),
observe that the Theorem~\ref{thm:privacy-downward} has an additional condition on $H_i$:
$H_i \subseteq O_i \cup \bigcup_{j: m_j \in
C(h_i)} A_j$.
Since $W$ is a
single-predecessor workflow,
for two private modules $m_i, m_{\ell}$, 
the public closures $C(h_i) \cap
C(h_\ell) = \emptyset$ (this follows directly 
from the definition of single-predecessor workflows).
Further, since $W$ is
single-predecessor, $W$ has no data-sharing by definition.  So for
any two modules $m_j, m_\ell$ in $W$ (public or private), the set of
attributes $A_j \cap A_\ell = \emptyset$. Clearly, when $m_i$ is a
private module, $m_i \notin C(h_\ell)$ for any private module $m_{\ell}$ in $W$, by the definition of public-closure. 
Hence for any two private modules $m_i, m_{\ell}$,
$$\left(O_i \cup \bigcup_{j: m_j \in C(h_i)} A_j\right) \cap \left(O_\ell \cup \bigcup_{j: m_j \in C(h_\ell)} A_j\right) = \emptyset.$$
In particular, for two private modules $m_i \neq  m_{\ell}$, $H_i \cap H_\ell = \emptyset$.
Hence, 
 for a public module $m_j \in C(h_i)$, and for any other private module $m_\ell$, 
$A_j \cap H_{\ell} = \emptyset$. Therefore, $A_j \cap H_i'$ = $A_j \cap (\bigcup_{\ell: m_{\ell} \textrm{ is private}} H_{\ell})$
$= A_j \setminus H_i$. Since $m_j$ is \UDS\ \wrt $A_j \cap H_i$ from the condition in the theorem, $m_j$
is also \UDS\ \wrt $A_j \cap H_i'$. This shows that $H_i'$ satisfies the conditions stated in the lemma.
\par
Theorem~\ref{thm:privacy-downward} also states that
each private module $m_i$  is $\Gamma$-standalone-private
\wrt 
$h_i$, i.e., $|\Out_{\tup{x}, m_i, h_i}| \geq \Gamma$ for all input
$\tup{x}$ to $m_i$ 
(see
Definition~\ref{def:standalone-privacy}). From
Lemma~\ref{lem:main-private-single-pred}, using $H_i'$ in place of $H_i$, this implies that for all input
$\tup{x}$ to private modules $m_i$, $|\Out_{\tup{x}, W, H_i'}| \geq \Gamma$ 
where 
$H_i' = \bigcup_{\ell: m_{\ell} \textrm{ is private}} H_{\ell}$.
From Definition~\ref{def:workflow-privacy}, this implies that each private module $m_i$ is $\Gamma$-workflow-private
$H_i'$ which is the same as $H$ in Theorem~\ref{thm:privacy-downward}.
Since this is true for all private module $m_i$ in $W$,
$W$ is $\Gamma$-private \wrt $H$.
\end{proof}
%

\subsection{Optimal Composition for Single Predecessor Workflows}\label{sec:optimization} 

\eat{
				Given a workflow $W$, represented by a relation $R$, and a privacy
				parameter $\Gamma$, we want to find a safe visible subset $V$ with
				minimum cost s.t. all the modules in the workflow are
				$\Gamma$-workflow-private \wrt. $V$. Recall that each attribute $a$
				in $R$  has a cost $c_a$ and ideally one would like to find a set
				$V$ s.t. the cost of hidden attributes $c(\widebar{V}) = \sum_{a \in
				\widebar{V}} c_a$ is minimized. The difficulty however is that, as
				mentioned in Section~\ref{sec:prelims-PODS11}, even for workflows that contain only
				private modules the problem is NP-hard (and in EXP-time) in the
				number of attributes in the relation $R$, which may be very large.
				
				To avoid this exponential dependency on the overall number of
				attributes, we will use here Theorem \ref{thm:privacy-downward} to
				assemble a cost efficient solution for the whole workflow out of
				solutions for the {\em individual modules}. As we show below,
				computing solutions for the individual modules  may still take time
				exponential in the number of the module's attributes. But this
				number is much smaller than that of the whole workflow, and the
				computation may be performed as a pre-processing step with the cost
				being amortized over possibly many uses of the module in different
				workflows.
				
				\paragraph*{The optimization problem}
}
Recall the \emph{optimal composition problem} mentioned in Section~\ref{sec:privacy-opt}.
This problem focused on optimally combining the safe solutions for private modules in an all-private
workflow in order to minimize the cost of hidden attributes. 
\eat{
		Let us use $M^-$ (resp. $M^+$) to
		denote the set of private (resp. public) modules $m_i$ in $W$
		(or their indices $i$ by overloading the notation).
}
In this section, we consider optimal composition 
for a single-predecessor workflow $W$ with private and public modules.
Our goal is to find subsets  $H_i$ for each private module $m_i$ in $W$  
 satisfying the conditions given
 in Theorem \ref{thm:privacy-downward}
 such that $\cost(H)$ is minimized for 
 $H = \bigcup_{i: m_i \textrm{ is private}} H_i$.
\eat{
			subsets $V_i$ for the private modules $m_i$ in $W$ where only output
			attributes of $m_i$ are hidden, $ \widebar{V_i} \subseteq H_i
			\subseteq O_i \cup \bigcup_{k \in C(\widebar{V_i})} A_k$, and for
			every public module $m_j \in C(\widebar{V_i})$, $m_j$ is \UDS\ \wrt.
			$A_j \setminus H_i$. We call this problem the \textbf{optimum-view}
			problem.
			
			In Section~\ref{sec:four-step-opt}, we discuss how the 
			optimal-composition problem can be solved in four steps, which also allows us to study the complexity of the problem
			that arises from the respective steps. In Section~\ref{sec:opt-standalone} and Section~\ref{sec:opt-single-private},
			we discuss two of these four steps in more detail.
}
This we solve in four steps: (I) find the safe solutions for
standalone-privacy for individual private modules; (II) find the
\UDS\ solutions for individual public modules; (III) find the
optimal hidden subset $H_i$ for the public-closure of every private
module $m_i$ using the outputs of the first two steps; and
(IV) combine $H_i$-s to find the final optimal solution $H$. 
We next
consider each of these steps. 


\paragraph*{I. 
Private Solutions for Individual Private Modules}

For each private module $m_i$ 
we compute the set of safe subsets 
$\safe_i = \{S_{i1}, \cdots, S_{ip_i}\}$, where 
each $S_{i\ell} \subseteq O_i$ 
is standalone-private for $m_i$.
Here $p_i$ is the number of 
safe subsets for $m_i$.
Recall from Theorem \ref{thm:privacy-downward} that the choice of
safe subset for $m_i$ determines its public-closure (and
consequently the possible $H_i$ sets and the cost of the overall
solution). It is thus not sufficient to consider only the safe
subsets that have the minimum cost; we need to keep {\em all}
 safe subsets for $m_i$, to be examined by  subsequent steps.
\par
The complexity of finding safe subsets for individual private
modules has been thoroughly studied in \cite{DKM+11} under the name
\emph{standalone \secureview\ problem}. It was shown that deciding
whether a given 
hidden subset of attributes is safe for a private module is
NP-hard in the number of attributes of the module.
\eat{
				 an exponential
				(again in the number of attributes of the module) lower bound on
				the communication complexity under different settings for this
				problem were also given. 
}
It was further shown that the set of
\emph{all} safe subsets for the module can be computed in time
exponential in the number of attributes assuming constant domain
size, which almost matches the lower bounds.

\par
Although the lower and upper bounds are somewhat disappointing, as
argued in \cite{DKM+11}, the number of attributes of an individual
module is fairly small. The assumption of constant domain is
reasonable for practical purposes, assuming that the integers and
reals are represented in a fixed number of bits. In these cases the
individual relations can be big, however this computation can be
done only once as a pre-processing step and the cost can be
amortized over possibly many uses of the module in different
workflows. Expert knowledge (from the module designer) can also be 
used to help find the safe subsets.

%

\paragraph*{II. 
Safe Solutions for Individual Public Modules}
This step focuses on finding the set of all \UDS\ solutions for the
individual public modules.
We denote the \UDS\ solutions for a public module 
$m_j$ by
$\uds_j = \{U_{j1}, \cdots, U_{jp_{j}}\}$, where each \UDS\ subset
$U_{j\ell} \subseteq A_j$, and $p_j$ denotes the number of \UDS\
solutions for the public module $m_j$. 
We will see below in Theorem~\ref{thm:nphard-uds}
that even deciding whether a given subset is \UDS\ for a module is
coNP-hard in the number of attributes (and that the set of all such
subsets can be computed in exponential time). However, as argued 
in the first step, this computation can be done once as a pre-processing step with
its cost amortized over possibly many workflows where the module is
used. In addition, 
it suffices to compute the \UDS\ subsets for only those public
modules that belong to some public-closure 
for some private module.
\eat{
		\par
		We show that verifying whether a standalone module $m_j$ is
		upstream-downstream-safe (\UDS) \wrt a given subset $V$ is coNP-hard.
}

\begin{theorem}\label{thm:nphard-uds}
Given public module $m_j$ with $k$ attributes, and
a subset of hidden attributes 
$H$, deciding whether $m_j$ is \UDS\ \wrt 
$H$ is coNP-hard in $k$.
Further, all \UDS\ subsets can be found in EXP-time in $k$.
\end{theorem}
\begin{proof}[Sketch of NP-hardness]
The reduction is from the UNSAT problem, where given $n$  variables
$x_1, \cdots, x_n$, and a 3NF formula $f(x_1, \cdots, x_n)$, the
goal is to check whether $f$ is \emph{not} satisfiable. In our
construction, $m_i$ has $n+1$ inputs $x_1, \cdots, x_n$ and $y$, and
the output is $z = m_i(x_1, \cdots, x_n, y) = f(x_1, \cdots, x_n)
\vee y$ (OR). The set of 
hidden attributes is $x_1, \cdots, x_n$ (\ie\ $y, z$ are visible). 
We claim that $f$ is not satisfiable
if and only if $m_i$ is \UDS\ \wrt 
$H$.
\end{proof}

The same construction in the NP-hardness proof, with attributes $y$ and $z$ assigned cost
zero and all other attributes assigned some higher constant cost,
can be used to show that testing whether a safe subset with cost
smaller than a given threshold exists is also coNP-hard. 

\par
Regarding the upper bound, the trivial algorithm of going  over all
$2^k$ subsets $h$ of $A_j$, 
and checking
if 
$h$ is \UDS\ for $m_j$, can be done in EXP-time 
in $k$ when
the domain size is constant.
Since the \UDS\ property is \emph{not monotone} \wrt further deletion of attributes,
if 
$h$ is \UDS, its 
supersets may not be \UDS. 
Recall however that the trivial solution $h = A_j$ (deleting all attributes) 
is always \UDS\ for $m_j$.
So for practical purposes, when the public-closure 
for a private module involves a small number of
attributes of the public modules in the closure, or if the
attributes of those public modules have small cost, 
this solution
can be used.
The complete proof of the theorem is given in Appendix~\ref{sec:proof-thm:nphard-uds}.

\eat{
						
						\par
						
						In the
						appendix we also show that the communication complexity to solve
						these problems is exponential in the number of attributes.
						
						\red{Sudeepa: I am thinking of removing the \\
						discussion on communication complexity since \\
						the proof is very similar to the PODS paper \\
						and the reviewers complained about the\\
						appendix being too long.}
						
						\scream{Sue:  Yes, I think this would be good.  You can then change the paragraph above to
						"It can also be shown that the communication complexity of
						these problems is exponential in the number of attributes."}
}

\paragraph*{III. Optimal $H_i$ for Each Private Module}
The third step aims to find a set $H_i$ of hidden attributes, of
minimum cost, for every private module $m_i$. 
As per the
theorem statement, this set $H_i$ should satisfy the conditions: 
(a) $H_i \cap O_i 
= S_{i\ell}$, for some safe subset $S_{i\ell} \in
\safe_{i}$; (b) 
for every public
module $m_j$ in the closure 
$C(S_{i\ell})$, there exists a
\UDS\ subset $U_{jq} \in \uds_{j}$ such that $U_{jq} = A_j \cap H_i$; 
and (c) $H_i$ does not include any attribute outside $O_i$ and $C(S_{i\ell})$.
\eat{
We will discuss solving this problem below; 
the motivation will be clear
from the next step.
}

\eat{
			\par
			Given a private module $m_i$, the set of safe subsets  $\safe_i$ of
			$m_i$, and the set of \UDS\ subsets $\uds_j$ for all public modules
			$m_j$, the goal of this step is to find the hidden subset $H_i$ with
			minimum cost $c(H_i)$such that $H_i \supseteq \comple{S_{i\ell}}$
			for some $S_{i\ell} \in \safe_i$ and for all $m_j \in
			C(\comple{S_{i\ell}})$, there is a $U_{jq} \in \uds_j$ such that
			$U_{jq} = A_j \setminus H_j$. We call this problem the
			\textbf{single-module} problem.
}

We show that, 
for the important class of chain and tree
workflows, this optimization problem is solvable in time polynomial
in the number of modules $n$, the total number of attributes in the workflow $|A|$, and the
maximum number of sets in $\safe_i$ and $\uds_j$ (denoted by
$L = \max_{i \in [1, n]} p_i$):
\begin{theorem}\label{prop:step3-chain-wf}
For each private module $m_i$ in a tree workflow (and therefore, in a chain workflow), the optimal subset $H_i$
can be found in 
polynomial time in $n$, $|A|$ and
$L$. 
\par
On the other hand, the problem is NP-hard when the workflow has arbitrary DAG structure
even when both the number of attributes and the number of safe and \UDS\
subsets of the individual modules are bounded by a small constant.
\end{theorem}
In contrast, the problem becomes NP-hard in $n$ when the
public-closure forms an arbitrary directed acyclic subgraph, even when $L$ is a constant
and the number of attributes of the individual modules is bounded by
a small constant.

Chain workflows are the simplest class of
tree-shaped workflow, hence clearly any algorithm for trees will
also work for chains. However, for the sake of simplicity,
 we 
 give the optimal algorithm for chain workflows first;
 then we discuss how it can be proved for tree workflows.

\eat{
					Note that to obtain an algorithm of time polynomial in $L$, for a
					given module $m_i$, we can go over all choices of safe subsets
					$S_{i\ell} \in \safe_i$ of $m_i$, compute the public-closure
					$C(\comple{S_{i\ell}})$, and choose a minimal cost subset $H_i =
					H_i(S_{i\ell})$ that satisfies the \UDS\ properties of all modules in
					the public-closure. Then, output, among them, a subset having the
					minimum cost. Consequently, it suffices to explain how, given a safe
					subset $S_{i\ell} \in \safe_i$, one can solve, in PTIME, the problem
					of finding a minimum cost hidden subset $H_i$ that satisfies the \UDS ty
					property of all modules in a subgraph formed by a given
					$C(\comple{S_{i\ell}})$.
}

\textbf{Optimal algorithm for chain workflows.~~}
Consider any private module $m_i$. Given a safe
subset $S_{i\ell} \in \safe_i$, 
we show below how an optimal subset $H_i$ in $C(S_{i\ell})$ satisfying the desired properties
can be obtained. We then repeat this process for all safe subsets
(bounded by $L$) $S_{i\ell} \in \safe_i$, and output the subset $H_i$ with minimum cost.
We drop the subscripts to simplify the notation (\ie\ use $S$ for $S_{i\ell}$, 
$C$ for $C(S_{i\ell})$, and $H$ for $H_i$).
\eat{
			, the given safe subset $S_{i\ell}$ will be
			denoted below by $S_{i*}$, the closure $C(\widebar{S_{i\ell}})$ will
			be denoted by $C_W$, and the output hidden subset will be denoted by
			$H$.
}


Our poly-time algorithm employs dynamic programming to find the optimal
$H$. First note that since $C$ is the public-closure of output
attributes for a chain workflow, $C$ should be a chain itself. Let
the modules in $C$ be renumbered as $m_1, \cdots, m_k$ in order.
Now we solve the problem by dynamic programming as follows. Let $Q$
be an $k \times L$ two-dimensional array, where $Q[j, \ell]$ denotes
the cost of minimum cost hidden subset $H^{j\ell}$ that satisfies
the \UDS\ condition for all public modules $m_1$ to $m_j$ and 
$A_{j} \cap H^{j\ell} = U_{j\ell} \in \uds_j$.
Here $j \leq k$, $\ell \leq p_j \leq L$, and $A_{j}$ is the attribute set of $m_j$; the actual solution can be stored easily by standard
argument.
\par
The initialization step is , for $1 \leq \ell \leq p_1$,
\begin{eqnarray*}
Q[1, \ell] & = & c(U_{1, \ell})~~~~ \text{if } U_{1, \ell} \supseteq S \\ 
& = & \infty~~~~ \text{otherwise}
\end{eqnarray*}
\eat{
				\begin{eqnarray*}
				Q[1, \ell] & = & c(\comple{U_{1, \ell}})~~~~ \text{if } \comple{U_{1, \ell}} \supseteq \comple{S} \\ 
				& = & \infty~~~~ \text{otherwise}
				\end{eqnarray*}
}
Recall that for a chain, $O_{j-1} = I_j$, for $j = 2$ to $k$.
Then for $j = 2$ to $k$, $\ell = 1$ to $p_j$,
\begin{eqnarray*}
Q[j, \ell] & = & \infty ~~~~ \text{if there is no } 1 \leq q \leq p_{j-1}\\
& & ~~~ \text{ such that } U_{j-1, q} \cap O_{j-1} = U_{j, \ell} \cap I_{j}\\
& = & c(O_j \cap U_{j\ell}) + \min_{q} Q[j-1, q] \\
& & \text{ where the minimum is taken over all such  } q
\end{eqnarray*}
\eat{
			\begin{eqnarray*}
			Q[j, \ell] & = & \infty ~~~~ \text{if there is no } 1 \leq q \leq p_{j-1}\\
			& & ~~~ \text{ such that } \comple{U_{j-1, q}} \cap O_{j-1} = \comple{U_{j, \ell}} \cap I_{j}\\
			& = & c(O_j \cap \comple{U_{j\ell}}) + \min_{q} Q[j-1, q] \\
			& & \text{ where the minimum is taken over all such  } q
			\end{eqnarray*}
}


 It is interesting to note that such a $q$ always exists for at least one $\ell \leq p_j$: while defining \UDS\ subsets in Definition~\ref{def:uds},
we discussed that any public module $m_j$ is \UDS\ when its entire attribute set $A_j$ is hidden. Hence
$A_{j-1} \in \uds_{j-1}$ and $A_j \in \uds_{j}$,
which will make the equality check true (for a chain $O_{j-1} = I_j$).
In Appendix~\ref{sec:proof-lem:correct-chain-dp} we show that shows that 
$Q[j,\ell]$ correctly stores the desired value. 
 Then the optimal solution
			$H$ has cost $\min_{1 \leq \ell \leq p_k} Q[k, \ell]$; the
			corresponding solution $H$ can be found by standard procedure, which 
			proves Theorem~\ref{prop:step3-chain-wf} for chain workflows.\\
\eat{
			The following lemma (whose proof is given in
			Appendix~\ref{sec:proof-lem:correct-chain-dp}) shows that $Q[j,
			\ell]$ correctly stores the desired value. Then the optimal solution
			$H$ has cost $\min_{1 \leq \ell \leq p_k} Q[k, \ell]$; the
			corresponding solution $H$ can be found by standard procedure.
			
			\begin{lemma}\label{lem:correct-chain-dp}
			For $1 \leq j \leq k$, the entry $Q[j, \ell]$, $1 \leq \ell \leq
			p_j$,  stores the minimum cost of the hidden attributes $\cup_{x =
			1}^j A_x \supseteq H^{j\ell}  \supseteq \widebar{S_{i*}}$ such that
			$A_{j} \setminus H^{j\ell} = U_{j\ell}$, and every module $m_x, 1
			\leq x\leq j$ in the chain is \UDS\ \wrt $A_x \setminus
			H^{j\ell}$.\end{lemma}
			
			This concludes our proof.
}



Observe that, more generally, the algorithm may also be used for 
non-chain workflows, if the public-closures of the safe
subsets for private modules have chain shape. This observation also
applies to the following discussion on tree workflows.


\textbf{Optimal algorithm for tree workflows.~~} Now consider tree-shaped workflows, where every
module in the workflow has at most one immediate predecessor (for
all modules $m_i$, 
if $I_i \cap O_j \neq \emptyset$ and $I_i \cap O_k
\neq \emptyset$, then $j = k$), but a module can have one or more
immediate successors.

 The treatment of tree-shaped workflows is similar to what we have seen  for
 chains. Observe that, here again, since 
 $C$ is the public-closure of output
 attributes for a tree-shaped workflow, 
 $C$
 will be a collection of trees all rooted at $m_i$.
 As for the case of chains, the processing of the public closure is based
 on dynamic-programming. The key difference is that
 the modules in the tree are processed bottom up (rather than top down as in what we
 have seen above) to handle branching. The 
 proof of Theorem~\ref{prop:step3-chain-wf} for tree workflows 
 is given in Appendix~\ref{sec:step3-tree-wf}.

\eat{
			 The details of the algorithm are given in Appendix~\ref{sec:step3-tree-wf}, proving
			 the following theorem.
			\begin{theorem}\label{prop:step3-tree-wf}
			The single-subset problem can be solved in PTIME (in $n$, $|A|$ and
			$L$) for tree-shaped workflows.
			\end{theorem}
}


\textbf{NP-hardness for public-closure of arbitrary shape.~~}
Finding the minimal-cost solution for public-closure with arbitrary DAG shape 
is NP-hard.
\eat{
				\begin{theorem}\label{lem:single-priv-soln-nphard}
				Given a private module $m_i$ in a workflow, 
				determining whether a subset $H_i$ with cost smaller than a given bound
				exists satisfying the
				required conditions is NP-hard. This is the case
				even when both the number of attributes and the number of safe and \UDS\
				subsets of the individual modules are bounded by a small constant.
				\end{theorem}
}
We give a reduction from 3SAT (see
Appendix~\ref{sec:proof-lem:single-priv-soln-nphard}). The NP
algorithm simply guesses a set of attributes and checks whether it
forms a legal solution and has cost lower than the given bound;
a corresponding EXP-time algorithm
that iterates over all subsets can be used to find the optimal
solution.
\par

The NP-completeness here is in $n$, the number of modules in the
public closure. We note, however, that in practice the number of
public modules that process the output on an individual private
module is small.
So the obtained
solution to the \textbf{optimum-view} problem is still better than
the naive one, which is exponential in the size of the {\em full}
workflow. 

\paragraph*{IV. Optimal 
Hidden Subset $H$ for the Workflow}
According to Theorem~\ref{thm:privacy-downward}, 
$H = \bigcup_{i: m_i \textrm{ is private}}H_i$
is a $\Gamma$-private solution for the
workflow. Observe that finding the optimal (minimum cost) such
solution 
$H$ for single-predecessor workflows is straightforward,
once the minimum cost $H_i$-s are found: 
Due to the condition in Theorem~\ref{thm:privacy-downward} that no unnecessary data are hidden, 
it can be easily checked that for any
 two private modules $m_i, m_k$ in a single predecessor workflow,  
 $H_i \cap H_k = \emptyset$. 
 This implies that 
the optimal solution $H$ can be obtained taking the union of the optimal hidden
subsets $H_i$ for individual private modules obtained in the previous step.

\eat{
		As we mentioned above, the first step has been discussed in \cite{DKM+11}  and the fourth step is
		straightforward for single-predecessor workflows.  We now focus on describing solutions to the second step
		 (Section~\ref{sec:opt-standalone}) and the third step (Section~\ref{sec:opt-single-private}).
}

\section{General Workflows}\label{sec:general-wf}
The previous sections focused on single-predecessor workflows. In
particular, we presented a privacy theorem for such workflows and
studied optimization \wrt this theorem. 
The following two observations
highlight how this privacy theorem 
can be extended to general workflows.
For lack of space the discussion is informal; the proof techniques are similar to
single-predecessor workflows and are given in Appendix~\ref{sec:app-general-wf}.\\


\textbf{Observation 1: Need for propagation through private modules.~~}
All examples in the previous sections that showed the necessity of the single-predecessor assumption
for private module $m_i$ had another private module $m_k$ as 
which is a successor of one public module in the public closure of $m_i$.
For instance, in
the proof of Proposition~\ref{prop:single-pred-needed} (see Figure~\ref{fig:wf-no-pred})
$m_i = m_1$ and $m_k = m_4$. If
we had continued hiding output attributes of $m_4$, 
we could obtain the required possible
worlds leading to a non-trivial privacy guarantee $\Gamma > 1$.  This implies that
for general workflows, the propagation of attribute hiding should
continue outside the public closure and through the descendant private
modules.\\
\eat{
All examples in the previous sections that showed the necessity of the single-predecessor assumption
had another private module $m_k$ as a
successor of the private module $m_i$ being considered. For instance, in
the proof of Proposition~\ref{prop:single-pred-needed} (see Figure~\ref{fig:wf-no-pred})
$m_i = m_1$ and $m_k = m_4$. If
we had continued hiding output attributes of $m_4$, 
we could obtain the required possible
worlds leading to a non-trivial privacy guarantee $\Gamma > 1$.  This implies that
for general workflows, the propagation of attribute hiding should
continue outside the public closure and through the descendant private
modules.}

\textbf{Observation 2: \DS ty  suffices (instead of \UDS ty).~~} The proof of
Lemma~\ref{lem:main-private-single-pred} shows that the \UDS ty property
of  modules in the public-closure is needed only when some
public module in the public-closure has a private successor whose
output attributes are visible.   If all modules in the public closure
have no such private successor, then a downstream-safety
property (called the \textbf{\DS ty  property}) is sufficient. More
generally, if attribute  hiding is propagated through private
modules (as discussed above), then it suffices to require the hidden
attributes to satisfy the \DS ty  property rather than the stronger \UDS ty 
property.\\

The intuition from the above two observations is formalized in a
\emph{privacy theorem for general workflows}, analogous to Theorem
\ref{thm:privacy-downward}. First, instead of public-closure, it
uses \emph{downward-closure}: for a private module $m_i$, and a set
of hidden attributes $h_i$, the downward-closure $D(h_i)$ consists
of all modules (public or private) $m_j$, that are reachable from
$m_i$ by a directed path. Second, instead of requiring the sets
$H_i$ of hidden attributes to ensure \UDS ty , it requires them to only
ensure \DS ty.

The proof of the revised theorem is similar to that of
Theorem \ref{thm:privacy-downward}, with the added complication
that the $H_i$
subsets are no longer guaranteed to be disjoint. This is resolved by
proving that \DS subsets are closed under union, allowing for the
(possibly overlapping) $H_i$ subsets computed for the individual
private modules to be 
unioned.

%

The hardness results from the previous section transfer to the
case of general workflows. Since the $H_i$-s in this case may be overlapping,
the union of optimal $H_i$ solutions  for individual modules $m_i$
may not give the optimal solution for the workflow.  Whether or not there exists a non-trivial
approximation  is an interesting  open problem.

 To conclude the discussion, note that for single-predecessor workflows, we now have two options to ensure
workflow-privacy: (i) to consider public-closures and ensure
\UDS ty  properties for their modules (following the privacy theorem for
single-predecessor workflows); or (ii) to consider
downward-closures and ensure \DS ty  properties for their modules
(following the privacy theorem for general workflows). Observe that
these two options are incomparable: Satisfying \UDS ty  properties may
require hiding more attributes than what is needed for
satisfying \DS ty  properties. On the other hand, the downward-closure
includes more modules than the public-closure (for instance the
reachable private modules), and additional attributes must be hidden to
satisfy their \DS ty  properties. One could therefore run both algorithms, and choose the lower cost solution.


%

\par

\section{Related Work}\label{sec:related}

Privacy concerns with respect to provenance were articulated in
\cite{DBLP:conf/icdt/DavidsonKRSTC11}, in the context of scientific
workflows, and in \cite{DeutchM11}, in the context of business
processes. Preserving module privacy in all-private workflows was
studied in~\cite{DKM+11} and the idea of privatizing (hiding the
``name'' of) public modules to achieve privacy in public/private
workflows was proposed. Unfortunately this is
not realistic for many common scenarios. This paper thus presents a
novel {\em propagation model} for attribute hiding which
does not place any assumptions on the user's prior knowledge 
about public modules. 

Recent work by other authors includes the development of
fine-grained access  control languages for
provenance~\cite{NiXBSH09,TanGMJMTM06,CadenheadKT11,CadenheadKKT11},
and a graph grammar approach for rewriting redaction policies over
provenance~\cite{DBLP:conf/sacmat/CadenheadKKT11}. 
The approach in \cite{DBLP:journals/pvldb/BlausteinCSAR11} provides users with
informative graph query results using {\em surrogates}, which give
less sensitive versions of nodes/edges, and proposes a utility
measure for the result.
A framework to output a {\em partial} view of  a workflow that conforms to a given
set of access permissions on the connections and input/output ports
was proposed in~\cite{CCL+08}. Although related to module privacy,
the approach may disconnect connections between modules rather than just
hiding the data which flows between them; furthermore, it may hide more provenance information
than our mechanism. More importantly, the
notion of privacy is informal and no guarantees on the quality of
the solutions are provided. 

A related area
is that of {\em privacy-preserving data mining} (see
surveys~\cite{DataM, VBF+04}, and the references therein). Here, the goal is
to hide individual data attributes while retaining the suitability of the data for
mining patterns.
Privacy preserving approaches have been studied for
\emph{social networks} (e.g. \cite{BDK07}),  
\emph{auditing queries} (e.g. \cite{NMK+06}), 
{\em network routing} \cite{gurney-2011-pvr},
and several other contexts. 
\par
Our notion of module privacy is closest to the notion
of $\ell$-diversity considered in \cite{MKG+07} which addresses some shortcomings of  
$\kappa$-anonymity \cite{Swee02}.
The notion of $\ell$-diversity tries to generalize
the values of the \emph{non-sensitive attributes} so that for
every such generalization, there are at least $\ell$ different values of
\emph{sensitive attributes}. The view-based approach for $k$-anonymity along with its complexity has been studied in \cite{Yao+05-view}.
Leakage of information due to knowledge on the techniques for minimizing data loss has been studied in \cite{Wong+07-minimality, Fang+08leakage, Cormode+10-attacks, Xiao+10-transparent}; however, our privacy guarantees are information theoretic under our assumptions.


Nevertheless, the privacy notion of $\ell$-diversity is susceptible to 
attack when the user has 
background knowledge \cite{GKS+08, Kifer09}.
  \emph{Differential privacy} \cite{DMN+06, Dwork08, Dwork09},
  which
 requires that the output distribution is {\em almost} invariant to
the inclusion of any particular record, 
gives a stronger privacy guarantee.
   Although it was first proposed for {\em statistical databases} and \emph{aggregate queries}, 
   it has since been studied in domains such as
   mechanism design~\cite{MT07}, data streaming~\cite{DNP+10},  
   and several database-related applications (e.g. \cite{KM11, XWG10, CPS+11, CMF+11}).
However, it is well-known that no {\em deterministic} algorithm can guarantee differential privacy,
and the standard approach of including random noise is
not suitable for our purposes --- provenance queries are
typically not aggregate queries, and we need the output views to be consistent (e.g. the same module must map
the same input to the same output in all executions of the workflow).
\eat{
					The \emph{Exponential Mechanism} \cite{MT07}, which implements differential
					privacy in non-numeric domains, can be applied to our problem
					by suitably defining the feasible output space.
					However, our initial research shows that even if we ignore the time complexity of implementing this mechanism,
					the utility guarantee provided is often trivial due to the nature of provenance queries.
}
Defining an appropriate notion of differential privacy for module functionality
\wrt provenance queries is an interesting open problem. It would also be interesting to study 
natural attacks for our application, and (theoretically or empirically) 
study the effectiveness of various notions of privacy under these attacks 
(see e.g. \cite{Cormode11}).
\section{Conclusion}\label{sec:conclusion}
In this paper, we addressed the problem of preserving module privacy in public/private workflows
(called workflow-privacy), 
by providing a view of provenance information in which the input to output mapping of private 
modules remains hidden. 
As several examples in this paper show, the workflow-privacy of a module critically depends on the structure (connection patterns) of the workflow, the behavior/functionality of other modules in the workflow, and the selection of hidden attributes.
We showed how workflow-privacy can be achieved by propagating 
data hiding
through public modules in both single-predecessor and general workflows.
\par
\par
Several interesting future research directions related to the application of differential privacy 
were discussed in Section~\ref{sec:related}. 
We assumed certain assumptions in the paper (constant domain size, acyclic nature of workflows, analysis using
relations of executions, etc.). Even with these assumptions, the problem is highly non-trivial and
large and important classes of workflows can be captured even under these assumption.
However, it would be immensely important to have models and solutions that can be used in scientific experiments
in practice. We have also mentioned the shortcomings of the $\Gamma$-privacy and the difficulty in 
using stronger privacy notions like differential privacy in the previous section.
It will be interesting to see if the possible world model thoroughly studied in this
paper can be used to facilitate the use of other privacy models under provenance queries.
\eat{
Another interesting problem is to develop PTIME approximation algorithms for module privacy 
(that can handle non-monotonicity of UDS and DS subsets) in single-predecessor and 
general workflows.
}
\eat{
The DS subsets for a module is not monotone w.r.t. inclusion of input attributes, and the UDS subsets are not monotone
w.r.t. inclusion of both input and output attributes which makes the natural choices (for instance, the greedy algorithm) 
non-applicable, and the study of approximation algorithms for this type of non-monotone problem may also be of independent interest.
}



\newpage

\appendix
\section{Proofs from Section~4}
\subsection{Proof of Proposition~\ref{prop:single-pred-needed}}\label{sec:proof-prop:single-pred-needed}

By Definition~\ref{def:single-pred-wf}, a workflow $W$ is \emph{not}
a single-predecessor workflow, if one of the following holds: (i) there is a public module $m_j$ in $W$
that belongs to the public-closure of a private module $m_i$ but has
no directed path from $m_i$, or, (ii) such a public module $m_j$
has directed path from more than one private modules, or,
(iii) $W$ has data sharing. 

To prove the proposition we provide three example workflows where
exactly one of the violating conditions (i), (ii), (iii) holds, and
Theorem~\ref{thm:privacy-downward} does not hold in those workflows.
Case (i) was shown in Section \ref{sec:privacy-thm-eg}. To complete
the proof we demonstrate here cases (ii) and (iii).

\paragraph*{Multiple private predecessor} We give an example
where Theorem~\ref{thm:privacy-downward} does not hold when 
a public module belonging to a public-closure has more than
one private predecessors.
\begin{example}\label{eg:single-pred}
Consider the workflow $W_a$ in Figure~\ref{fig:wf-twopaths}, 
which is a modification of $W_a$ by the addition of  private module
$m_0$, that takes $a_0$ as input and produces $a_2 = m_0(a_0) = a_0$
as output. The public module $m_3$ is in public-closure of $m_1$,
but has directed public paths from both $m_0$ and $m_1$. The
relation $R_a$ for $W_a$ in given in Table~\ref{tab:reln-mult-pred}
where the
 hidden attributes $\{a_2, a_3, a_4, a_5\}$ are colored in grey.


\begin{table}[t]
\centering
\begin{tabular} {r|cc|cccc|c|}
\cline{2-8}
&$a_0$ & $a_1$ & \cellcolor[gray]{0.6}$a_2$ & \cellcolor[gray]{0.6}$a_3$ & \cellcolor[gray]{0.6}$a_4$ & \cellcolor[gray]{0.6}$a_5$ & $a_6$ \\ \hline\hline
 $r_1$ & 0&  0 &  \cellcolor[gray]{0.6}0 & \cellcolor[gray]{0.6}0 & \cellcolor[gray]{0.6}0 &\cellcolor[gray]{0.6} 0 & 0 \\\hline
 $r_2$ & 1&  0 & \cellcolor[gray]{0.6}1 & \cellcolor[gray]{0.6}0 & \cellcolor[gray]{0.6}1 & \cellcolor[gray]{0.6}1 & 1 \\\hline
 $r_3$ & 0&  0 &\cellcolor[gray]{0.6}1 &  \cellcolor[gray]{0.6}0 & \cellcolor[gray]{0.6}1 & \cellcolor[gray]{0.6}1 & 1 \\\hline
 $r_4$ & 1 &1 & \cellcolor[gray]{0.6}1 &\cellcolor[gray]{0.6}1 &\cellcolor[gray]{0.6}1 &\cellcolor[gray]{0.6}1 & 1\\\hline
\end{tabular}
\caption{Relation $R_a$ for workflow $W_a$ given in
Figure~\ref{fig:wf-twopaths}} \label{tab:reln-mult-pred}
\end{table}

Now we have exactly the same problem as before: When $\widehat{m}_1$ maps 0 to
1, $a_5 = 1$ irrespective of the value of $a_4$.  In the first row
$a_6 = 0$, whereas  in the second row  $a_6 = 1$.  However, whatever
the new definitions of $\widehat{m}_0$ are for $m_0$ and $\widehat{m}_4$ for $m_4$,
$\widehat{m}_4$ cannot map 1 to both 0 and 1. Hence $\Gamma = 1$. $\Box$
\end{example}

\begin{figure}[t]
\centering
\subfloat[{\small $m_3$ has paths from $m_0,
m_1$}]{\label{fig:wf-twopaths}
\includegraphics[scale=0.25]{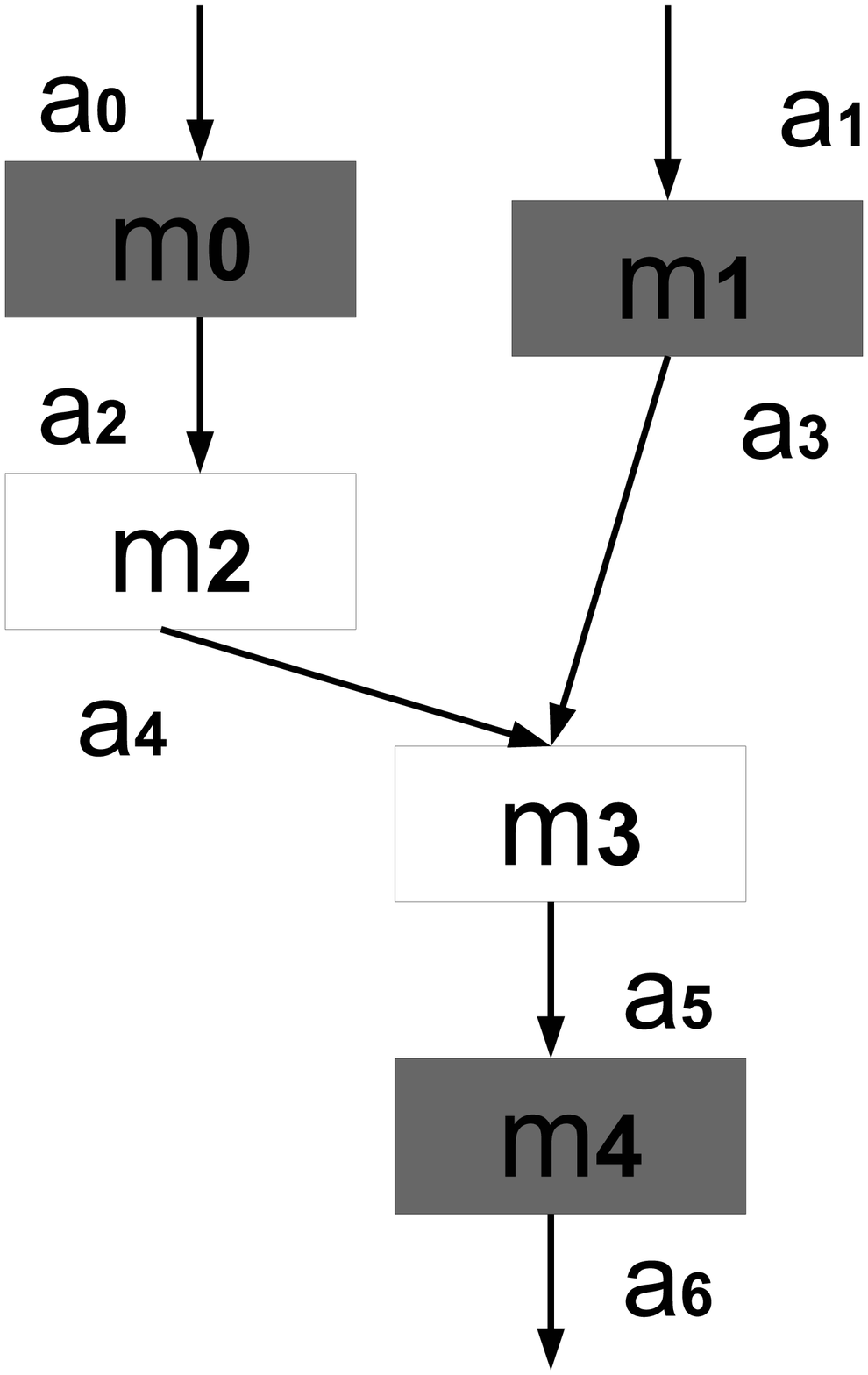}
}~~~~~~~
 \subfloat[{\small $a_3$ is shared as input to $m_2, m_3$}]{\label{fig:wf-datashare}
\includegraphics[scale=0.25]{wf-data-sharing.eps} 
} 
\caption{Proof of Proposition~\ref{prop:single-pred-needed}: (a) Multiple private predecessors, (b) Data sharing. 
White modules are public, Grey are private.} 
\label{fig:wf-private-pred}
\end{figure}

\paragraph*{Data sharing} Now we give an example
where Theorem~\ref{thm:privacy-downward} does not hold when 
the workflow has data sharing.
\begin{example}\label{eg:data-share}
Consider the workflow, say $W_b$, 
 given in Figure~\ref{fig:wf-datashare}. 
 All attributes take
 values in $\{0, 1\}$. The initial inputs are $a_1, a_2$, and  final
outputs are $a_6, a_7$; only $m_4$ is public.  The functionality of
modules is as follows: (i) $m_1$ takes $a_1, a_2$ as input and
produces $m_1(a_1, a_2) = (a_3 = a_1, a_4 = a_2)$. (ii) $m_2$ takes
$a_3, a_4$ as input and produces $a_5 = m_2(a_3, a_4) = a_3 \vee
a_4$  (OR). (iii) $m_3$  takes $a_5$ as input and produces $a_6 =
m_3(a_5) = a_5$. (iv) $m_4$ takes $a_3$ as input and produces $a_7 =
m_4(a_3) = a_3$. Note that data $a_3$ is input to both $m_2, m_4$,
hence the workflow has data sharing.

Now focus on private module $m_1 = m_i$.  Clearly hiding output
$a_3$ of $m_1$ gives $2$-standalone privacy. and for hidden
attribute $h_i = \{a_3\}$, the public-closure $C(h_i) = \{m_2\}$. As
given in the theorem, $H_i \subseteq O_i \cup \bigcup_{j: m_j \in
C(h_i)} A_j$
$= \{a_3, a_4, a_5\}$ in this case.

 We claim that hiding even all of
$\{a_3, a_4, a_5\}$ gives only trivial 1-workflow-privacy of $m_1$,
although the \UDS\ condition is satisfied for $m_2$ (actually hiding
$a_3, a_4$ gives 4-standalone-privacy for $m_1$).
Table~\ref{tab:reln-data-share} gives the relation $R_b$, where the
hidden attribute values are in Grey.

\begin{table}[t]
{\scriptsize
\centering
\begin{tabular} {r|cc|ccc|cc|}
\cline{2-8}
& $a_1$ & $a_2$ & \cellcolor[gray]{0.6}$a_3$ &
\cellcolor[gray]{0.6}$a_4$ & \cellcolor[gray]{0.6}$a_5$ & $a_6$ &
$a_7$ \\ \hline \hline
$r_1$ & 0 & 0 & \cellcolor[gray]{0.6}0 &
\cellcolor[gray]{0.6}0 & \cellcolor[gray]{0.6}0 & 0 & 0\\\hline
$r_2$ & 0 & 1 & \cellcolor[gray]{0.6}0 & \cellcolor[gray]{0.6}1 &
\cellcolor[gray]{0.6}1 & 1 & 0 \\\hline $r_3$ & 1 & 0 &
\cellcolor[gray]{0.6}1 & \cellcolor[gray]{0.6}0 &
\cellcolor[gray]{0.6}1 & 1 & 1 \\\hline $r_4$ & 1 & 1 &
\cellcolor[gray]{0.6}1 & \cellcolor[gray]{0.6}1 &
\cellcolor[gray]{0.6}1 & 1 & 1 \\\hline
\end{tabular}
\caption{Relation $R_b$ for workflow $W_b$ given in
 Figure~\ref{fig:wf-private-pred}. 
 } \label{tab:reln-data-share}
 }
\end{table}

When $a_3$ (and also $a_4$) is hidden, a possible candidate output
of input tuple $x = (0, 0)$ to $m_1$ is $(\ul{1}, \ul{0})$. So we
need to have a possible world where $m_1$ is redefined as $\widehat{m}_1(0, 0)
= (1, 0)$. Then $a_5$ takes value 1 in the first row, and this is
the only row with visible attributes $a_1 = 0, a_2 = 0$. So this
requires that $\widehat{m}_3(a_5 = 1) = (a_6 = 0)$ and $\widehat{m}_4(a_3 = 1) = (a_7 =
0)$, to have the same projection on visible $a_6, a_7$.

The second, third and fourth rows, $r_2, r_3, r_4$, have $a_6 = 1$,
so to have the same projection, we need $a_5 = 0$ for these three
rows, so we need $\widehat{m}_3(a_5 = 0) = (a_6 = 1)$ (since we had to already
define $\widehat{m}_3(1) = 0$). When $a_5$ is 0, since the public module $m_2$
is an OR function, the only possibility of the values of $a_3, a_4$
in rows $r_2, r_3, r_4$ are $(0, 0)$. Now we have a conflict on the
value of the visible attribute $a_7$, which is 0 for $r_2$ but 1 for
$r_3, r_4$, whereas for all these rows the value of $a_3$ is 0.
$\widehat{m}_4$ being a function with dependency $a_3 \rightarrow a_7$, cannot
map $a_3$ to both 0 and 1. Similarly we can check that if $\widehat{m}_1(0, 0)
= (0, 1)$ or $\widehat{m}_1(0, 0) = (1, 1)$ (both $a_3, a_4$ are hidden), we
will have exactly the same problem. Hence all possible worlds of
$R_{d}$ with these hidden attributes must map $\widehat{m}_1(0, 0)$ to $(0,
0)$,
and therefore $\Gamma = 1$. $\Box$
\end{example}

\eat{

\begin{figure}[t!]
\centering
\includegraphics[scale=0.2]{figures/wf-data-sharing.eps} 
\includegraphics[scale=0.2]{figures/wf-mult-pred.eps}
\caption{White modules are public, Grey are private
   (a) $a_3$ is shared as input to $m_2, m_3$
   (b) $m_2$ is in public-closure of $m_1$, but has directed paths from both $m_0, m_1$.
} \label{fig:wf-private-pred}
\end{figure}

}


\subsection{Proof of Lemma~\ref{lem:main-private-single-pred}}\label{sec:proof-lem:main-private-single-pred}
The proof of Lemma~\ref{lem:main-private-single-pred} uses the following lemma.
It 
states that the 
if $\tup{y}$ is a candidate output of an input $\tup{x}$ to module $m_i$ with respect to
hidden attributes $h_i$ (\ie\ $\tup{y} \in \Out_{\tup{x}, m_i, h_i}$),
then 
$\tup{y}$ and the actual output of $\tup{x}$, $\tup{z} = m_i(\tup{x})$, 
must be equivalent.

\begin{lemma}\label{lem:out-x-output}
Let $m_i$ be a standalone private module with relation $R_i$, let $\tup{x}$
be an input to $m_i$, and let $h_i \subseteq O_i$ be a
subset of hidden attributes. 
If $\tup{y} \in \Out_{\tup{x}, m_i, h_i}$ then $\tup{y} \equiv_{h_i} \tup{z}$
where $\tup{z} = m_i(\tup{x})$.
\end{lemma}

Note that, in Example~\ref{eg:main-lemma}, $\tup{y} = (\ul{1}, 0)$ and $\tup{z} = (\ul{0}, 0)$ are equivalent on the visible attributes
as Lemma~\ref{lem:out-x-output} says (hidden attributes are underlined). 
\eat{
				\textsc{Lemma}~\ref{lem:out-x-output}.~~ 
				\emph{Let $m_i$ be a standalone private module with relation $R_i$, let $\tup{x}$
				be an input to $m_i$, and let $V_i$ be a
				subset of visible attributes such that  $\widebar{V_i} \subseteq O_i$ (only output attributes are hidden).
				If $\tup{y} \in \Out_{\tup{x}, m_i, V_i}$ then $\tup{y} \equiv_{V_i} \tup{z}$
				where $\tup{z} = m_i(\tup{x})$.}\\
}
\begin{proof}
A subset of output attributes of $m_i$, $h_i \subseteq O_i$, is hidden. Recall that $A_i = I_i \cup O_i$
denotes the set of attributes of $m_i$ and let $R_i$ be the standalone relation for $m_i$.
If $\tup{y} \in \Out_{\tup{x}, m_i, h_i}$, 
then from Definition~\ref{def:standalone-privacy},
\begin{equation}
\exists R' \in \Worlds(R_i, {h_i}),~~\exists \tup{t'} \in R'~~ s.t~~
  \tup{x} = \proj{I_i}{\tup{t'}} \wedge \tup{y}=\proj{O_i}{\tup{t'}}
  \end{equation}
  Further, from Definition~\ref{def:pos-worlds-standalone}, $R' \in \Worlds(R_i, {h_i})$
  only if $\proj{A_i \setminus h_i}{R_i} = \proj{A_i \setminus h_i}{R'}$. Hence there must exist a tuple $\tup{t} \in R_i$
  such that
  \begin{equation}
  \proj{A_i \setminus h_i}{\tup{t}} = \proj{A_i \setminus h_i}{\tup{t'}} \label{equn:out-x-1}
  \end{equation}
  Since $h_i \subseteq O_i$, $I_i \subseteq A_i \setminus h_i$.
  From~(\ref{equn:out-x-1}), $\proj{I_i}{\tup{t}} = \proj{I_i}{\tup{t'}} = \tup{x}$.
  Let 
  $\tup{z} = \proj{O_i}{\tup{t}}$, i.e. $\tup{z} = m_i(\tup{x})$.
  From~(\ref{equn:out-x-1}), $\proj{O_i \setminus h_i}{\tup{t}} = \proj{O_i \setminus h_i}{\tup{t'}}$, then
  $\proj{O_i \setminus h_i}{\tup{z}} = \proj{O_i \setminus h_i}{\tup{y}}$.
  Tuples $\tup{y}$ and $\tup{z}$ are defined on $O_i$. Hence from Definition~\ref{def:equiv},
  $\tup{y} \equiv_{h_i} \tup{z}$.
\end{proof}

\begin{corollary}\label{cor:out-equiv}
For a module $m_i$, and 
hidden attributes $h_i \subseteq O_i$, if two tuples
$\tup{y}, \tup{z}$ defined on $O_i$ are such that $\tup{y} \equiv_{h_i} \tup{z}$, then also 
$\tup{y} \equiv_{H_i} \tup{z}$ where $H_i \supseteq h_i$ is a set of hidden attributes in the workflow.
\end{corollary}
\begin{proof}

Since $\tup{y} \equiv_{h_i} \tup{z}$, $\proj{O_i \setminus h_i}{\tup{y}} = \proj{O_i \setminus h_i}{\tup{z}}$.
Since 
$h_i \subseteq H_i$,  $O_i \setminus h_i \supseteq O_i \setminus H_i$.
Therefore, $\proj{O_i \setminus H_i}{\tup{y}} = \proj{O_i \setminus H_i}{\tup{z}}$, 
\ie\ 
$\tup{y} \equiv_{H_i} \tup{z}$.
\end{proof}

\textbf{Note.}~~ Lemma~\ref{lem:out-x-output} does not use any property of single-predecessor workflows and also works
for general workflows. This lemma will be used again for the privacy theorem of general workflows
(Theorem~\ref{thm:privacy-general}).\\

\smallskip
\noindent
\textbf{Definition of \Flip\ and \EFlip\ (extended \Flip) functions.~~}
To prove Lemma~\ref{lem:main-private-single-pred}, we need to show existence of a possible world
satisfying the criteria. This possible world will be obtained by joining alternative definitions of
private modules, and the original definition of public modules. 
We will need the following flipping functions to formally present how we derive the alternative
module definitions from original modules.
These function examines parts of inputs, and possibly changes parts of original outputs.

\begin{definition}\label{def:flip}
Given  subsets of attributes $P, Q \subseteq A$, two tuples $\tup{p}, \tup{q}$ defined on $P$, and a tuple $\tup{u}$ defined on $Q$, $\Flip_{\tup{p}, \tup{q}}(\tup{u}) = \tup{w}$
is a tuple defined on $Q$ constructed as follows: 
\begin{itemize}
\item if $\proj{Q \cap P}{\tup{u}} = \proj{Q \cap P}{\tup{p}}$, then $\tup{w}$ is such that $\proj{Q \cap P}{\tup{w}} = \proj{Q \cap P}{\tup{q}}$
and $\proj{Q \setminus P}{\tup{w}} = \proj{Q \setminus P}{\tup{w}}$, 
\item else if $\proj{Q \cap P}{\tup{u}} = \proj{Q \cap P}{\tup{q}}$, then $\tup{w}$ is such that $\proj{Q \cap P}{\tup{w}} = \proj{Q \cap P}{\tup{p}}$
and $\proj{Q \setminus P}{\tup{w}} = \proj{Q \setminus P}{\tup{w}}$, 
\item otherwise, $\tup{w} = \tup{u}$.
\end{itemize}
\end{definition}

The following observations capture the properties of \Flip\ function.
\begin{observation}\label{obs:flip}
~~\\
\begin{enumerate}
    \item \label{equn:f1} If $\Flip_{\tup{p}, \tup{q}}(\tup{u}) = \tup{w}$, then $\Flip_{\tup{p}, \tup{q}}(\tup{w}) = \tup{u}$.
    \item \label{equn:f2} $\Flip_{\tup{p}, \tup{q}}(\Flip_{\tup{p}, \tup{q}}(\tup{u})) = \tup{u}$.
    \item \label{equn:f3} If $P \cap Q = \emptyset$, $\Flip_{\tup{p}, \tup{q}}(\tup{u}) = \tup{u}$.
    \item \label{equn:f4} $\Flip_{\tup{p}, \tup{q}}(\tup{p}) = \tup{q}, \Flip_{\tup{p}, \tup{q}}(\tup{q}) = \tup{p}$.
    \item \label{equn:f5} If  $\proj{Q \cap P}{\tup{p}} = \proj{Q \cap P}{\tup{q}}$, then $\Flip_{\tup{p}, \tup{q}}(\tup{u}) = \tup{u}$.
    \item \label{equn:f6} If $Q = Q_1 \cup Q_2$, where $Q_1 \cap Q_2 = \emptyset$, and if
    $\Flip_{\tup{p}, \tup{q}}(\proj{Q_1}{\tup{u}}) = \tup{w_1}$ and $\Flip_{\tup{p}, \tup{q}}(\proj{Q_2}{\tup{u}}) = \tup{w_2}$,
    then $\Flip_{\tup{p}, \tup{q}}(\tup{u}) = \tup{w}$ such that $\proj{Q_1}{\tup{w}} = \tup{w_1}$
    and $\proj{Q_2}{\tup{w}} = \tup{w_2}$.

\end{enumerate}
\end{observation}


The above definition of flipping will be useful when we consider the scenario where $M$ does not have any successor.
When $M$ has successors, we need an extended definition of tuple flipping based on other tuples, denoted by \EFlip, as defined below.

\begin{definition}\label{def:eflip}
Given  subsets of attributes $P, Q, R \subseteq A$, where 
two tuples $\tup{p}, \tup{q}$ defined on $P \cup R$, a tuple $\tup{u}$ defined on $Q$ and a tuple $\tup{v}$ defined on $R$, $\EFlip_{\tup{p}, \tup{q}; \tup{v}}(\tup{u}) = \tup{w}$
is a tuple defined on $Q$ constructed as follows: 
\begin{itemize}
\item if $\tup{v} = \proj{R}{\tup{p}}$, then $\tup{w}$ is such that $\proj{Q \cap P}{\tup{w}} = \proj{Q \cap P}{\tup{q}}$
and $\proj{Q \setminus P}{\tup{w}} = \proj{Q \setminus P}{\tup{w}}$, 
\item else if $\tup{v} = \proj{R}{\tup{q}}$, then $\tup{w}$ is such that $\proj{Q \cap P}{\tup{w}} = \proj{Q \cap P}{\tup{p}}$
and $\proj{Q \setminus P}{\tup{w}} = \proj{Q \setminus P}{\tup{w}}$, 
\item otherwise, $\tup{w} = \tup{u}$.
\end{itemize}
Note that $\EFlip_{\tup{p}, \tup{q}; \proj{P\cap Q}{\tup{u}}}(\tup{u}) = \Flip_{\tup{p}, \tup{q}}(\tup{u})$, where $R = P \cap Q$.
\end{definition}
\begin{observation}\label{obs:eflip}
\begin{enumerate}
    \item If $\EFlip_{\tup{p}, \tup{q}; \tup{v}}(\tup{u}) = \tup{w}$, and $\tup{u}'$ is a tuple defined on $Q' \subseteq Q$ , then $\EFlip_{\tup{p}, \tup{q}; \tup{v}}(\tup{u}') = \proj{Q'}{\tup{w}}$.
\end{enumerate}
\end{observation}

\smallskip
\noindent
\textbf{Proof of Lemma~\ref{lem:main-private-single-pred}.~~} Now we are ready to prove the lemma.
As mentioned in Section~\ref{sec:privacy-single-pred}, we will assume that there is a single (composite) 
public module $M$ in the public closure $C(h_i)$ of $m_i$.
Recall that $I_i, O_i, A_i$ denote the set of input, output and all attributes of $m_i$
respectively.

\smallskip
\noindent
\textsc{Lemma}~\ref{lem:main-private-single-pred}.~~
\emph{Consider a standalone private module $m_i$, a set of 
hidden attributes $h_i$, any input $\tup{x}$ to $m_i$, 
and any candidate output $\tup{y} \in \Out_{\tup{x}, m_i, h_i}$ of $\tup{x}$.
Then $\tup{y} \in \Out_{\tup{x}, W, H_i}$ when $m_i$ belongs to a single-predecessor workflow $W$,  
and a set attributes $H_i \subseteq A$ is hidden 
such that (i) 
$h_i \subseteq H_i$, (ii) only output attributes from $O_i$ are included in 
$h_i$ (i.e. $h_i \subseteq O_i$),
and (iii) every module $m_j$ in the public-closure $C(h_i)$ is \UDS\ \wrt $H_i$.
}\\

\begin{proof} 
We fix a module $m_i$, an input $\tup{x}$ to $m_i$, hidden attributes $h_i \subseteq O_i$, and a
candidate output $\tup{y} \in \Out_{\tup{x}, m_i, h_i}$ for $\tup{x}$. 
We assume that there is a single public module $M$ in the public closure $C(h_i)$.
By the properties of single-predecessor workflows, $M$ gets all its inputs from $m_i$
and sends its outputs to zero or more than one private modules.
We denote the inputs and outputs of $M$ by $I \subseteq h_i$ and $O$ respectively.
 However $m_i$ can also send (i) its visible outputs
to other public modules (these public modules will have $m_i$ as its only predecessor, but 
these public modules will not have
any public path in undirected sense to $M$), and it can send  
(ii) visible and hidden attributes to other private modules.

From the conditions in the lemma, a set $H_i$ is hidden in the workflow where
(i) $h_i \subseteq H_i$, (ii) $h_i \subseteq O_i$,
and (iii) $M$ is \UDS\ \wrt $H_i$.
We will show that $\tup{y} \in \Out_{\tup{x}, W, H_i}$. 
We prove this by showing the existence of a possible world $R' \in \Worlds(R, H_i)$,
such that if $\proj{I_i}{\tup{t}} = \tup{x}$ for some $\tup{t} \in R'$, then $\proj{O_i}{\tup{t}} = \tup{y}$.
Since $\tup{y} \in \Out_{x, m_i, h_i}$, by Lemma~\ref{lem:out-x-output}, $\tup{y} \equiv_{h_i} \tup{z}$ where $\tup{z} = m_i(\tup{x})$.
We consider two cases separately based on whether $M$ has no successor or at least one private successors. \\


\textbf{Case I.~~} First consider the easier case
that $M$ does not have any successor, so all outputs of $M$ belong to the set of final outputs.
We redefine the module $m_i$ to $\widehat{m}_i$ as follows.
For an input $\tup{u}$ to $m_i$,
$\widehat{m}_i(\tup{u})$ = $\Flip_{\tup{y}, \tup{z}}(m_i(\tup{u}))$.
All public modules are unchanged, $\widehat{m}_j = m_j$.
All private modules $m_j \neq m_i$ are redefined as follows: On an input $\tup{u}$ to $m_j$,
$\widehat{m}_j(\tup{u}) = m_j(\Flip_{\tup{y}, \tup{z}}(\tup{u}))$.
The required possible world $R'$ is obtained by taking the join of the standalone relations of these $\widehat{m}_j$-s, $j \in [n]$.

First note that by the definition of $\widehat{m}_i$, $\widehat{m}_i(\tup{x}) = \tup{y}$ 
(since $\widehat{m}_i(x)$ = $\Flip_{\tup{y}, \tup{z}}(m_i(x))$ = $\Flip_{\tup{y}, \tup{z}}(\tup{z})$ = $\tup{y}$,
from Observation~\ref{obs:flip}(\ref{equn:f4})).
Hence if $\proj{I_i}{\tup{t}} = \tup{x}$ for some $\tup{t} \in R'$, then $\proj{O_i}{\tup{t}} = \tup{y}$.

Next we argue that  $R' \in \Worlds(R, H_i)$. Since $R'$ is the join of the standalone relations for
modules $\widehat{m}_j$-s, $R'$ maintains all functional dependencies $I_j \rightarrow O_j$.
Also none of the public modules are unchanged, hence for any public module $m_j$ and any tuple $t$
in $R'$, $\proj{O_j}{\tup{t}} = m_j(\proj{I_j}{\tup{t}})$. So we only need to show that
the projection of $R$ and $R'$ on the visible attributes are the same.


%
Let us assume, wlog. that the modules are numbered in topologically sorted order.
Let $I_0$ be the initial input attributes to the workflow, and let $\tup{p}$ be a tuple defined on $I_0$.
There are two unique tuples $\tup{t} \in R$ and $\tup{t'} \in R'$
such that $\proj{I_1}{\tup{t}} = \proj{I_1}{\tup{t'}} = \tup{p}$.
Since $M$ does not have any successor, let us assume that $M = m_{n+1}$, also wlog. assume that
the public modules in $C$ are not counted in $j = 1$ to $n+1$ by renumbering the modules.
Note that any intermediate or final attribute $a \in A \setminus I_0$ belongs to $O_j$, for a unique $j \in [1, n]$
(since for $j \neq \ell$, $O_j \cap O_{\ell} = \phi$).
So it suffices to show that $t, t'$ projected on $O_j$ are equivalent \wrt visible attributes for all
module $j$, $j = 1$ to $n+1$.

Let $\tup{c}_{j, m}, \tup{c}_{j, \widehat{m}}$ be the values of input attributes $I_j$ and $\tup{d}_{j, m}, \tup{d}_{j, \widehat{m}}$
be the values of output attributes $O_j$ of module $m_j$, in $\tup{t} \in R$ and $\tup{t'} \in R'$
respectively on initial input attributes $\tup{p}$
(i.e. $\tup{c}_{j, m} = \proj{I_j}{\tup{t}}$, $\tup{c}_{j, \widehat{m}} = \proj{I_j}{\tup{t'}}$,
$\tup{d}_{j, m} = \proj{O_j}{\tup{t}}$ and $\tup{d}_{j, \widehat{m}} = \proj{O_j}{\tup{t'}}$).
We prove by induction on $j = 1$ to $n$ that

\begin{equation}
\forall j, 1 \leq j \leq n, \tup{d}_{j, \widehat{m}} = \Flip_{\tup{y}, \tup{z}}(\tup{d}_{j, m})\label{equn:flip}
\end{equation}

First we argue that proving (\ref{equn:flip}) shows that the join of $\angb{\widehat{m}_i}_{1 \leq i \leq n}$
is a possible world of $R$ \wrt 
hidden attributes $H_i$.
(A) When $m_j$ is a private module,
note that $\tup{d}_{j, m}$
and $\tup{d}_{j, \widehat{m}} = \Flip_{\tup{y}, \tup{z}}(\tup{d}_{j, m})$ may differ only on attributes $O_j \cap O_i$
But $\tup{y} \equiv_{h_i} \tup{z}$, i.e. these tuples are equivalent on the visible attributes.
Hence for all private modules, the $t, t'$ are equivalent \wrt $O_j$.
(actually for all $j \neq i$, $O_j \cap O_i = \emptyset$, so the outputs are equal and therefore equivalent).
(B) When $m_j$ is a public module, $j \neq n+1$, $O_j \cap O_i = \emptyset$, hence the values of $t, t'$
on $O_j$ are the same and therefore equivalent. (C) Finally, consider $M = m_{n+1}$ that is not covered by
(\ref{equn:flip}). $M$ gets all its inputs from $m_i$. From (\ref{equn:flip}), 
$$\tup{d}_{i, \widehat{m}} = \Flip_{\tup{y}, \tup{z}}(\tup{d}_{i, m})$$
Since $\tup{y}, \tup{z}, \tup{d}_{i, m}, \tup{d}_{i, \widehat{m}}$ are all defined on attributes $O_i$, and input to $m_{n+1}$, $I_{n+1} \subseteq O_i$,
$$\tup{c}_{n+1, \widehat{m}} = \Flip_{\tup{y}, \tup{z}}(\tup{c}_{n+1, m})$$
 Hence $\tup{c}_{n+1, \widehat{m}} \equiv_{H_i} \tup{c}_{n+1, m}$. Since these two inputs of $m_{n+1}$
are equivalent \wrt $H_i$, by the \UDS\ property of $M = m_{n+1}$, its outputs are also equivalent,
i.e. $\tup{d}_{n+1, \widehat{m}} \equiv_{H_i} \tup{d}_{n+1, m}$. Hence the projections of $t, t'$ on $O_{n+1}$ are also equivalent.
Combining (A), (B), (C), $t, t'$ are equivalent \wrt $H_i$.

\textbf{Proof of (\ref{equn:flip}).~~} The base case follows for $j = 1$.
If $m_1 \neq m_i$ ($m_j$ can be public or private), then $I_1 \cap O_i = \emptyset$, so  for all input $\tup{u}$,

$$\widehat{m}_j(\tup{u}) = m_j(\Flip_{\tup{y}, \tup{z}}(\tup{u})) = m_j(\tup{u})$$ 

Since the inputs $\tup{c}_{1, \widehat{m}} = \tup{c}_{1, m}$
(both projections of initial input $\tup{p}$ on $I_1$),
the outputs $\tup{d}_{1, \widehat{m}} = \tup{d}_{1, m}$. This shows (\ref{equn:flip}).
If $m_1 = m_i$, the inputs are the same, and by definition of $\widehat{m}_1$,
\begin{eqnarray*}
\tup{d}_{1, \widehat{m}} & = & \widehat{m}_1(\tup{c}_{1, \widehat{m}})\\
&  = & \Flip_{\tup{y}, \tup{z}}(m_i(\tup{c}_{1, \widehat{m}}))\\
&  = & \Flip_{\tup{y}, \tup{z}}(m_i(\tup{c}_{1, m}))\\
&  = & \Flip_{\tup{y}, \tup{z}}(\tup{d}_{1, m})
\end{eqnarray*}
This shows (\ref{equn:flip}).

Suppose the hypothesis holds until $j-1$, consider $m_{j}$. From the induction hypothesis,
$\tup{c}_{j, \widehat{m}} = \Flip_{\tup{y}, \tup{z}}(\tup{c}_{j, m})$, hence $\tup{c}_{j, m} = \Flip_{\tup{y}, \tup{z}}(\tup{c}_{j, \widehat{m}})$ (see Observation~\ref{obs:flip} (\ref{equn:f1})).

\begin{itemize}
  \item[(i)] If $j = i$, again,\\

  \begin{eqnarray*}
\tup{d}_{i, \widehat{m}}  &= &\widehat{m}_i(\tup{c}_{i, \widehat{m}})\\
& = & \Flip_{\tup{y}, \tup{z}}(m_i(\tup{c}_{i, \widehat{m}}))\\
& = & \Flip_{\tup{y}, \tup{z}} (m_i(\Flip_{\tup{y}, \tup{z}}(\tup{c}_{i, m})))\\
& = & \Flip_{\tup{y}, \tup{z}}(m_i((\tup{c}_{i, m}))\\
& = & \Flip_{\tup{y}, \tup{z}}(\tup{d}_{i, m})
\end{eqnarray*}

 $\Flip_{\tup{y}, \tup{z}}(\tup{c}_{i, m}) = \tup{c}_{i, m}$ follows due to the fact that $I_i \cap O_i = \emptyset$, $\tup{y}, \tup{z}$ are defined on $O_i$,
whereas $\tup{c}_{i, m}$ is defined on $I_i$ (see Observation~\ref{obs:flip} (\ref{equn:f3})).

\item [(ii)] If $j \neq i$ and $m_j$ is a private module,
\begin{eqnarray*}
\tup{d}_{j, \widehat{m}} &  =  & \widehat{m}_j(\tup{c}_{j, \widehat{m}})\\
& = & m_j(\Flip_{\tup{y}, \tup{z}}(\tup{c}_{j, \widehat{m}}))\\
& = & m_j(\tup{c}_{j, m})\\
& = & \tup{d}_{j, m}\\
& = & \Flip_{\tup{y}, \tup{z}}(\tup{d}_{j, m})
\end{eqnarray*}
$\Flip_{\tup{y}, \tup{z}}(\tup{d}_{j, m})$ = $\tup{d}_{j, m}$ follows due to the fact that $O_j \cap O_i = \emptyset$, $\tup{y}, \tup{z}$ are defined on $O_i$,
whereas $\tup{d}_{i, m}$ is defined on $O_j$ (again see Observation~\ref{obs:flip} (\ref{equn:f3})).

\item[(iii)] If $m_j$ is a public module, $j \leq n$, $\widehat{m}_j = m_j$.\\
Here
\begin{eqnarray*}
\tup{d}_{j, \widehat{m}} & = & \widehat{m}_j(\tup{c}_{j, \widehat{m}})\\
& = & m_j(\tup{c}_{j, \widehat{m}})\\
& = & m_j(\Flip_{\tup{y}, \tup{z}}(\tup{c}_{j, m}))\\
& = & m_j(\tup{c}_{j, m})\\
& = & \tup{d}_{j, m}\\
& = & \Flip_{\tup{y}, \tup{z}}(\tup{d}_{j, m})
\end{eqnarray*}
$\Flip_{\tup{y}, \tup{z}}(\tup{d}_{j, m})$ $= \tup{d}_{j, m}$ again follows due to the fact that $O_j \cap O_i = \emptyset$.
$\Flip_{\tup{y}, \tup{z}}(\tup{c}_{j, m})$ $= \tup{c}_{j, m}$ follows due to following reason.
If  $I_j \cap O_i = \emptyset$, i.e. if $m_j$ does not get any input from $m_i$, again this is true (Observation~\ref{obs:flip}(\ref{equn:f3})).
If $m_j$ gets an input from $m_i$, i.e. $I_j \cap O_i \neq \emptyset$, since $m_j \neq m_{n+1}$,
$I_j \cap O_i$ does not include any hidden attributes from 
$h_i$.
But $\tup{y} \equiv_{h_i} \tup{z}$, i.e. the visible attribute values of $\tup{y}, \tup{z}$ are the same.
In other words, $\proj{I_j \cap O_i}{\tup{y}} = \proj{I_j \cap O_i}{\tup{z}}$, and from Observation~\ref{obs:flip} (\ref{equn:f5}),
$\Flip_{\tup{y}, \tup{z}}(\tup{c}_{j, m}) = \tup{c}_{j, m}$.
\end{itemize}

This completes the proof of the lemma for Case-I.\\

\par

\medskip
\noindent
\textbf{Case II.}

Now consider the case when $M$ has one or more private successors
(note that $M$ cannot have any public successor by definition).
Let $M = m_{k}$, and assume that the modules $m_1, \cdots, m_n$ are sorted in topological order.
Hence $I = I_k, O = O_k$,
and $I_k \subseteq O_i$.
Let $w_y = M(\proj{I_k}{\tup{y}})$, $w_z = M(\proj{I_k}{\tup{z}})$.
Instead of $\tup{y}, \tup{z}$, the flip function will be \wrt $\tup{Y}, \tup{Z}$, where $\tup{Y}$ is the concatenation of $\tup{y}$ and $w_y$
($\proj{O_i}{\tup{Y}} = \tup{y}$, $\proj{O_k}{\tup{Y}} = w_y$),
and $\tup{Z}$ is the concatenation of $\tup{z}$ and $w_z$. Hence $\tup{Y}, \tup{Z}$ are defined on attributes $O_i \cup O_k$.
\par
We redefine the module $m_i$ to $\widehat{m}_i$ as follows.
Note that since input to $M$, $I_k \subseteq O_i$, $O_i$ is disjoint union of $I_k$ and $O_i \setminus I_k$.
For an input $\tup{u}$ to $m_i$,
$\widehat{m}_i(\tup{u})$, defined on $O_i$ is such that 
$$\proj{O_i \setminus I_k}{\widehat{m}_i(\tup{u})} = \Flip_{\tup{Y}, \tup{Z}}(\proj{O_i \setminus I_k}{m_i(\tup{u})})$$
 and
$$\proj{I_k}{\widehat{m}_i(\tup{u})} = \EFlip_{\tup{Y}, \tup{Z}; M(\proj{I_k}{m_i(\tup{u})})}(\proj{I_k}{m_i(\tup{u})})$$
For the component with \EFlip, in terms of the notations in Definition~\ref{def:eflip}, $R = O_k$, $P = Q = O_i$. 
$\tup{p} = \tup{Y}, \tup{q} = \tup{Z}$, defined on $P \cup R = O_i \cup O_k$. $\tup{v} = M(\proj{I_k}{m_i(\tup{u})})$, defined on $O_k$.
$\tup{u}$ in Definition~\ref{def:eflip} corresponds to $m_i(\tup{u})$.
All public modules are unchanged, $\widehat{m}_j = m_j$.
All private modules $m_j \neq m_i$ are redefined as follows: On an input $\tup{u}$ to $m_j$,
$\widehat{m}_j(\tup{u}) = m_j(\Flip_{\tup{Y}, \tup{Z}}(\tup{u}))$.
The required possible world $R'$ is obtained by taking the join of the standalone relations of these $\widehat{m}_j$-s, $j \in [n]$.

First note that by the definition of $\widehat{m}_i$, $\widehat{m}_i(\tup{x}) = \tup{y}$ due to the following reason:

(i)  $M(\proj{I_k}{m_i(x)})$ = $M(\proj{I_k}{\tup{z}})$ = $w_z = \proj{O_k}{\tup{Z}}$,
so 
\begin{eqnarray*}
\proj{I_k}{\widehat{m}_i(x)} & =  & \EFlip_{\tup{Y}, \tup{Z}; M(\proj{I_k}{m_i(x)})}(\proj{I_k}{m_i(x)})\\
&  =  & \EFlip_{\tup{Y}, \tup{Z}; M(\proj{I_k}{\tup{z}})}(\proj{I_k}{\tup{z}})\\
&  = & \proj{I_k}{\tup{y}})
\end{eqnarray*}

(ii) 
\begin{eqnarray*}
\proj{O_i \setminus I_k}{\widehat{m}_i(x)} & = = & \Flip_{\tup{Y}, \tup{Z}}(\proj{O_i \setminus I_k}{m_i(x)})\\
&  = & \Flip_{\tup{Y}, \tup{Z}}(\proj{O_i \setminus I_k}{\tup{z}})\\
&  =  & \proj{O_i \setminus I_k}{\tup{y}}
\end{eqnarray*}

 Taking union of (i) and (ii), $\widehat{m}_i(x) = \tup{y}$.
Hence if $\proj{I_i}{\tup{t}} = \tup{x}$ for some $\tup{t} \in R'$, then $\proj{O_i}{\tup{t}} = \tup{y}$.

Again, next we argue that  $R' \in \Worlds(R, H_i)$,
\eat{ Since $R'$ is the join of the standalone relations for
modules $\widehat{m}_j$-s, $R'$ maintains all functional dependencies $I_j \rightarrow O_j$.
Also none of the public modules are unchanged, hence for any public module $m_j$ and any tuple $t$
in $R'$, $\proj{O_j}{\tup{t}} = m_j(\proj{I_j}{\tup{t}})$. So we only need to}
and it suffices to show that
the projection of $R$ and $R'$ on the visible attributes are the same.


%
Let $I_0$ be the initial input attributes to the workflow, and let $\tup{p}$ be a tuple defined on $I_0$.
There are two unique tuples $\tup{t} \in R$ and $\tup{t'} \in R'$
such that $\proj{I_1}{\tup{t}} = \proj{I_1}{\tup{t'}} = \tup{p}$.
Note that any intermediate or final attribute $a \in A \setminus I_0$ belongs to $O_j$, for a unique $j \in [1, n]$
(since for $j \neq \ell$, $O_j \cap O_{\ell} = \phi$).
So it suffices to show that $t, t'$ projected on $O_j$ are equivalent \wrt visible attributes for all
module $j$, $j = 1$ to $n+1$.

Let $\tup{c}_{j, m}, \tup{c}_{j, \widehat{m}}$ be the values of input attributes $I_j$ and $\tup{d}_{j, m}, \tup{d}_{j, \widehat{m}}$
be the values of output attributes $O_j$ of module $m_j$, in $\tup{t} \in R$ and $\tup{t'} \in R'$
respectively on initial input attributes $\tup{p}$
(i.e. $\tup{c}_{j, m} = \proj{I_j}{\tup{t}}$, $\tup{c}_{j, \widehat{m}} = \proj{I_j}{\tup{t'}}$,
$\tup{d}_{j, m} = \proj{O_j}{\tup{t}}$ and $\tup{d}_{j, \widehat{m}} = \proj{O_j}{\tup{t'}}$).
We prove by induction on $j = 1$ to $n$ that

\begin{eqnarray}
\forall j \neq i, 1 \leq j \leq n, \tup{d}_{j, \widehat{m}} = \Flip_{\tup{Y}, \tup{Z}}(\tup{d}_{j, m})\label{equn:flip2}\\
\proj{I_k}{\tup{d}_{i, \widehat{m}}} = \EFlip_{\tup{Y}, \tup{Z}; M(\proj{I_k}{\tup{d}_{i,m}})}(\proj{I_k}{\tup{d}_{i, m}})\label{equn:flip3}\\
\proj{O_i \setminus I_k}{\tup{d}_{i, \widehat{m}}} = \Flip_{\tup{Y}, \tup{Z}}(\proj{O_i \setminus I_k}{\tup{d}_{i,m}})\label{equn:flip4}
\end{eqnarray}

First we argue that proving (\ref{equn:flip2}), (\ref{equn:flip3})  and (\ref{equn:flip4})
shows that the join of $\angb{\widehat{m}_i}_{1 \leq i \leq n}$
is a possible world of $R$ \wrt 
hidden attributes $H_i$.
\begin{itemize}
\item [(A)] When $m_j$ is a private module,  $j \neq i$,
note that $\tup{d}_{j, m}$
and $\tup{d}_{j, \widehat{m}} = \Flip_{\tup{Y}, \tup{Z}}(\tup{d}_{j, m})$ may differ only on attributes $(O_k \cup O_i) \cap O_j$.
But for $j \neq i$ and $j \neq k$ ($m_j$ is private module whereas $m_k$ is the composite public module),
$(O_k \cup O_i) \cap O_j = \emptyset$.
Hence for all private modules other than $m_i$, the $t, t'$ are equal \wrt $O_j$ and therefore equivalent.
\item [(B)] For $m_i$, from (\ref{equn:flip3}),\\
 $\proj{I_k}{\tup{d}_{i, \widehat{m}}} = \EFlip_{\tup{Y}, \tup{Z}; M(\proj{I_k}{\tup{d}_{i, m}})}(\proj{I_k}{\tup{d}_{i, m}})$.
Here $\proj{I_k}{\tup{d}_{i, m}}$ and $\proj{I_k}{\tup{d}_{i, \widehat{m}}}$
may differ on $I_k$ only if \\
$M(\proj{I_k}{\tup{d}_{i, m}}) \in \{w_y, w_z\}$.
By Corollary~{cor:out-equiv}, $\tup{y} \equiv_{H_i} \tup{z}$, i.e.
$\proj{I_k}{\tup{y}} \equiv_{H_i} \proj{I_k}{\tup{z}}$.
But since $M$ is UDS, by the downstream-safety property, $w_y \equiv_{H_i} w_z$.
Then by the upstream-safety property, all inputs $\proj{I_k}{\tup{d}_{i, m}} \equiv_{H_i} \tup{y} \equiv_{H_i} \tup{z}$ such that
$M(\proj{I_k}{\tup{d}_{i, m}}) \in \{w_y, w_z\}$. In particular, if $M(\proj{I_k}{\tup{d}_{i, m}}) = w_y$, then
$\proj{I_k}{\tup{d}_{i, \widehat{m}}} = \proj{I_k}{\tup{z}}$, and \\
$\proj{I_k}{\tup{z}}, \proj{I_k}{\tup{d}_{i, m}}$ will be equivalent \wrt $H_i$.
Similarly, if $M(\proj{I_k}{\tup{d}_{i, m}}) = w_z$, then
$\proj{I_k}{\tup{d}_{i, \widehat{m}}} = \proj{I_k}{\tup{y}}$, and $\proj{I_k}{\tup{y}}, \proj{I_k}{\tup{d}_{i, m}}$ will be equivalent \wrt $H_i$.
So $t, t'$ are equivalent \wrt 
$I_k \setminus H_i$.

Next we argue that $t, t'$ are equivalent \wrt 
$(O_i \setminus I_k) \setminus H_i$.
From (\ref{equn:flip3}), 

$$\proj{O_i \setminus I_k}{\tup{d}_{i, \widehat{m}}} = \Flip_{\tup{Y}, \tup{Z}}(\proj{O_i \setminus I_k}{\tup{d}_{i,m}})$$

$\proj{O_i \setminus I_k}{\tup{d}_{i, \widehat{m}}}$ and $\proj{O_i \setminus I_k}{\tup{d}_{i, m}}$ can differ only if\\
$\proj{O_i \setminus I_k}{\tup{d}_{i, m}} = \proj{O_i \setminus I_k}{\tup{y}}$. Then \\
$\proj{O_i \setminus I_k}{\tup{d}_{i, \widehat{m}}} = \proj{O_i \setminus I_k}{\tup{z}}$,
or, $\proj{O_i \setminus I_k}{\tup{d}_{i, m}} = \proj{O_i \setminus I_k}{\tup{z}}$.
Therefore, $\proj{O_i \setminus I_k}{\tup{d}_{i, \widehat{m}}} = \proj{O_i \setminus I_k}{\tup{y}}$.
But $\proj{O_i \setminus I_k}{\tup{y}}$ and  $\proj{O_i \setminus I_k}{\tup{z}}$ are equivalent \wrt 
$H_i$.
Hence $\proj{O_i \setminus I_k}{\tup{d}_{i, m}}$ and $\proj{O_i \setminus I_k}{\tup{d}_{i, \widehat{m}}}$ are equivalent \wrt 
$H_i$.

Hence $t, t'$ are equivalent on $O_i$.
\item[(C)] When $m_j$ is a public module, $\tup{d}_{j, \widehat{m}} = \Flip_{\tup{Y}, \tup{Z}}(\tup{d}_{j, m})$. Here $\tup{d}_{j, m}, \tup{d}_{j, \widehat{m}}$
can differ only on $(O_k \cup O_i) \cap O_j$. If $j \neq k$, the intersection is empty, and we are done.
If $j = k$, $\tup{d}_{j, m}, \tup{d}_{j, \widehat{m}}$ may differ only if $\tup{d}_{j, m} \in \{w_y, w_z\}$.
But note that $\tup{y} \equiv_{h_i} \tup{z}$, so  $\proj{I_k}{\tup{y}} \equiv_{h_i} \proj{I_k}{\tup{z}}$, and $\proj{I_k}{\tup{y}} \equiv_{H_i} \proj{I_k}{\tup{z}}$.
Since $m_k$ is UDS, for these two equivalent inputs the respective outputs $w_y, w_z$ are also equivalent.
Hence in all cases the values of $t, t'$ on $O_k$ are equivalent.
\end{itemize}
Combining (A), (B), (C), the projections of $t, t'$ on $O_{j}$ are equivalent for all $1 \leq j \leq n$;
i.e.  $t, t'$ are equivalent \wrt 
$H_i$\\

\medskip
\noindent

\textbf{Proof of (\ref{equn:flip2}), (\ref{equn:flip3})  and (\ref{equn:flip4}).~~} The base case follows for $j = 1$.
If $m_1 \neq m_i$ ($m_1$ can be public or private, but $k \neq 1$ since $m_i$ is its predecessor),
then $I_1 \cap (O_i \cup O_k) = \emptyset$, so  for all input $\tup{u}$,
$\widehat{m}_j(\tup{u}) = m_j(\Flip_{\tup{Y}, \tup{Z}}(\tup{u})) = m_j(\tup{u})$ (if $m_1$ is private) and
$\widehat{m}_j(\tup{u}) = m_j(\tup{u})$ (if $m_1$ is public). Since the inputs $\tup{c}_{1, \widehat{m}} = \tup{c}_{1, m}$
(both projections of initial input $\tup{p}$ on $I_1$),
the outputs $\tup{d}_{1, \widehat{m}} = \tup{d}_{1, m}$. This shows (\ref{equn:flip2}).
If $m_1 = m_i$, the inputs are the same, and by definition of $\widehat{m}_1$,
\begin{eqnarray*}
\proj{I_k}{\tup{d}_{1, \widehat{m}}} & = & \proj{I_k}{\widehat{m}_1(\tup{c}_{1, \widehat{m}})}\\
& = & \EFlip_{\tup{Y}, \tup{Z}; M(\proj{I_k}{m_1(\tup{c}_{1, \widehat{m}})})}(\proj{I_k}{m_1(\tup{c}_{1, \widehat{m}})})\\
& = & \EFlip_{\tup{Y}, \tup{Z}; M(\proj{I_k}{m_1(\tup{c}_{1, m})})}(\proj{I_k}{m_1(\tup{c}_{1, m})})\\
& = & \EFlip_{\tup{Y}, \tup{Z}; M(\proj{I_k}{\tup{d}_{1, m}})}(\proj{I_k}{\tup{d}_{1, m}})
\end{eqnarray*}

This shows (\ref{equn:flip3}) for $i = 1$.
Again, by definition of $\widehat{m}_1$,
\begin{eqnarray*}
\proj{O_1 \setminus I_k}{\tup{d}_{1, \widehat{m}}} & = & \proj{O_1 \setminus I_k}{\widehat{m}_1(\tup{c}_{1, \widehat{m}})}\\
&  = &\Flip_{\tup{Y}, \tup{Z}}(\proj{O_1 \setminus I_k}{m_1(\tup{c}_{1, \widehat{m}})}\\
& = & \Flip_{\tup{Y}, \tup{Z}}(\proj{O_1 \setminus I_k}{m_1(\tup{c}_{1, m})}\\
& = & \Flip_{\tup{Y}, \tup{Z}}(\proj{O_1 \setminus I_k}{\tup{d}_{1, m}})
\end{eqnarray*}
%
This shows (\ref{equn:flip4}).

Suppose the hypothesis holds until $j-1$, consider $m_{j}$.
From the induction hypothesis, if $I_j \cap O_i = \emptyset$ ($m_j$ does not get input from $m_i$)
then $\tup{c}_{j, \widehat{m}} = \Flip_{\tup{Y}, \tup{Z}}(\tup{c}_{j, m})$, hence $\tup{c}_{j, m} = \Flip_{\tup{Y}, \tup{Z}}(\tup{c}_{j, \widehat{m}})$ (see Observation~\ref{obs:flip}(\ref{equn:f1})).

\begin{itemize}
  \item [(i)] If $j = i$, $I_i \cap O_i = \emptyset$, hence $\tup{c}_{i, \widehat{m}} = \Flip_{\tup{Y}, \tup{Z}}(\tup{c}_{i, m}) = \tup{c}_{i,m}$ ($I_i \cap (O_i \cup O_k) = \emptyset$,
  $m_k$ is a successor of $m_i$, so $m_i$ cannot be successor of $m_k$). By definition of $\widehat{m}_i$,
\begin{eqnarray*}
  \proj{I_k}{\tup{d}_{i, \widehat{m}}}  &= & \proj{I_k}{\widehat{m}_i(\tup{c}_{i, \widehat{m}})}\\
  &  = & \EFlip_{\tup{Y}, \tup{Z}; M(\proj{I_k}{m_i(\tup{c}_{i, \widehat{m}})})}(\proj{I_k}{m_i(\tup{c}_{i, \widehat{m}})})\\
  & =  & \EFlip_{\tup{Y}, \tup{Z}; M(\proj{I_k}{m_i(\tup{c}_{i, m})})}(\proj{I_k}{m_i(\tup{c}_{i, m})})\\
  & =  & \EFlip_{\tup{Y}, \tup{Z}; M(\proj{I_k}{\tup{d}_{i, m}})}(\proj{I_k}{\tup{d}_{i, m}})
\end{eqnarray*}
This shows (\ref{equn:flip3}).

Again,
\begin{eqnarray*}
\proj{O_i \setminus I_k}{\tup{d}_{i, \widehat{m}}}  &  = & \proj{O_i \setminus I_k}{\widehat{m}_i(\tup{c}_{i, \widehat{m}})}\\
&  =  & \Flip_{\tup{Y}, \tup{Z}}(\proj{O_i \setminus I_k}{m_i(\tup{c}_{i, \widehat{m}})})\\
& = & \Flip_{\tup{Y}, \tup{Z}}(\proj{O_i \setminus I_k}{m_i(\tup{c}_{i, m})})\\
&  = & \Flip_{\tup{Y}, \tup{Z}}(\proj{O_i \setminus I_k}{\tup{d}_{i, m}})
\end{eqnarray*}
This shows (\ref{equn:flip4}).

\item [(ii)] If $j = k$, $m_k$ gets all its inputs from $m_i$, so $\proj{I_k}{\tup{d}_{i, m}} = \tup{c}_{k, m}$.
Hence 
\begin{eqnarray*}
\tup{c}_{k, \widehat{m}} & =  & \EFlip_{\tup{Y}, \tup{Z}; M(\proj{I_k}{\tup{d}_{i, m}})}(\tup{c}_{k, m})\\
& = & \EFlip_{\tup{Y}, \tup{Z}; M(\tup{c}_{k, m})}(\tup{c}_{k, m})\\ 
& = & \EFlip_{\tup{Y}, \tup{Z}; \tup{d}_{k, m}}(\tup{c}_{k, m})
\end{eqnarray*}

Therefore, 
\begin{eqnarray*}
\tup{d}_{k, \widehat{m}} & = & \widehat{m}_k(\tup{c}_{k, \widehat{m}})\\
&  = & m_k(\tup{c}_{k, \widehat{m}})\\
& = & m_k(\EFlip_{\tup{Y}, \tup{Z}; \tup{d}_{k, m}}(\tup{c}_{k, m}))
\end{eqnarray*}

Lets evaluate the term $m_k(\EFlip_{\tup{Y}, \tup{Z}; \tup{d}_{k, m}}(\tup{c}_{k, m}))$.
This says that for an input to $m_k$ is $\tup{c}_{k, m}$, and its output $\tup{d}_{k, m}$, (a) if $\tup{d}_{k, m} = w_y$, then

$$\EFlip_{\tup{Y}, \tup{Z}; \tup{d}_{k, m}}(\tup{c}_{k, m}) = \proj{I_k}{\tup{z}},$$ 

and in turn 

$$\tup{d}_{k, \widehat{m}} = m_k(\EFlip_{\tup{Y}, \tup{Z}; \tup{d}_{k, m}}(\tup{c}_{k, m})) = w_z;$$

(b) if $\tup{d}_{k, m} = w_z$, then
$$\EFlip_{\tup{Y}, \tup{Z}; \tup{d}_{k, m}}(\tup{c}_{k, m}) = \proj{I_k}{\tup{y}},$$

 and in turn 
 
 $$\tup{d}_{k, \widehat{m}} = m_k(\EFlip_{\tup{Y}, \tup{Z}; \tup{d}_{k, m}}(\tup{c}_{k, m})) = w_y;$$
 
(c) otherwise 
\begin{eqnarray*}
\tup{d}_{k, \widehat{m}} & = & m_k(\EFlip_{\tup{Y}, \tup{Z}; \tup{d}_{k, m}}(\tup{c}_{k, m}))\\
& = & m_k(\tup{c}_{k, m}) = \tup{d}_{k, m}
\end{eqnarray*}
According to Definition~\ref{def:flip}, the above implies that 

\begin{eqnarray*}
\tup{d}_{k, \widehat{m}} & = & \Flip_{w_y, w_z}(\tup{d}_{k, m})\\
& = & \Flip_{\tup{Y}, \tup{Z}}(\tup{d}_{k, m})
\end{eqnarray*}

 This shows
(\ref{equn:flip2}).

\item[(iii)] If $j \neq i$ and $m_j$ is a private module, $m_j$ can get inputs from $m_i$.
(but since there is no data sharing 
$I_j \cap I_k = \emptyset$),
and other private or public modules $m_{\ell}, \ell \neq i$ ($\ell$ can be equal to $k$). Let us partition the input to $m_j$ ($\tup{c}_{j, m}$ and
$\tup{c}_{j, \widehat{m}}$ defined on $I_j$)
on attributes $I_j \cap O_i$
and $I_j \setminus O_i$ 
From (\ref{equn:flip2}),
using the induction hypothesis,

\begin{equation}
\proj{I_j \setminus O_i}{\tup{c}_{j, \widehat{m}}} = \Flip_{\tup{Y}, \tup{Z}}(\proj{I_j \setminus O_i}{\tup{c}_{j, m}})\label{equn:flip5}
\end{equation}

Now $I_k \cap I_j = \emptyset$, since there is no data sharing.
Hence $(I_j \cap O_i) \subseteq (O_i \setminus I_k)$.
From (\ref{equn:flip4}) using Observation~\ref{obs:eflip},

\begin{equation}
\proj{I_j \cap O_i}{\tup{c}_{j, \widehat{m}}} = \Flip_{\tup{Y}, \tup{Z}}(\proj{I_j \cap O_i}{\tup{c}_{j, m}})\label{equn:flip6}
\end{equation}

From (\ref{equn:flip5}) and (\ref{equn:flip6}), using Observation~\ref{obs:flip} (\ref{equn:f6}),
and since $\tup{c}_{j, m}, \tup{c}_{j, \widehat{m}}$
are defined on $I_j$, so

\begin{equation}
\tup{c}_{j, \widehat{m}} = \Flip_{\tup{Y}, \tup{Z}}(\tup{c}_{j, m})\label{equn:flip7}
\end{equation}

From (\ref{equn:flip7}),
\begin{eqnarray*}
\tup{d}_{j, \widehat{m}} & = & \widehat{m}_j(\tup{c}_{j, \widehat{m}})\\
&  = & m_j(\Flip_{\tup{Y}, \tup{Z}}(\tup{c}_{j, \widehat{m}}))\\
&  = & m_j(\tup{c}_{j, m})\\
& = & \tup{d}_{j, m}\\
& = & \Flip_{\tup{Y}, \tup{Z}}(\tup{d}_{j, m})
\end{eqnarray*}

$\Flip_{\tup{Y}, \tup{Z}}(\tup{d}_{j, m}) = \tup{d}_{j, m}$ follows due to the fact that $O_j \cap (O_i \cup O_k) = \emptyset$ ($j \neq \{i, k\}$), $\tup{Y}, \tup{Z}$ are defined on $O_i \cup O_k$,
whereas $\tup{d}_{j, m}$ is defined on $O_j$ (again see Observation~\ref{obs:flip} (\ref{equn:f3})).

\item[(iv)] Finally consider  $m_j$ is a public module such that $j \neq k$.
$m_j$ can still get input from $m_i$, but none of the attributes in $I_j \cap O_i$ can be hidden by the definition of $m_k = M = C(h_i)$. 
Further, by the definition of $M = m_k$, $m_j$ cannot get any input
from $m_k$ ($M$ is the closure of public module); so $I_j \cap O_k = \emptyset$.
Let us partition the inputs to $m_j$ ($\tup{c}_{j, m}$ and $\tup{c}_{j, \widehat{m}}$ defined on $I_j$)
into three two disjoint subsets: (a) $I_j \cap O_i$, and (b)
$I_j \setminus O_i$.
Since there is no data sharing $I_k \cap I_j = \emptyset$,
and we again get (\ref{equn:flip7}) that
\begin{eqnarray*}
\tup{c}_{j, \widehat{m}} & = & \Flip_{\tup{Y}, \tup{Z}}(\tup{c}_{j, m})\\
&  = & \tup{c}_{j, m}
\end{eqnarray*}

 $\Flip_{\tup{y}, \tup{z}}(\tup{c}_{j, m}) = \tup{c}_{j, m}$ follows due to following reason.
If  $I_j \cap O_i = \emptyset$, i.e. if $m_j$ does not get any input from $m_i$, again this is true (then $I_j \cap (O_i \cup O_k) = (I_j \cap O_i) \cup (I_j \cap O_k) = \emptyset$).
If $m_j$ gets an input from $m_i$, i.e. $I_j \cap O_i \neq \emptyset$,
since $j \neq k$,
$I_j \cap O_i$ does not include any hidden attributes from 
$h_i$
($m_k$ is the closure $C(h_i)$).
But $\tup{y} \equiv_{h_i} \tup{z}$, i.e. the visible attribute values of $\tup{y}, \tup{z}$ are the same.
In other words, $\proj{I_j \cap O_i}{\tup{y}} = \proj{I_j \cap O_i}{\tup{z}}$, and again from Observation~\ref{obs:flip} (\ref{equn:f5}),
$$\Flip_{\tup{Y}, \tup{Z}}(\tup{c}_{j, m})  = \Flip_{\tup{y}, \tup{z}}(\tup{c}_{j, m}) = \tup{c}_{j, m}$$
 (again, $I_j \cap O_k = \emptyset$).

Therefore,
\begin{eqnarray*}
\tup{d}_{j, \widehat{m}} & = & \widehat{m}_j(\tup{c}_{j, \widehat{m}})\\
& = & m_j(\tup{c}_{j, \widehat{m}})\\
& = & m_j(\Flip_{\tup{Y}, \tup{Z}}(\tup{c}_{j, m}))\\
& = & m_j(\tup{c}_{j, m})\\
& = & \tup{d}_{j, m}\\
& = & \Flip_{\tup{Y}, \tup{Z}}(\tup{d}_{j, m})
\end{eqnarray*}

$\Flip_{\tup{Y}, \tup{Z}}(\tup{d}_{j, m}) = \tup{d}_{j, m}$ again follows due to the fact that $O_j \cap (O_i \cup O_k) = \emptyset$, since $j \notin \{i, k\}$.
\end{itemize}
Hence all the cases for the induction hypothesis hold true, and this completes the proof of the lemma for Case-II.
\end{proof}

\subsection{Proof of Lemma~\ref{lem:composite-UDS}}\label{sec:proof-lem:composite-UDS}
Recall that $I_i, O_i, A_i$ denote the set of input, output and all attributes of a module $m_i$.\\

\smallskip
\noindent
\textsc{Lemma}~\ref{lem:composite-UDS}.~~
\emph{
Let $M$ be a composite module consisting only of public modules.
Let $H$ be a subset of hidden
attributes such that every public module $m_j$ in $M$ is \UDS\ \wrt
$A_j \cap H$. Then $M$ is \UDS\ \wrt $(I \cup O) \cap H$.
}\\

\begin{proof}
Let us assume, wlog., that the modules in $M$ are $m_1, \cdots, m_p$
where modules are listed in topological order.
For $j = 1$ to $p$, let $M^j$ be the composite module comprising
$m_1, \cdots, m_j$, and let $I^j, O^j$ be its input and output.
Hence $M^p = M, I^p = I$ and $O^p = O$. We prove by induction on $2 \leq j \leq p$
that $M^j$ is \UDS\ \wrt $H \cap (I^j \cup O^j)$.
We present the proof without going through the notations for the sake of simplicity.
\par
The base case directly follows for $j = 1$, since $A_1 = I_1 \cup O_1 = I^1 \cup O^1$. Let the hypothesis hold until $M^j$ and consider $M^{j+1}$. 
By induction hypothesis, $M^j$ is \UDS\ \wrt $(I^j \cup O^j) \cap H$. The module $m_{j+1}$
may consume some outputs of $M^j$ ($m_2$ to $m_j$). Hence
\par
\begin{equation}
I^{j+1} = I^j \cup I_{j+1} \setminus O^j ~~\textit{and}~~O^{j+1} = O^j \cup O_{j+1} \setminus I_{j+1} \label{equn:1}
\end{equation}

Consider two equivalent inputs $x_1, x_2$ \wrt\ hidden attributes $H \cap (I^{j+1} \cup O^{j+1}))$.
Therefore their projection on visible attributes $I^{j+1} \setminus H = (I^{j+1} \cup O^{j+1}) \setminus H$ are the same -----------(A)\\

Partition $I^{j+1}$ into $I^j$ and $I^{j+1} \setminus I^j = I_{j+1} \setminus I^j$.
Projection of $x_1$ and $x_2$, let $x_{11}, x_{12}$, on $I^j \setminus H$ will be the same. 
Therefore, the inputs to $M^j$ are equivalent. By hypothesis, their outputs,
say $z_1, z_2$ will have same values on $O^j \setminus H = (I^{j+1} \cup O^{j+1}) \setminus H$ ----------- (B).\\

Again, on inputs $x_1, x_2$ to $M^{j+1}$, inputs to  $m_{j+1}$ will be concatenation of
(i) projection of output $z_1, z_2$ from $M^j$ on $O^j \cap I_{j+1}$
and (ii) projection of $x_1, x_2$ on $I_{j+1} \setminus I^j$.
From (A) and (B), they will be equivalent on visible attributes $(I^{j+1} \cup O^{j+1}) \setminus H$.
Therefore, the inputs to $m_{j+1}$ are equivalent \wrt\ $H$. 
Since $m_{j+1}$ is \UDS, the outputs, say $w_1, w_2$ are also equivalent ------------(C).\\

Now note that $y_1$ is defined on $O^{j+1} = (O^j \setminus I_{j+1}) \cup O_{j+1}$.
Its projection on $O^j \setminus I_{j+1}$ is projection of $z_1$ on $O^j \setminus I_{j+1}$, 
and its projection on $O_{j+1}$ is $z_1$.
Similarly $y_2$ can be partitioned.
From (B) and (C), the projections are equivalent, therefore the outputs $y_1$ and $y_2$ are equivalent.
\par
This shows that for two equivalent input the outputs are equivalent.
The other direction, for two equivalent outputs all of their inputs are
equivalent can be proved in similar way by considering modules in
reverse topological order from $m_{k}$ to $m_2$.
\end{proof}

\section{Proofs from Section~5}\label{sec:app-optimization}


\subsection{Proof of Theorem~\ref{thm:nphard-uds}} \label{sec:proof-thm:nphard-uds} 
\emph{\textsc{Theorem}~\ref{thm:nphard-uds}.~~ 
Given public module $m_j$ with $k$ attributes, and
a subset of hidden attributes 
$H$, deciding whether $m_j$ is \UDS\ \wrt 
$H$ is coNP-hard in $k$.
Further, all \UDS\ subsets can be found in EXP-time in $k$.
}\\

\smallskip
\noindent
\paragraph*{Proof of NP-hardness}
\begin{proof}
We do a reduction from UNSAT, where given $n$ variables $x_1, \cdots, x_n$, and a boolean formula $f(x_1, \cdots, x_n)$, the goal is to check
whether 
$f$ is \emph{not} satisfiable.
In our construction, $m_i$ has $n+1$ inputs $x_1, \cdots, x_n$ and $y$, and the output is $z$ = $m_i(x_1, \cdots, x_n, y)$ = $f(x_1, \cdots, x_n) \vee y$ (OR).
hidden attributes $H = \{x_1, \cdots, x_n\}$, so $y, z$ atr visible. We claim that $f$ is not satisfiable if and only
if $m_i$ is \UDS\ \wrt $H$.
\par
Suppose $f$ is not satisfiable, so for all assignments of $x_1, \cdots, x_n$, $ f(x_1, \cdots, x_n) = 0$.
For output $z = 0$, then the visible attribute $y$ must have 0 value in all the rows of the relation of $m_i$.
Also for $z = 1$, the visible attribute $y$ must have 1 value, since in all rows $ f(x_1, \cdots, x_n) = 0$.
Hence for equivalent inputs \wrt $H$, the outputs are equivalent and vice versa. Therefore $m_i$ is \UDS\ \wrt $H$.
\par
Now suppose $f$ is satisfiable, then there is at least one assignment of $x_1, \cdots, x_n$, such that $ f(x_1, \cdots, x_n) = 1$.
In this row, for $y = 0$, $z = 1$. However for all assignments of $x_1, \cdots, x_n$, whenever $y = 1$, $z = 1$.
So for output $z = 1$, all inputs producing $z$ are not equivalent \wrt the visible attribute $y$, therefore $m_i$
is not upstream-safe and hence not \UDS.
\end{proof}

\paragraph*{Upper Bound to Find All \UDS\ Solutions}
The lower bounds 
studied for the second step of the four step optimization show that for a public module $m_j$, it is not possible
to have poly-time algorithms (in $|A_j|$) even to decide if a given subset $H \subseteq A_j$ is \UDS, unless $P = NP$.
Here we present Algorithm~\ref{algo:uds} that finds \emph{all} \UDS\ solutions of $m_j$ is time exponential
in $k_j = |A_j|$, assuming that the maximum domain size of attributes $\Delta$ is a constant.
\par
\begin{algorithm}[ht]
\caption{Algorithm to find all \UDS\ solutions $\uds_j$ for a public module $m_j$}
\begin{algorithmic}[1] \label{algo:uds}
\STATE{-- Set $\uds_j = \emptyset$.}
\FOR{every subset $H$ of $A_j$} 
        \STATE{\textit{/*Check if $H$ is downstream-safe */}}
    \FOR{every assignment $\tup{x^+}$ of the visible input attributes in $I_j \setminus H$}
        \STATE{--Check if for every assignment $\tup{x^-}$ of the hidden input attributes in $I_j \cap V$, whether
                the value of $\proj{O_j \setminus H}{m_j(\tup{x})}$ is the same, where $\proj{I_j \setminus H}{\tup{x}} = \tup{x^+}$
                and $\proj{I_j \cap H}{\tup{x}} = \tup{x^-}$}
                \IF{not}
                    \STATE{-- \textbf{$H$ is \emph{not} downstream-safe.}}
                    \STATE{-- Go to the next $H$.}
                \ELSE
                   \STATE{-- \textbf{$H$ is downstream-safe.}}
                   \STATE{-- Let $\tup{y^+} = \proj{O_j \setminus  H}{m_j(\tup{x})}$ = projection of all such tuples that have projection = $\tup{x^+}$
                   on the visible input attributes}
                   \STATE{-- Label this set of input-output pairs $(\tup{x}, m_j(\tup{x}))$ by $\angb{\tup{x^+}, \tup{y^+}}$.}
                \ENDIF
                \STATE{}
            \STATE{\textit{/*Check if $H$ is upstream-safe */}}
            \STATE{-- Consider the pairs $\angb{\tup{x^+}, \tup{y^+}}$ constructed above.}
            \STATE{-- Let $n_1$ be the number of distinct $\tup{x^+}$ values, and let $n_2$ be the number of distinct $\tup{y^+}$ values/}
            \IF{$n_1 == n_2$}
                    \STATE{-- \textbf{$H$ is upstream-safe.}}
                \STATE{-- Add $H$ to $\uds_j$.}
            \ELSE
                    \STATE{-- \textbf{$H$ is \emph{not} upstream-safe.}}
                \STATE{-- Go to the next $H$.}
            \ENDIF
    \ENDFOR
\ENDFOR
\RETURN{The set of subsets $\uds_j$.}
\end{algorithmic}
\end{algorithm}

\textbf{Time complexity.~~} The outer for loop runs for all possible subsets of $A_j$, i.e. $2^{k_j}$
times. The inner for loop runs for maximum $\Delta^{|I_j \setminus H|}$ times (this is the maximum number of such tuples
$\tup{x^+}$), whereas the check if $H$ is a valid downstream-safe subset takes $O(\Delta^{|I_j \cap H|})$ time.
Here we ignore the time complexity to check equality of tuples which will take only polynomial in $|A_i|$
time and will be dominated by the exponential terms. For the upstream-safety check, the number of
$\angb{\tup{x^+}, \tup{y^+}}$ pairs are at most $\Delta^{|I_j \setminus H|}$, and to compute the distinct number of
$\tup{x^+}, \tup{y^+}$ tuples from the pairs can be done in $O(\Delta^{2|I_j \setminus H|})$ time by a naive search; the time complexity can be
improved by the use of a hash function. Hence the total time complexity is dominated by
$2^{k_j} \times O(\Delta^{|I_j \setminus H|}) \times O(\Delta^{|I_j \cap H|} + \Delta^{2|I_j \setminus H|})$
$= O(2^{k_j}\Delta^{3k^j})$ = $O(2^{4k_j})$. By doing a tighter analysis, the multiplicative factor in the exponents can be improved,
however,
we make the point here that the algorithm runs in time exponential in $k_j = |A_j|$.

\textbf{Correctness.~~} The following lemma proves the correctness of Algorithm~\ref{algo:uds}.
\begin{lemma}\label{lem:correct-uds-algo}
Algorithm~\ref{algo:uds} adds $H \subseteq A_j$ to $\uds_j$ if and only if $m_j$ is \UDS\ \wrt $H$.
\end{lemma}
\begin{proof}
(if) Suppose $H$ is a \UDS\ subset for $m_j$. Then $V$ is downstream-safe, i.e. for equivalent inputs \wrt the visible attributes
$I_j \setminus H$, the projection of the output on the visible attributes 
$O_j \setminus H$ will be the same, so
$H$ will pass the downstream-safety test.
\par
Since $H$ is \UDS, $H$ is also upstream-safe. Clearly, by definition, $n_1 \geq n_2$.
Suppose $n_1 > n_2$. Then there are two $\tup{x_1^+}$ and $\tup{x_2^+}$
that pair with the same $\tup{y^+}$. By construction, $\tup{x_1^+}$ and $\tup{x_2^+}$ (and all input tuples $\tup{x}$
to $m_j$ that project on these two tuples)  have different value on the visible input attributes
$I_j \setminus H$, but they map to outputs $\tup{y}$-s that have the same value on visible output attributes $O_j \setminus H$.
Then $H$ is not upstream-safe, which is a contradiction. Hence $n_1 = n_2$, and $H$ will also pass the
test for upstream-safety and be included in $\uds_j$.\\

(only if) Suppose $H$ is not \UDS, then it is either not upstream-safe or not downstream-safe.
Suppose it is not downstream-safe. Then for at least one assignment $\tup{x^+}$, the values of $\tup{y}$
generated by the assignments $\tup{x^-}$ will not be equivalent \wrt the visible output attributes, and the downstream-safety test will fail.
\par
Suppose $H$ is downstream-safe but not upstream-safe. Then there are Then there are two $\tup{x_1^+}$ and $\tup{x_2^+}$
that pair with the same $\tup{y^+}$. This makes $n_1 > n_2$, and the upstream-safety test will fail.
\end{proof}

\subsection{Correctness of Optimal Algorithm for Chain Workflows} 
\label{sec:proof-lem:correct-chain-dp}

Recall that after renumbering the modules, $m_1, \cdots, m_k$
denote the modules in the public closure $C$ of a private module $m_i$.
The following lemma  shows that $Q[j,\ell]$ correctly stores the desired value:
the cost of minimum cost hidden subset $H^{j\ell}$ that satisfies
the \UDS\ condition for all public modules $m_1$ to $m_j$, and 
$A_{j} \cap H^{j\ell} = U_{j\ell} \in \uds_j$.
%
	Recall that we use the simplified notations $S$ for the safe subset $S_{i\ell}$ of $m_i$, 
$C$ for public closure $C(S_{i\ell})$, and $H$ for $H_i$.

			\begin{lemma}\label{lem:correct-chain-dp}
			For $1 \leq j \leq k$, the entry $Q[j, \ell]$, $1 \leq \ell \leq
			p_j$,  stores the minimum cost of the hidden attributes $H^{j\ell}$
			such that $\cup_{x = 1}^j A_x \supseteq H^{j\ell}  \supseteq S$, 
			$A_{j} \cap H^{j\ell} = U_{j\ell}$, and every module $m_x, 1
			\leq x\leq j$ in the chain is \UDS\ \wrt $A_x \setminus
			H^{j\ell}$.
			\end{lemma}

The following proposition will be useful to prove the lemma.
\begin{proposition}\label{prop:uds-property}
For a public module $m_j$, for two \UDS\ hidden subsets $U_1, U_2 \subseteq A_j$,
if $U_1 \cap O_j = U_2 \cap O_j$, then $U_1 = U_2$.
\end{proposition}
The proof of the proposition is simple, and therefore is omitted.\\
\smallskip
\noindent
\textbf{Proof of Lemma~\ref{lem:correct-chain-dp}~~}

\begin{proof}
We prove this by induction from $j = 1$ to $k$. The base case follows by the definition of $Q[1, \ell]$,
for $1 \leq \ell \leq p_1$. Here the requirements are
$A_1 \supseteq H^{1\ell} \supseteq S$, and  $H^{1\ell} = U_{1\ell}$. 
So we set the cost at $Q[1, \ell]$ to $c(U_{1\ell}) = c(H^{1\ell})$, if $U_{1\ell} \supseteq S$ .
\par
Suppose the hypothesis holds until $j-1$, and consider $j$.
Let $H^{j\ell}$ be the minimum solution s.t. $A_j \cap H^{j\ell} = U_{j\ell}$
and satisfies the other conditions of the lemma.
\par
First consider the case when there is  no $q$
such that $U_{j-1, q} \cap O_{j-1} = U_{j, \ell} \cap I_{j}$, where we set the cost to be $\infty$.
If there is no such $q$. i.e. for all $q \leq p_{j-1}$,
then clearly there cannot be any solution $H^{j\ell}$
that contains $U_{j, \ell}$ and also guarantees \UDS\ properties of all $x < j$ (in particular for $x = j-1$).
In that case the cost of the solution is indeed $\infty$.
\par
Otherwise (when such a $q$ exists), let us divide the cost of the solution $c(H^{j\ell})$
into two disjoint parts:
$$c(H^{j\ell}) = c(H^{j\ell} \cap O_j) + c(H^{j\ell} \setminus O_j)$$
\par
We argue that 
$c(O_j \cap H^{j\ell}) = c(O_j \cap U_{j\ell})$.
$A_j \cap H^{j\ell} = U^{j\ell}$. 
Then $O_j \cap U_{j\ell}$ = $O_j \cap A_j \cap H^{j\ell}$ = $O_j \cap H^{j\ell}$, since $O_j \subseteq A_j$.
Hence $c(O_j \cap H^{j\ell}) = c(O_j \cap U_{j\ell})$. This accounts for the cost of the first part of $Q[j, \ell]$.
\par
Next we argue that $c(H^{j\ell} \setminus O_j) =$ minimum cost $Q[j-1, q]$, $1 \leq q \leq p_j$, where the minimum is over those
those $q$ where $U_{j-1, q} \cap O_{j-1} = U_{j, \ell} \cap I_{j}$.
Due to the chain structure of $C$, $O_j \cap \bigcup_{x = 1}^{j-1} A_j = \emptyset$, and $O_j \cup \bigcup_{x = 1}^{j-1} A_x = \bigcup_{x = 1}^{j} A_x$.
Since $\cup_{x = 1}^j A_x \supseteq H^{j\ell}$, $H^{j\ell} \setminus O_j = H^{j\ell} \cap \bigcup_{x = 1}^{j-1} A_x$.
\par
Consider $H' = H^{j\ell} \cap \bigcup_{x = 1}^{j-1} A_x$. By definition of $H^{j\ell}$, $H'$ must satisfy the \UDS\ requirement of all
$1 \leq x \leq j-1$. Further, $\bigcup_{x = 1}^{j-1} A_x \supseteq H'$.
$A_j \cap H^{j\ell} = U_{j, \ell}$, hence
 $U_{j, \ell}  
  \subseteq H^{j\ell}$.
  \par
We are considering the case where there is a $q$ such that
\begin{equation}
U_{j-1, q} \cap O_{j-1} = U_{j, \ell} \cap I_{j}\label{equn:uds-1}
\end{equation}
Therefore
 $$U_{j-1, q} \cap O_{j-1} \subseteq U_{j, \ell} \subseteq H^{j\ell}$$
 We claim that if $q$ satisfies (\ref{equn:uds-1}), then $A_{j-1} \cap H' = U_{j-1, q}$.
 Therefore, by induction hypothesis, $Q[j-1, \ell]$ stores the minimum cost solution $H'$ that includes $U_{j-1, q}$, and part of the
 the optimal solution cost $c(H^{j\ell} \setminus O_j)$ for $m_j$ is the minimum value of such $Q[j-1, q]$.
 \par
 So it remains to show that $A_{j-1} \cap H' = U_{j-1, q}$.
$A_{j-1} \cap H' = A_{j-1} \cap H^{j\ell} \in \uds_{j-1}$, since $H^{j\ell}$ gives \UDS\ solution for $m_{j-1}$.
 Suppose $A_{j-1} \cap H^{j\ell} = U_{j-1,y}$.
Then we argue that $U_{j-1, q} = U_{j-1, y}$, which will complete the proof.
\par
$U_{j-1}, y \cap O_{j-1}$  
= $(A_{j-1} \cap H^{j\ell}) \cap O_{j-1} = H^{j\ell} \cap O_{j-1}$,
$= H^{j\ell} \cap I_j = (A_j \cap U_{j, \ell}) \cap I_j$,
\ie\
\begin{equation}
U_{j-1}, y \cap O_{j-1} = U_{j, \ell} \cap I_{j} \label{equn:uds-2}
\end{equation}
From (\ref{equn:uds-1}) and (\ref{equn:uds-2}),
$$U_{j-1}, q \cap O_{j-1} = U_{j-1}, y \cap O_{j-1} $$
since both $U_{j-1, q}, U_{j-1, y} \in \uds_{j-1}$, from Proposition~\ref{prop:uds-property}, $U_{j-1, q} = U_{j-1, y}$.
This completes the proof of the lemma.
\end{proof}

\subsection{Optimal Algorithm for Tree Workflows}\label{sec:step3-tree-wf}
\eat{
			\textsc{Theorem}~\ref{prop:step3-tree-wf}.~~ \emph{ The
			single-subset problem can be solved in PTIME (in $n$, $|A|$ and $L$)
			for tree workflows. }
}
Here we prove Theorem~\ref{prop:step3-chain-wf} for tree workflows.

\paragraph*{Optimal algorithm for tree workflows}

Similar to the 
algorithm for chain workflows, to obtain
an algorithm of time polynomial in $L$ for tree workflows, for a given module $m_i$, we
can go over all choices of safe subsets $S_{i\ell} \in \safe_i$ of
$m_i$, compute the public-closure $C(S_{i\ell})$, and
choose a minimal cost subset $H_i = H_i(S_{i\ell})$ that satisfies
the \UDS\ properties of all modules in the public-closure. Then,
output, among them, a subset having the minimum cost. Consequently,
it suffices to explain how, given a safe subset $S_{i\ell} \in
\safe_i$, one can solve, in PTIME, the problem of finding a minimum
cost hidden subset $H_i$ that satisfies the \UDS\ property of all
modules in a subgraph formed by a given $C(S_{i\ell})$.

To simplify notations, the given safe subset $S_{i\ell}$ will be
denoted below by $S$, the closure $C(S_{i\ell})$ will
be denoted by $C$, and the output hidden subset $H_i$ will be denoted by
$H$.
Our PTIME algorithm uses dynamic programming to find the optimal
$H$.
\par
First note that since $C$ is the public-closure of (some) output
attributes for a tree workflow, $C$ is a collection of trees all of which are rooted at the private module
$m_i$. 
Let us consider the tree $T$ rooted at $m_i$ with subtrees in $C$,
(note that $m_i$ can have private children that are not considered in $T$).
Let 
$k$ be the number of modules in $T$, and the modules in $T$ be renumbered as $m_i, m_1, \cdots, m_k$,
where the private module $m_i$ is the root, and the rest are public modules.
\par
Now we solve the problem by dynamic programming as follows.
Let $Q$
be an $k \times L$ two-dimensional array, where $Q[j, \ell]$, $1 \leq j\leq k, 1 \leq \ell \leq p_j$ denotes
the cost of minimum cost hidden subset $H^{j\ell}$ that (i) satisfies
the \UDS\ condition for all public modules 
in the subtree of $T$ rooted at $m_j$, that we denote by $T_j$;
and, (ii) $H^{j\ell} \cap A_j = U_{j\ell}$.
(recall that $I_j O_j, A_{j}$ is the set of input, output and all attributes of $m_j$ respectively);
the actual solution can be stored easily by standard
argument. The algorithm is described below.

\begin{itemize}
\item \textbf{Initialization for leaf nodes}.~~
The initialization step handles all leaf nodes $m_j$ in $T$. For a leaf node $m_j$,
$1 \leq \ell \leq p_j$,
\begin{eqnarray*}
Q[j, \ell] & = & c(U_{j, \ell})
\end{eqnarray*}
\item \textbf{Internal nodes}.~~
The internal nodes are considered in a bottom-up fashion (by a post-order traversal),
and $Q[j, \ell]$ is computed for a node $m_j$ after its children are processed.

For an internal node $m_j$, let $m_{i_1}, \cdots, m_{i_x}$ be its children in $T$. Then for $1 \leq \ell \leq p_j$,
\begin{enumerate}
    \item Consider \UDS\ subset $U_{j,\ell}$.
    \item For $y = 1$ to $x$, let $U^y = U_{j, \ell} \cap I_{i_y}$ 
    Since there is no data sharing, $U^y$-s are disjoint
    \item For $y = 1$ to $x$,
        \begin{eqnarray*}
        k^y & = & \rm{argmin}_k Q[i_y, k] ~~~\text{ where the minimum is over} \\
        & & ~~~~ 1 \leq k \leq p_{i_y}~~~ \text{ s.t. }~~~ U_{i_y, k} \cap I_{i_y} = U^y\\
        & = & \perp \text{ (undefined),~~~~~~~ if there is no such } k
        \end{eqnarray*}

    \item $Q[j, \ell]$ is computed as follows.
        \begin{eqnarray*}
        Q[j, \ell] & = & \infty ~~~~~~~~~~~~~~~ \text{if } \exists y, 1 \leq y \leq x, ~~~k^y = \perp\\
        & = & c(I_j \cap U_{j\ell}) + \sum_{y = 1}^x Q[i_y, k^y]~~~~~~~~(\text{otherwise})
        \end{eqnarray*}
\end{enumerate}

\item \textbf{Final solution for $S$}.~~ Now consider the private module $m_i$ that is the root of $T$. Recall that we
fixed a safe solution $S$ of $m_i$ for doing the analysis. Let $m_{i_1}, \cdots, m_{i_x}$ be  the children of $m_i$ in $T$
(which are public modules). Similar to the step before, we consider the min-cost solutions of its children which
exactly match the hidden subset $S$ of $m_i$.
\begin{enumerate}
    \item Consider safe subset $S$ of $m_i$.
    \item For $y = 1$ to $x$, let $S^y = S \cap I_{i_y}$ 
    Since there is no data sharing, again, $S^y$-s are disjoint
    \item For $y = 1$ to $x$,
        \begin{eqnarray*}
        k^y & = & \rm{argmin}_k Q[i_y, k] ~~~\text{ where the minimum is over} \\
        & & ~~~~ 1 \leq k \leq p_{i_y}~~~ \text{ s.t. }~~~ U_{i_y, k} \cap I_{i_y} \supseteq S^y\\
        & = & \perp \text{ (undefined),~~~~~~~ if there is no such } k
        \end{eqnarray*}
    \item The cost of the optimal $H$ (let us denote that by $c^*$) is computed as follows.
        \begin{eqnarray*}
        c^* & = & \infty ~~~~~~~~~~~~~~~ \text{if } \exists y, 1 \leq y \leq x, ~~~k^y = \perp\\
        & = & \sum_{y = 1}^x Q[i_y, k^y]~~~~~~~~(\text{otherwise})
        \end{eqnarray*}
\end{enumerate}
\end{itemize}

It is not hard to see that the trivial solution of \UDS\ subsets that include all attributes of the modules
gives a finite cost solution by the above algorithm.


Lemma~\ref{lem:correct-tree-dp} stated and proved below
shows that $Q[j,
\ell]$ correctly stores the desired value. 
Given this lemma, the correctness of the algorithm easily follows. For hidden subset $H \supseteq S$
in the closure, for every public child $m_{i_y}$ of $m_i$, $H \cap I_{i_y} \supseteq S \cap I_{i_y} = S^y$.
Further, each such $m_{i_y}$ has to be \UDS\ \wrt $H^{j\ell}$. In other words, for each $m_{i_y}$, $H \cap I_{i_y}$
must equal $U_{i_y, k^y}$ for some $1 \leq k^y \leq p_{i_y}$. The last step in our algorithm (that computes $c^*$)
tries to find such a $k^y$ that has the minimum cost $Q[i_y, k^y]$, and the total cost $c^*$
of $H$ is $\sum_{m_{i_y}} Q[i_y, k^y]$ where the sum is over all children of $m_i$ in the tree $T$
(the trees rooted at $m_{i_y}$ are disjoint, so the optimal cost $c^*$ is sum of those costs).
This proves Theorem~\ref{prop:step3-chain-wf} for tree workflows \\


\textbf{$Q[j, \ell]$ stores correct values.~~}
The following lemma shows that the algorithm stores correct values in $Q[j, \ell]$
for all public modules $m_j$ in the closure $C$.

\begin{lemma}\label{lem:correct-tree-dp}
For $1 \leq j \leq k$, let $T_j$ be the subtree rooted at $m_j$ and let $\attr_j = \bigcup_{m_q \in T_j} A_q$.
The entry $Q[j, \ell]$, $1 \leq \ell \leq
p_j$,  stores the minimum cost of the hidden attributes $H^{j\ell} \subseteq \attr_j$ 
such that
$A_j \cap H^{j\ell} = U_{j\ell}$, and every module $m_q \in T_j$, is \UDS\ \wrt $A_q \cap H^{j\ell}$.
\end{lemma}

To complete the proof of Theorem~\ref{prop:step3-chain-wf} for tree workflows, we need
to prove Lemma~\ref{lem:correct-tree-dp}, that we prove next.

\begin{proof} 
We prove the lemma by an induction on all nodes at depth $h = H$ down to 1 of the tree $T$, where
depth $H$ contains all leaf nodes and depth 1 contains the children of the root $m_i$ (which is at depth 0).
\par
First consider any leaf node $m_j$ at height $H$. Then $T_j$ contains only $m_j$ and $\attr_j = A_j$.
For any $1 \leq \ell \leq p_j$, since $\attr_j = A_j \supseteq H^{j\ell}$ and $A_j \cap H^{j\ell} = U_{j, \ell}$.
In this case $H^{j\ell}$ is unique and $Q[j, \ell]$ correctly stores
$c(U_{j, \ell}) = c(H^{j\ell})$.
\par
Suppose the induction holds for all nodes up to height $h + 1$, and consider a node $m_j$ at height
$h$. Let $m_{i_1}, \cdots, m_{i_x}$ be the children of $m_j$ which are at height $h+1$.
Let $H^{j\ell}$ be the min-cost solution, which is partitioned into two disjoint component:

$$c(H^{j\ell}) = c(H^{j\ell} \cap I_j) + c(H^{j\ell} \setminus I_j)$$
\par
First we argue that $c(H^{j\ell} \cap I_j) = c(U_{j, \ell})$.
$A_j \cap H^{j\ell} = U^{j\ell}$. 
Then $I_j \cap U_{j\ell}$ = $I_j \cap A_j \cap H^{j\ell}$ = $I_j \cap H^{j\ell}$, since $I_j \subseteq A_j$.
Hence $c(I_j \cap H^{j\ell}) = c(I_j \cap U_{j\ell})$. This accounts for the cost of the first part of $Q[j, \ell]$.
\par
Next we analyze the cost $c(H^{j\ell} \setminus I_j)$.
This cost comes from the subtrees $T_{i_1}, \cdots, T_{i_x}$ which are disjoint due to the tree structure and absence of
data-sharing.
Let us partition the subset $H^{j\ell} \setminus I_j$
into disjoint parts $(H^{j\ell} \setminus I_j) \cap \attr_{i_y}$, $1 \leq y \leq x$.
Below we prove that $c((H^{j\ell} \setminus I_j) \cap \attr_{i_y}) = Q[i_y, k^y]$, $1 \leq y \leq x$,
where $k^y$ is computed as in the algorithm. This will complete the proof of the lemma.
\par
To see this, let $H' = (H^{j\ell} \setminus I_j) \cap \attr_{i_y}$.
Clearly, $\attr_{i_y} \supseteq H'$. Every $m_q \in T_j$ is \UDS\ \wrt $A_q \cap H^{j\ell}$.
If also $m_q \in T_{i_y}$, then $A_q \cap H' = A_q \cap H^{j\ell}$, and therefore all $m_q \in T_{i_y}$
are also \UDS\ \wrt $H'$. In particular, $m_{i_y}$ is \UDS\ \wrt $H'$, and therefore $A_{i_y} \cap H' = U_{i_y, k^y}$
for some $k^y$,
since $U_{i_y, k^y}$ was chosen as the \UDS\ set by our algorithm.
\par
Finally we argue that $c(H') = c(H^{i_y, k^y})$, where $H^{i_y, k^y}$ is the min-cost solution for $m_{i_y}$
among all such subsets. This follows from our induction hypothesis, since $m_{i_y}$ is a node at depth $h+1$.
Therefore, $c(H') = c(H^{i_y, k^y}) = Q[i_y, k^y]$, \ie\
$$c((H^{j\ell} \setminus I_j) \cap \attr_{i_y}) = Q[i_y, k^y]$$
as desired.
This 
proves the lemma.
 \end{proof}

\subsection{Proof of NP-hardness for DAG Workflows}\label{sec:proof-lem:single-priv-soln-nphard}
Here we prove NP-hardness for arbitrary DAG workflows as stated in Theorem~\ref{prop:step3-chain-wf}
by a reduction from 3SAT.

\begin{figure}[t]
\centering
\includegraphics[scale=0.3]{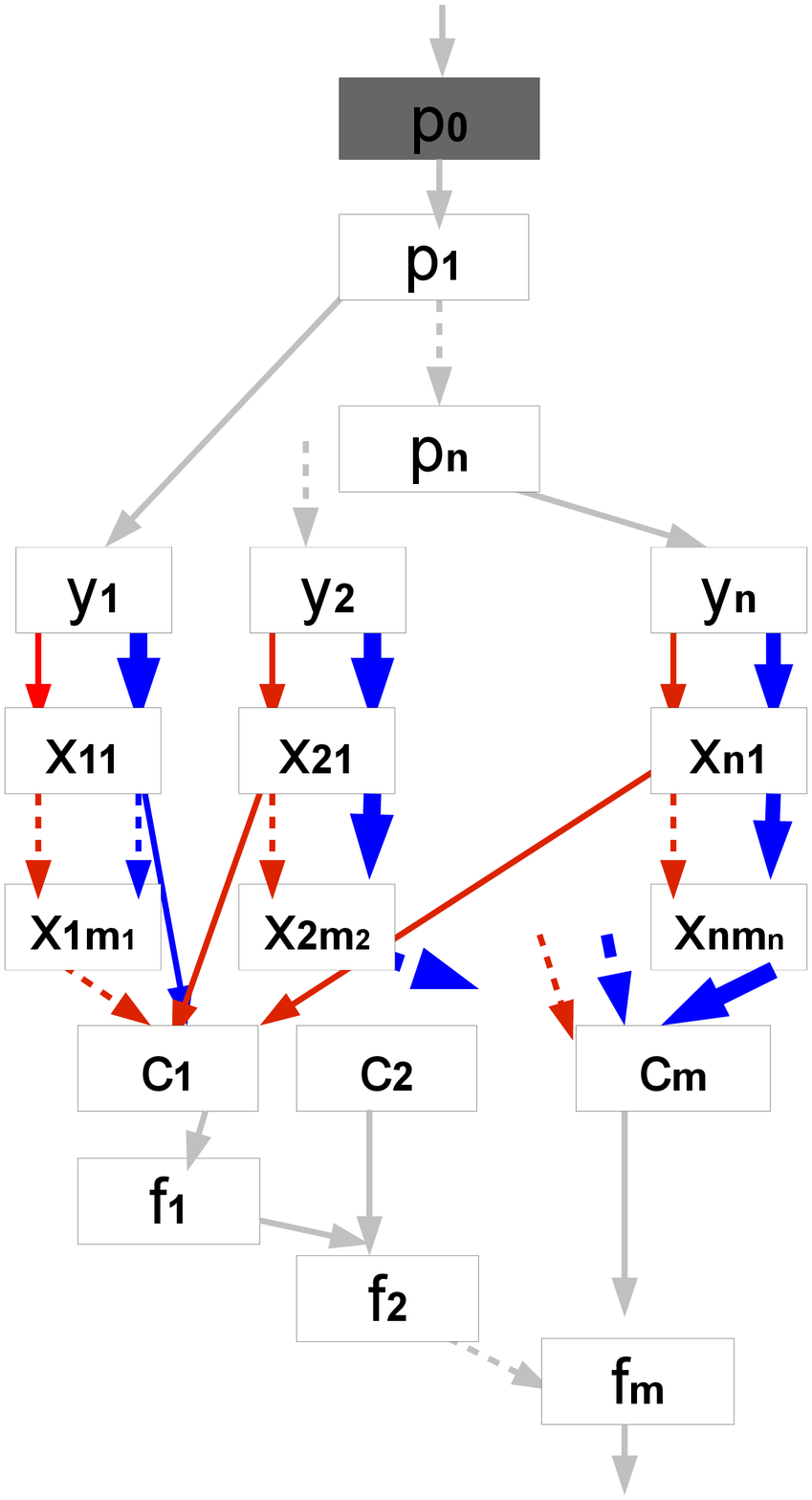}
\caption{Reduction from 3SAT.  White modules are public, Grey are private.
Red thin edges denote TRUE assignment, Blue bold edges denote FALSE assignment.
} \label{fig:wf-nphard}
\end{figure}

Given a CNF formula $\psi$ on $n$ variables $z_1, \cdots, z_n$ and $m$ clauses $\psi_1, \cdots, \psi_m$,
we construct a graph as shown in Figure~\ref{fig:wf-nphard}. Let variable $z_i$ occurs in $m_i$ different clauses
(as positive or negative literals). In the figure, the module $p_0$
is the single-private module ($m_i$), having a single output attribute $a$. The rest of the modules are the public modules in the
public-closure $C(\{a\})$.
\par
For every variable $z_i$, we create $m_i +2$ nodes: $p_i, y_i$ and $x_{i,1}, \cdots, x_{i,m_i}$. For every clause $\psi_j$, we create
$2$ modules $C_j$ and $f_j$.
\par
The \textbf{edge connections} are as follows:
\begin{enumerate}
    \item $p_0$ sends its single output $a$ to $p_1$.
    \item For every $i = 1$ to $n-1$, $p_i$ has two outputs; one is sent to $p_{i+1}$ and the other is sent to $y_{i}$.
    $p_n$ sends its single output to $y_{n}$.
    \item Each $y_i$, $i = 1$ to $n$, sends two outputs to $x_{i,1}$. The blue outgoing edge from $y_i$ denotes positive assignment of the variable $z_i$,
    whereas the red edge denotes negative assignment of the variable $z_i$.
    \item Each $x_{i,j}$, $i = 1$ to $n$, $j = 1$ to $m_i-1$, sends two outputs (blue and red) to $x_{i,j+1}$. In addition, if $x_{i,j}$, $i = 1$ to $n$, $j = 1$ to $m_i$ sends a blue (resp. red) edge to clause node $C_k$ if the variable $z_i$ is a positive (resp. negative) in the clause $C_k$
    (and $C_k$ is the $j$-th such clause containing $z_i$).
    \item Each $C_j$, $j = 1$ to $m$, sends its single output to $f_j$.
    \item Each $f_j$, $j = 1$ to $m-1$, sends its single output to $f_{j+1}$, $f_m$ outputs the single final output.
\end{enumerate}
\par
The \UDS\ sets are defined as follows:
\begin{enumerate}
    \item For every $i = 1$ to $n-1$, $p_i$ has a single \UDS\ set: hide all its inputs and outputs.
    \item Each $y_i$, $i = 1$ to $n$, has three \UDS\ choices: (1)    hide its unique input and blue output, (2)  hide its unique input and red output,
    (3) hide its single input and both blue and red outputs.
    \item Each $x_{i,j}$, $i = 1$ to $n$, $j = 1$ to $m_i$, has three choices: (1) hide blue input and all blue outputs,
    (2) hide red input and all red outputs, (3) hide all inputs and all outputs.
    \item Each $C_j$, $j = 1$ to $m$, has choices: hide the single output and at least of the three inputs.
    \item Each $f_j$, $j = 1$ to $m$, has the single choice: hide all its inputs and outputs.
\end{enumerate}
\par
\textbf{Cost.~~} The outputs from $y_i$, $i = 1$ to $n$ has unit cost, the cost of the other attributes is 0.
The following lemma proves correctness of the construction.

\begin{lemma}
There is a solution of single-module problem of cost $= n$ if and only if the 3SAT formula $\psi$
is satisfiable.
\end{lemma}

\begin{proof}
(if) Suppose the 3SAT formula is satisfiable, so there is an assignment of the variables $z_i$ that makes $\Psi$ true.
If $z_i$ is set to \true\ (resp \false), choose the blue (resp. red) outgoing edge from $y_i$. Then choose the other edges
accordingly: (1) choose outgoing edge from $p_0$, (2) choose all input and outputs of $p_i$, $i = 1$ to $n$;
(3) if blue (resp. red) input of $x_{i,j}$ is  chosen, all its blue (resp. red) outputs  are chosen; and,
(4) all inputs and outputs of $f_j$ are chosen. Clearly, all these are \UDS\ sets by construction.

So we have to only argue about the clause nodes $C_j$. Since $\psi$ is satisfied by the given assignment, there is a
literal $z_i \in C_j$ (positive or negative), whose assignment makes it true. Hence at least one of the inputs to $C_j$
will be chosen. So the \UDS\ requirements of all the \UDS\ clauses are satisfied.
The total cost of the solution is $n$ since exactly one output of the $y_i$ nodes, $i = 1$ to $n$, have been chosen.
\\

(only if) Suppose there is a solution to the single-module problem of cost $n$. Then each $y_i$ can choose exactly one
output (at least one output has to be chosen to satisfy \UDS\ property for each $y_i$, and more than one output
cannot be chosen as the cost is $n$). If $y_i$ chooses blue (resp. red) output, this forces the $x_{i,j}$ nodes to select the corresponding
blue (resp. red) inputs and outputs. No $x_{i, j}$ can choose the \UDS\ option of selecting all its inputs and outputs as in that case
finally $y_i$ be forced to select both outputs which will exceed the cost.  Since $C_j$ satisfies \UDS\ condition,
this in turn forces each $C_j$ to select the corresponding blue (resp. red) inputs.
\par
If the blue (resp. red) output of $y_i$ is chosen, the variable is set to \true\ (resp. \false). By the above argument,
at least one such red or blue input will be chosen as input to each $C_j$, that satisfies the corresponding
clause $\psi_j$.
\end{proof}

\section{General Workflows}\label{sec:app-general-wf}
In this section we discuss the privacy theorem for general workflows as outlined in Section~\ref{sec:general-wf}.
First we define directed-path and downward-closure as follows 
(similar to public path and public-closure).
\begin{definition}\label{def:directedpath}
A module $m_1$ has \emph{\textbf{a directed}} path to another module $m_2$, if there are  
modules $m_{i_1},m_{i_2}, \cdots, m_{i_j}$ such that $m_{i_1}=m_1$,
$m_{i_j}=m_2$, and for all $1\leq k < j$, $O_{i_k} \cap I_{i_{k+1}}
\neq \emptyset$.

An attribute $a \in A$ has a directed path from to module $m_j$, if there is a module $m_k$
such that $a \in I_k$ and $m_k$ has a directed path to $m_j$.
\end{definition}

\begin{definition}\label{def:downward-closure}
Given a private module $m_i$ and a set of hidden output attributes
$h_i \subseteq O_i$ of $m_i$, the \emph{\textbf{downward-closure}} of
$m_i$ \wrt $h_i$, denoted by $D(h_i)$, is the 
set of modules $m_j$ (both private and public) such that there 
is a directed path from some attribute $a \in h_i$ to $m_j$.
\end{definition}

Also recall downstream-safety (\DS ty) defined in Definition~\ref{def:uds}
which says that for equivalent inputs to a module with respect to 
hidden attributes, the outputs must be equivalent.
We prove the following theorem in this section:\\

\begin{theorem}\label{thm:privacy-general} 
({\bf Privacy Theorem for General workflows}) 
Let $W$ be any workflow. For each private module 
$m_i$ in $W$, let  $H_i$ be a subset of hidden attributes such that
(i) $h_i = H_i \cap O_i$ is safe for $\Gamma$-standalone-privacy of $m_i$, 
(ii) each private and public module $m_j$ in the downward-closure 
$D(h_i)$
is \DS\ \wrt\ $A_j \cap H_i$, 
and (iii) $H_i \subseteq O_i \cup \bigcup_{j: m_j \in
D(h_i)} A_j$. Then the workflow $W$ is
$\Gamma$-private \wrt\ 
$H = \bigcup_{i: m_i \textrm{ is private}}H_i$.

\end{theorem}

In the proof of Theorem~\ref{thm:privacy-downward} from Lemma~\ref{lem:main-private-single-pred},
we used the fact that for single-predecessor workflows, for two distinct private modules $m_i, m_k$,
the public-closures and the hidden subsets $H_i, H_j$ are disjoint. 
However, it is not hard to see that 
this is not the case for general workflows, where the downward-closure and the subsets $H_i$ may overlap.
Further, the \DS\ property is not monotone (hiding more output attributes will maintain the \DS\ property, but hiding more input attributes
may destroy the \DS\ property). So we need to argue that the \DS\ property is maintained when we take union of $H_i$
sets in the workflow which is formalized by the following lemma.

\begin{lemma}\label{lem:ds-union}
If a module $m_j$ is \DS\ \wrt sets $H_1, H_2 \subseteq A_j$, then $m_j$ is \DS\ \wrt $H = H_1 \cup H_2$.
\end{lemma}
Given two equivalent inputs $\tup{x_1} \equiv_{H} \tup{x_2}$ \wrt $H = H_1 \cup H_2$, we have to show that their
outputs are equivalent: $m_j(\tup{x_1}) \equiv_{H} m_j(\tup{x_2})$. Even if $\tup{x_1}, \tup{x_2}$ are equivalent \wrt $H$, 
they may not be equivalent \wrt $H_1$ or $H_2$. In the proof we construct a new tuple $\tup{x_3}$ such that
$\tup{x_1} \equiv_{H_1} \tup{x_3}$, and $\tup{x_2} \equiv_{H_2} \tup{x_3}$. Then using the \DS\ properties of $H_1$ and $H_2$,
we show that $m_j(\tup{x_1}) \equiv_{H} m_j(\tup{x_3}) \equiv_{V} m_j(\tup{x_2})$.
The formal proof is given below.
\begin{proof}
Let $H = H_1 \cup H_2$. 
Let $\tup{x_1}$ and $\tup{x_2}$ be two input tuples to $m_j$ such that $\tup{x_1} \equiv_{H} \tup{x_2}$.
i.e. 
\begin{equation}
\proj{I_j \setminus H}{\tup{x_1}} = \proj{I_j \setminus H}{\tup{x_2}}\label{equn:ds-1}
\end{equation}
%
 For $a \in I_j$, let $x_3[a]$ denote the value of $a$-th attribute
of $x_3$ (similarly $x_1[a], x_2[a]$). From (\ref{equn:ds-1}), for $a \in  I_j \setminus H$, $x_1[a] = x_2[a]$.
Let us define a tuple $\tup{x_3}$ as follows on four disjoint subsets of $I_j$: 
\begin{eqnarray*}
x_3[a] & = & x_1[a]~~~~~ \text{ if } a \in I_j \cap H_1 \cap H_2\\ 
& = & x_1[a]~~~~~ \text{ if } a \in I_j \cap (H_2 \setminus H_1)\\ 
& = & x_2[a]~~~~~ \text{ if } a \in I_j \cap (H_1 \setminus H_2)\\ 
& = & x_1[a] = x_2[a] ~~~~~ \text{ if } a \in I_j \setminus H 
\end{eqnarray*}
For instance, on attribute set $I_j = \angb{a_1, \cdots, a_5}$, let
$\tup{x_1} = \angb{\ul{2},\ul{3},\ul{2},6,7}$, $\tup{x_2} = \angb{\ul{4}, \ul{5}, \ul{9}, 6, 7}$, 
$H_1 = \{a_1, a_2\}$ and $H_2 = \{a_2, a_3\}$, $H = \{a_1, a_2, a_3\}$
(in $\tup{x_1}, \tup{x_2}$, the hidden attribute values in $H$ are underlined).
Then $\tup{x_3} = \angb{\ul{4}, \ul{3}, \ul{2}, 6, 7}$.\\

(1) First we claim that, $\tup{x_1} \equiv_{H_1} \tup{x_3}$, or,  
\begin{equation}
\proj{I_j \setminus H_1}{\tup{x_1}} = \proj{I_j \setminus H_i}{\tup{x_3}} \label{equn:ds-2}
\end{equation}
Partition $I_j \setminus H_1$ into two disjoint subsets, $I_j \cap (H_2 \setminus H_1)$, and, $I_j \setminus (H_1 \cup H_2) = I_j \setminus H$.
From the definition of $\tup{x_3}$, for all $a \in I_j \cap (H_2 \setminus H_1)$ and all $a \in I_j \setminus H$,
$x_1[a] = x_3[a]$. This shows (\ref{equn:ds-2}).\\

(2) Next we claim that, $\tup{x_2} \equiv_{H_2} \tup{x_3}$, or,  
\begin{equation}
\proj{I_j \setminus H_2}{\tup{x_2}} = \proj{I_j \setminus H_2}{\tup{x_3}} \label{equn:ds-3}
\end{equation}
Again partition $I_j \setminus H_2$ into two disjoint subsets, $I_j \cap (H_1 \setminus H_2)$, and, $I_j \setminus (H_1 \cup H_2) = I_j \setminus H$.
From the definition of $\tup{x_3}$, for all $a \in I_j \cap (H_1 \setminus H_2)$ and all $a \in I_j \setminus H$,
$x_2[a] = x_3[a]$. This shows (\ref{equn:ds-3}). (\ref{equn:ds-2}) and (\ref{equn:ds-3}) can also be verified from the above example.\\

(3) Now by the condition stated in the lemma, $m_j$ is \DS\ \wrt $H_1$ and $H_2$. Therefore, using (\ref{equn:ds-2}) and (\ref{equn:ds-3}),
$m_j(\tup{x_1}) \equiv_{H_1} m_j(\tup{x_3})$ and $m_j(\tup{x_2}) \equiv_{H_2} m_j(\tup{x_3})$
or,  
\begin{equation}
\proj{O_j \setminus H_1}{m_j(\tup{x_1})} = \proj{O_j \setminus H_1}{m_j(\tup{x_3})} \label{equn:ds-4}
\end{equation}
and
\begin{equation}
\proj{O_j \setminus H_2}{m_j(\tup{x_2})} = \proj{O_j \setminus H_2}{m_j(\tup{x_3})} \label{equn:ds-5}
\end{equation}

Since $O_j \setminus H$ $= O_j \setminus (H_1 \cup H_2)$ $\subseteq O_j \setminus H_1$, from (\ref{equn:ds-4})
\begin{equation}
\proj{O_j \setminus H}{m_j(\tup{x_1})} = \proj{O_j \setminus H}{m_j(\tup{x_3})} \label{equn:ds-6}
\end{equation}

Similarly, $O_j \setminus H$ $\subseteq O_j \setminus H_2$, from (\ref{equn:ds-5})
\begin{equation}
\proj{O_j \setminus H}{m_j(\tup{x_2})} = \proj{O_j \setminus H}{m_j(\tup{x_3})} \label{equn:ds-7}
\end{equation}

From (\ref{equn:ds-6}) and (\ref{equn:ds-7}),
\begin{equation}
\proj{O_j \setminus H}{m_j(\tup{x_1})} = \proj{O_j \setminus H}{m_j(\tup{x_2})} \label{equn:ds-8}
\end{equation}
In other words, the output tuples $m_j(\tup{x_1}), m_j(\tup{x_2})$, that are defined on attribute set $O_j$,
\begin{equation}
m_j(\tup{x_1}) \equiv_{H} m_j(\tup{x_2}) \label{equn:ds-9}
\end{equation}
Since we started with two arbitrary input tuples $\tup{x_1} \equiv_{V} \tup{x_2}$, this shows that for all 
equivalent input tuples the outputs are also equivalent. In other words,
$m_j$ is \DS\ \wrt $H = H_1 \cup H_2$.
\end{proof}

Along with this lemma, two other simple observations will be useful.
\begin{observation}\label{obs:ds}
\begin{enumerate}
	\item Any module $m_j$ is \DS\ \wrt 
	$\emptyset$ (hiding nothing maintains downstream-safety property).
	\item If $m_j$ is \DS\ \wrt $H$, and if $H'$ is such that $H \subseteq H'$, but $I_j \setminus H' = I_j \setminus H$, then $m_j$
	is also \DS\ \wrt $H'$ (hiding more output attributes maintains downstream-safety property).
\end{enumerate}
\end{observation}

\subsection{Main Lemma for Privacy Theorem for General Workflows}\label{sec:proof-lem:main-private-general}
The following lemma is the crucial component in the proof of Theorem\ref{thm:privacy-general},
and is analogous to Lemma~\ref{lem:main-private-single-pred} for single-predecessor workflows.

\begin{lemma}\label{lem:main-private-general}
Consider a standalone private module $m_i$, a set of 
hidden attributes $h_i$, any input $\tup{x}$ to $m_i$, 
and any candidate output $\tup{y} \in \Out_{\tup{x}, m_i, h_i}$ of $\tup{x}$.
Then $\tup{y} \in \Out_{\tup{x}, W, H_i}$ when $m_i$ belongs to an arbitrary (general) workflow $W$,  
and a set attributes $H_i \subseteq A$ is hidden 
such that (i) $h_i \subseteq H_i$, (ii) only output attributes from $O_i$ are included in 
$h_i$ (i.e. $h_i \subseteq O_i$),
and (iii) every module $m_j$ in the downward-closure $D(h_i)$ is \DS\ \wrt $A_j \cap H_i$.

\end{lemma}

\begin{proof}
We fix a module $m_i$, an input $\tup{x}$ to $m_i$, a set of safe hidden 
attributes
$h_i$, and a
candidate output $\tup{y} \in \Out_{\tup{x}, m_i, h_i}$ for $\tup{x}$. 
For simplicity, let us refer to the set of modules in $D(h_i)$
by $D$.
We will show that $\tup{y} \in \Out_{\tup{x}, W, H_i}$ 
where the hidden attributes $H_i$ satisfies the conditions in the lemma.
In the proof, we show the existence of a possible world $R' \in \Worlds(R, H_i)$, 
such that if $\proj{I_i}{\tup{t}} = \tup{x}$ for some $\tup{t} \in R'$, then $\proj{O_i}{\tup{t}} = \tup{y}$.
Since $\tup{y} \in \Out_{x, m_i, h_i}$, by Lemma~\ref{lem:out-x-output}, $\tup{y} \equiv_{h_i} \tup{z}$ where $\tup{z} = m_i(\tup{x})$.


We will use the $\Flip$ function used in the proof of Lemma~\ref{lem:main-private-single-pred}
(see Appendix~\ref{sec:proof-lem:main-private-single-pred}).
We redefine the module $m_i$ to $\widehat{m}_i$ as follows. 
For an input $\tup{u}$ to $m_i$, 
$\widehat{m}_i(\tup{u}) = \Flip_{\tup{y}, \tup{z}}(m_i(\tup{u}))$.
All other public and private modules are unchanged, $\widehat{m}_j = m_j$.
The required possible world $R'$ is obtained by taking the join of the standalone relations of these $\widehat{m}_j$-s, $j \in [n]$.

First note that by the definition of $\widehat{m}_i$, $\widehat{m}_i(\tup{x}) = \tup{y}$ (since $\widehat{m}_i(x) = \Flip_{\tup{y}, \tup{z}}(m_i(x)) = \Flip_{\tup{y}, \tup{z}}(\tup{z}) = \tup{y}$,
from Observation~\ref{obs:flip}).
Hence if $\proj{I_i}{\tup{t}} = \tup{x}$ for some $\tup{t} \in R'$, then $\proj{O_i}{\tup{t}} = \tup{y}$.

Next we argue that  $R' \in \Worlds(R, H_i)$. Since $R'$ is the join of the standalone relations for 
modules $\widehat{m}_j$-s, $R'$ maintains all functional dependencies $I_j \rightarrow O_j$. 
Also none of the public modules are unchanged, hence for any public module $m_j$ and any tuple $t$
in $R'$, $\proj{O_j}{\tup{t}} = m_j(\proj{I_j}{\tup{t}})$. So we only need to show that 
the projection of $R$ and $R'$ on the visible attributes are the same. 


%
Let us assume, wlog. that the modules are numbered in topologically sorted order.
Let $I_0$ be the initial input attributes to the workflow, and let $p$ be a tuple defined on $I_0$.
There are two unique tuples $\tup{t} \in R$ and $\tup{t'} \in R'$
such that $\proj{I_1}{\tup{t}} = \proj{I_1}{\tup{t'}} = \tup{p}$.
Note that any intermediate or final attribute $a \in A \setminus I_0$ belongs to $O_j$, for a unique $j \in [1, n]$
(since for $j \neq \ell$, $O_j \cap O_{\ell} = \phi$).
So it suffices to show that $t, t'$ projected on $O_j$ are equivalent \wrt visible attributes for all 
module $j$, $j = 1$ to $n$.
\par
Let $\tup{c}_{j, m}, \tup{c}_{j, \widehat{m}}$ be the values of input attributes $I_j$ and $\tup{d}_{j, m}, \tup{d}_{j, \widehat{m}}$
be the values of output attributes $O_j$ of module $m_j$, in $\tup{t} \in R$ and $\tup{t'} \in R'$
respectively on initial input attributes $\tup{p}$
(i.e. $\tup{c}_{j, m} = \proj{I_j}{\tup{t}}$, $\tup{c}_{j, \widehat{m}} = \proj{I_j}{\tup{t'}}$,
$\tup{d}_{j, m} = \proj{O_j}{\tup{t}}$ and $\tup{d}_{j, \widehat{m}} = \proj{O_j}{\tup{t'}}$).
We prove by induction on $j = 1$ to $n$ that 

\begin{eqnarray}
\tup{d}_{j, \widehat{m}} & \equiv_{H_i} &  \tup{d}_{j, m} ~~~~~~~~~~~\text{if } j = i \text{ or } m_j \in D \label{equn:gen-flip-1}\\
\tup{d}_{j, \widehat{m}} & = & \tup{d}_{j, m}~~~~~~~~~~~~~~~\text{otherwise}\label{equn:gen-flip-2}
\end{eqnarray}

If the above is true for all $j$, then $\proj{O_j}{\tup{t}} \equiv_{H_i} \proj{O_j}{\tup{t}}$, along with the fact that the initial inputs
$\tup{p}$ are the same, this implies that $\tup{t} \equiv_{H_i} \tup{t'}$.

\textbf{Proof of (\ref{equn:gen-flip-1}) and (\ref{equn:gen-flip-2}).~~} The base case follows for $j = 1$.
If $m_1 \neq m_i$ ($m_j$ can be public or private), then $I_1 \cap O_i = \emptyset$, so  for all input $\tup{u}$,
$\widehat{m}_j(\tup{u}) = m_j(\Flip_{\tup{y}, \tup{z}}(\tup{u})) = m_j(\tup{u})$. Since the inputs $\tup{c}_{1, \widehat{m}} = \tup{c}_{1, m}$
(both projections of initial input $p$ on $I_1$),
the outputs $\tup{d}_{1, \widehat{m}} = \tup{d}_{1, m}$. This shows (\ref{equn:gen-flip-2}).
If $m_1 = m_i$, the inputs are the same, and by definition of $\widehat{m}_1$,
$\tup{d}_{1, \widehat{m}} = \widehat{m}_1(\tup{c}_{1, \widehat{m}}) = \Flip_{\tup{y}, \tup{z}}(m_i(\tup{c}_{1, \widehat{m}}))$ $= \Flip_{\tup{y}, \tup{z}}(m_i(\tup{c}_{1, m})) = \Flip_{\tup{y}, \tup{z}}(\tup{d}_{1, m})$.
Since $\tup{y}, \tup{z}$ only differ in the hidden attributes, by the definition of the $\Flip$ function
$\tup{d}_{1, \widehat{m}} \equiv_{H_i} \tup{d}_{1, m}$.
This shows (\ref{equn:gen-flip-1}). 
Note that the module $m_1$ cannot belong to $D$ since then it will have predecessor $m_i$
and cannot be the first module in topological order.

Suppose the hypothesis holds until $j-1$, consider $m_{j}$. There will be three cases to consider.


\begin{itemize}
  \item[(i)] If $j = i$, for all predecessors $m_k$ of $m_i$ ($O_k \cap I_i \neq \emptyset$), $k \neq i$ and $m_k \notin D$,
  since the workflow is a DAG. Therefore from (\ref{equn:gen-flip-2}), using the induction hypothesis, $\tup{c}_{i, \widehat{m}} = \tup{c}_{i, m}$.
  Hence $\tup{d}_{i, \widehat{m}} = \widehat{m}_i(\tup{c}_{i, \widehat{m}}) = \Flip_{\tup{y}, \tup{z}}(m_i(\tup{c}_{i, \widehat{m}}))$ $= \Flip_{\tup{y}, \tup{z}}(m_i(\tup{c}_{i, m})) = \Flip_{\tup{y}, \tup{z}}(\tup{d}_{i, m})$.
  Again, $\tup{y}, \tup{z}$ are equivalent \wrt $H_i$, so $\tup{d}_{i, \widehat{m}} \equiv_{H_i} \tup{d}_{i, m}$. This shows (\ref{equn:gen-flip-1})
  in the inductive step.

\item [(ii)] If $j \neq i$ ($\widehat{m}_j = m_j$) and $m_j \notin D$, then $m_j$ does not get any of its inputs from any module in $D$,
or any hidden attributes from $m_i$ (then by the definition of $D$, $m_j \in D$). Using IH, from (\ref{equn:gen-flip-2})  
and from (\ref{equn:gen-flip-1}), using the fact that $\tup{y}, \tup{z}$ are equivalent on visible attributes, $\tup{c}_{j, \widehat{m}} = \tup{c}_{j, m}$. 
Then $\tup{d}_{j, \widehat{m}} = m_j(\tup{c}_{j, \widehat{m}}) = m_j(\tup{c}_{j, m}) = \tup{d}_{j, m}$. This shows (\ref{equn:gen-flip-2})
  in the inductive step.

\item[(iii)] If $j \neq i$, but $m_j \in D$, $m_j$ can get all its inputs either from $m_i$, from other modules in $D$,
or from modules not in $D$. Using the IH from (\ref{equn:gen-flip-1}) and (\ref{equn:gen-flip-2}), $\tup{c}_{j, \widehat{m}} \equiv_{H_i} \tup{c}_{j, m}$.
Since $m_j \in D$, by the condition of the lemma, $m_j$ is \DS\ \wrt $H_i$. Therefore the corresponding outputs $\tup{d}_{j, \widehat{m}}  = m_j(\tup{c}_{j, \widehat{m}})$
and $\tup{d}_{j, m}  = m_j(\tup{c}_{j, m})$ are equivalent, or $\tup{d}_{j, \widehat{m}} \equiv_{H_i} \tup{d}_{j, m}$. This again shows (\ref{equn:gen-flip-1})
  in the inductive step.
\end{itemize} 
Hence the IH holds for all $j = 1$ to $n$ and this completes the proof of the lemma. 
\end{proof}

\subsection{Proof of Theorem~\ref{thm:privacy-general}}
Finally, we prove Theorem~\ref{thm:privacy-general}
using Lemmas~\ref{lem:ds-union} and \ref{lem:main-private-general}.
%
%

\begin{proof}[of Theorem~\ref{thm:privacy-general}]
We argue that if $H_i$ satisfies the conditions in 
Theorem~\ref{thm:privacy-general}, 
then $H_i' = \bigcup_{i: m_i \textrm{ is private}}H_i$ satisfies the
conditions in Lemma~\ref{lem:main-private-general}. 
The first two conditions are easily satisfied by $H_i'$: 
(i) $h_i \subseteq H_i \subseteq H_i'$  and (ii)  $h_i \subseteq O_i$.
So we need to show (iii), \ie\ all modules in the downward-closure $D(h_i)$ are \DS\ \wrt\ $A_j \cap H_i'$. 
\par
From the conditions in the theorem, each module $m_j \in D(h_i)$ 
is \DS\ \wrt $A_j \cap H_i$.
We show that for any other private module $m_k \neq m_i$, $m_j$
is also \DS\ \wrt $A_j \cap H_k$. There may be three such cases as discussed below.
\par
\textbf{Case-I:}~~ If $m_j \in D(h_k)$, by the \DS ty conditions in the theorem, $m_j$ is \DS\ \wrt $A_j \cap H_k$.
\par 
\textbf{Case-II:}~~ If $m_j \notin D(h_k)$ and $m_j \neq m_k$, for any private module $m_k \neq m_i$,
$A_j \cap H_k = \emptyset$ (since $H_k \subseteq O_k \cup \bigcup_{\ell \in D(h_k)} A_\ell$
from the theorem). 
From Observation~\ref{obs:ds}, $m_j$ is \DS\ \wrt $A_j \cap H_k$.
\par
\textbf{Case-III:}~~ 
If $m_j \notin D(h_k)$ but $m_j = m_k$ (or $j = k$), then $H_k \cap A_j \subseteq O_j$
(again since $H_k \subseteq O_k \cup \bigcup_{\ell \in D(h_k)} A_\ell$ and $O_k = O_j$). 
From Observation~\ref{obs:ds}, $m_j$ is \DS\ \wrt 
$\emptyset$, and 
$A_j \cap H_k 
\supseteq \emptyset$. Further, $I_j \setminus \emptyset = I_j = I_j \setminus (A_j \cap H_k)$.
This is because $H_k \cap A_j \subseteq O_j$, since $O_j \cap I_j = \emptyset$, $I_j \cap (A_j \cap H_k) = \emptyset$.
Hence from the same observation, $m_j$ is \DS\ \wrt $A_j \cap H_k$.
\par 
Hence $m_j$ is \DS\ \wrt $A_j \cap H_i$ and for all private modules $m_k$, $m_k \neq m_i$, $m_j$ is \DS\ \wrt $A_j \cap H_k$.
By Lemma~\ref{lem:ds-union}, then $m_j$ is \DS\ \wrt $(A_j \cap H_i) \cup (A_j \cap H_k)$
$= A_j \cap (H_i \cup H_k)$. 
By a simple induction on all private modules 
$m_k$, $m_j$ is \DS\ \wrt $A_j \cap (\bigcup_{k: m_k \textrm{ is private}}) H_k$ = $A_j \cap H_i'$.
Hence $H_i'$ satisfies the conditions stated in the lemma. The rest of the proof follows by the same argument as in the proof of 
Theorem~\ref{thm:privacy-downward}.
\eat{
						\par
						Now the proof can be completed using exactly the same argument as in Theorem~\ref{thm:privacy-downward}.
						Theorem~\ref{thm:privacy-general} also states that
						each private module $m_i$ is $\Gamma$-standalone-private
						\wrt hidden attributes $h_i$, i.e., $|\Out_{\tup{x}, m_i, h_i}| \geq \Gamma$ for all input
						$\tup{x}$ to $m_i$, for all private modules $m_i$ 
						(see
						Definition~\ref{def:standalone-privacy}). From
						Lemma~\ref{lem:main-private-single-pred}, using $H_i'$ in place of $H_i$, this implies that for all input
						$\tup{x}$ to private modules $m_i$, $|\Out_{\tup{x}, W, H_i}| \geq \Gamma$ where 
						From Definition~\ref{def:workflow-privacy}, this implies that each $m_i \in M^-$ is $\Gamma$-workflow-private
						\wrt $V = A \setminus \bigcup_{i \in M^-} H_{i}$; equivalently $W$ is $\Gamma$-private \wrt $V$. 
}
\end{proof}


%




\end{document}